\documentclass[a4paper]{article}
\usepackage{color}
\usepackage{amsmath}
\usepackage{amsthm}
\usepackage{amssymb}
\usepackage{amsmath}
\usepackage{mathcomp}
\usepackage{dsfont}
\usepackage{comment}
\usepackage{authblk}

\usepackage[active]{srcltx}
\usepackage{stackengine}
\usepackage[bookmarksnumbered=true]{hyperref}

\usepackage{geometry}
\newtheorem{definition}{Definition}[section]
\newtheorem{lemma}[definition]{Lemma}
\newtheorem{proposition}[definition]{Proposition}
\newtheorem{theorem}[definition]{Theorem}
\newtheorem{remark}[definition]{Remark}
\newtheorem{corollary}[definition]{Corollary}

\numberwithin{equation}{section}

\def\tr{\mathrm{tr}}

\def\vr{v^{\text{r}}}
\def\ur{u^{\text{r}}}
\def\omegar{\omega^{\text{r}}}

\begin{document}

\title{The dilute Fermi gas via Bogoliubov theory}

\author[1]{Marco Falconi}
\affil[1]{University of Roma Tre, Department of Mathematics and Physics, L.go S. L. Murialdo 1, 00146 Roma, Italy}
\author[2]{Emanuela L. Giacomelli}
\affil[2]{LMU M\"unich, Department of Mathematics, Theresienstr. 39, 80333 M\"unchen, Germany}
\author[2]{Christian Hainzl}
\author[3]{Marcello Porta}
\affil[3]{SISSA, Mathematics Area, Via Bonomea 265, 34136 Trieste, Italy}

\maketitle

\abstract{We study the ground state properties of interacting Fermi gases in the dilute regime, in three dimensions. We compute the ground state energy of the system, for positive interaction potentials. We recover a well-known expression for the ground state energy at second order in the particle density, which depends on the interaction potential only via its scattering length. The first proof of this result has been given by Lieb, Seiringer and Solovej in \cite{LSS}. In this paper we give a new derivation of this formula, using a different method; it is inspired by Bogoliubov theory, and it makes use of the almost-bosonic nature of the low-energy excitations of the systems. With respect to previous work, our result applies to a more regular class of interaction potentials, but it comes with improved error estimates on the ground state energy asymptotics in the density.}

\tableofcontents

\section{Introduction}

In this paper we consider interacting, spin $1/2$ fermions, in three dimensions, in the thermodynamic limit. We will focus on the ground state energy of the system, for positive and short-ranged interaction potential. Let $\rho_{\sigma}$ be the density of particles with spin up, $\sigma =\, \uparrow$, or spin down, $\sigma =\, \downarrow$. Let $e(\rho_{\uparrow}, \rho_{\downarrow})$ be the ground state energy density of the system, in the termodynamic limit. We will be interested in the dilute regime, corresponding to $\rho_{\sigma} \ll 1$. It is well-known that, in units such that $\hbar = 1$ and setting the masses of the particles to be equal to $1/2$:
\begin{equation}\label{eq:LSS}
e(\rho_{\uparrow}, \rho_{\downarrow}) = \frac{3}{5}(6\pi^{2})^{\frac{2}{3}} (\rho_{\uparrow}^{\frac{5}{3}} + \rho_{\downarrow}^{\frac{5}{3}}) + 8\pi a \rho_{\uparrow} \rho_{\downarrow} + o(\rho^{2})\;.
\end{equation}
The first term in the right-hand side of Eq. (\ref{eq:LSS}) is purely kinetic, and its $\rho^{5/3}$-dependence is a consequence of the fermionic nature of the wave function. It is easy to find a fermionic state that reproduces the correct $\rho^{5/3}$ dependence of the energy; this is the free Fermi gas, {\it i.e.} the fermionic state that minimizes the total kinetic energy of the systems, in a way compatible with Pauli principle.

The effect of the interaction is visible at the next order; denoting by $V$ the two-body potential, the constant $a$ in Eq. (\ref{eq:LSS}) is the {\it scattering length} of the potential. For small potentials, it can be computed as a perturbative expansion in $V$, via the Born series. It is easy to check that taking the free Fermi gas as a trial state for the many-body problem, the $\rho^{2}$-dependence of the ground state energy is off by an order $1$ multiplicative constant: instead of $8\pi a$ one finds $\hat V(0)$, which is strictly larger than $8\pi a$. To reproduce the correct dependence of the energy in the interaction, one has to understand the effect of {\it correlations} in the fermionic ground state, which is not an easy task even from the point of view of an upper bound. 

The first proof of (\ref{eq:LSS}) has been given by Lieb, Seiringer and Solovej in an important work \cite{LSS}. The proof of \cite{LSS} covers a large class of positive two-body potentials, including the case of hard spheres. The result of \cite{LSS} has then been extended by Seiringer to the computation of the thermodynamic pressure for positive temperature Fermi gases \cite{Se}. Concerning interacting lattice fermions (Hubbard model), the analogue of Eq. (\ref{eq:LSS}) follows from the upper bound of Giuliani \cite{G} and from the lower bound of Seiringer and Yin \cite{SY}.

For bosonic systems, in a seminal paper \cite{LY} Lieb and Yngvason proved that the ground state energy density of the dilute Bose gas is, assuming the particles to be spinless:
\begin{equation}\label{eq:LY}
e(\rho) = 4\pi a \rho^{2} + o(\rho^{2})\;.
\end{equation}
In this expression, the interaction determines the ground state energy at leading order, in contrast to (\ref{eq:LSS}). This is consistent with the fact that bosons tend to minimize the energy occupying the lowest momentum state, which gives no contribution to the ground state energy. This is of course forbidden for fermions, due to Pauli principle. The result of \cite{LY} has been  recently improved by Fournais and Solovej in \cite{FS}. The work \cite{FS} obtained a more refined asymptotics for the ground state energy density, from the point of view of a lower bound. Combined with the upper bound of Yau and Yin in \cite{YY}, the work \cite{FS} determined the next order correction to the ground state energy of dilute bosons, and put on rigorous grounds the celebrated Lee-Huang-Yang formula.

Comparing Eq. (\ref{eq:LSS}) with Eq. (\ref{eq:LY}), one is naturally tempted to think the low energy excitations around the free Fermi gas as pairs of fermions, which can be described by emergent bosonic particles, whose ground state energy reconstructs the $8\pi a \rho_{\uparrow} \rho_{\downarrow}$ term in (\ref{eq:LSS}) (the extra factor $2$ is due to the spin degrees of freedom). The main motivation of the present paper is to make this intuition mathematically precise.
 
For the bosonic problem, a natural trial state that captures the correct dependence of the ground state energy on the scattering length is provided by a suitable unitary transformation, a {\it Bogoliubov rotation}, of a coherent state, \cite{ESYtrial}. Interestingly, for small potentials, the energy of this trial state also reproduces the Lee-Huang-Yang formula for the next order correction to the energy \cite{ESYtrial, NRS1, NRS2}, up to higher order terms in the interaction. In this paper we introduce the fermionic analogue of such transformations, roughly by considering pairs of fermions as effective bosons. The main difficulty we have to face is that, in the language of second quantization, quadratic expressions in the bosonic creation and annihilation operators become {\it quartic} in terms of the fermionic operators. As a consequence, the nice algebraic properties of Bogoliubov transformations are only {\it approximately} true, in the fermionic setting; quantifying the validity of this approximation is a nontrivial task, and it is the main technical challenge faced in the present paper.

The main application of our method is a new proof of (\ref{eq:LSS}). Our result comes with a substantial improvement of the error estimate. However, it is restricted to more regular interaction potentials with respect to \cite{LSS}. In particular, the result of \cite{LSS} includes the case of hard spheres, which we cannot cover at the moment. We believe that a larger class of interactions could be treated by approximation arguments, but we have not tried to extend the result in this direction. Nevertheless, we think that our approach is conceptually simple, and that it gives a new point of view on dilute Fermi gases.

Our method borrows ideas from a series of recent, groundbreaking works of Boccato, Brennecke, Cenatiempo and Schlein \cite{BBCS1, BBCS2, BBCS3, BBCS, BS}. There, Bogoliubov theory for interacting Bose gases in the Gross-Pitaevskii regime has been put on rigorous grounds, and it has been used to obtain sharp asymptotics on the ground state energy and on the excitation spectrum. Concerning the energy asymptotics of interacting fermions in the mean-field regime, the first rigorous result about the correlation energy, defined as the difference between the many-body and Hartree-Fock ground state energies, has been obtained in \cite{HPR}. In \cite{HPR}, the correlation energy has been rigorously computed for small potentials via upper and lower bounds, that agree at second order in the interaction. The proof is based on rigorous second order perturbation theory, first developed in \cite{HQED, HSei}. The method that we introduce in the present paper is related to the bosonization approach of \cite{BNPSS, BNPSS2}. The method of \cite{BNPSS, BNPSS2} allowed to compute the correlation energy of weakly interacting, mean-field fermionic systems, at all orders in the interaction strength. The result of \cite{BNPSS, BNPSS2} confirmed the prediction of the random phase approximation, see \cite{Bpro} for a review. Both \cite{HPR, BNPSS, BNPSS2} make use of Fock space methods and fermionic Bogoliubov transformations, extending ideas previously introduced in the context of many-body fermionic dynamics \cite{BPS, BPS2, BJPSS, PRSS, BPSbook}.

A key technical ingredient of \cite{BNPSS, BNPSS2} is the localization of the low energy excitations around the Fermi surface in terms of suitable patches, where the quasi-particle dispersion relation can be approximated by a linear one. This is not needed in the dilute regime considered here, due to the fact that the Fermi momentum is much smaller than the typical momentum exchanged in the two-body scattering. Another difference with respect to \cite{BNPSS, BNPSS2} is that here we consider interacting systems in the thermodynamic limit; controlling this limit is nontrivial, due to the slow decay of the correlations for the free Fermi gas, which plays the role of reference state in our analysis, and of the solution of the scattering equation. Despite the mean-field regime and the dilute regime are somewhat opposite, we find it remarkable that similar bosonization ideas apply in both cases.

As a future perspective, we think that the method presented in this paper might provide a good starting point for the derivation of more refined energy asymptotics, by importing tools that have been developed in the last decade for interacting bosons. An outstanding open problem is to prove the fermionic analogue of the Lee-Huang-Yang formula, due to Huang and Yang in \cite{HY}, which gives the next order correction to the ground state energy of dilute fermions, of order $\rho^{7/3}$ (in three dimensions).

The paper is organized as follows. In Section \ref{sec:res} we define the model and state our main result, Theorem \ref{thm:main}. In Section \ref{sec:fock} we formulate the problem in Fock space, and we introduce fermionic Bogoliubov transformations, which will allow to extract the $\rho^{5/3}$ dependence of the ground state energy, and part of the $\rho^{2}$ dependence. In Section \ref{sec:T} we define the fermionic analogue of the bosonic Bogoliubov transformation, called {\it correlation structure}, that will allow us to compute the ground state energy at order $\rho^{2}$; see Section \ref{sec:heu} for a heuristic discussion. Section \ref{sec:prop} is the main technical section of the paper; here we discuss the properties of the correlation structure, that mimic the algebraic properties of bosonic Bogoliubov transformations at leading order in the density. In Section \ref{sec:lwbd} we use the discussion of Section \ref{sec:prop} to prove a lower bound for the ground state energy, that displays the correct dependence of the scattering length at order $\rho^{2}$. Then, in Section \ref{sec:upper} we conclude the proof of Theorem \ref{thm:main} by proving a matching upper bound, by the choice of a suitable trial state. Finally, in Appendix \ref{sec:scat} we collect properties of the solution of the scattering equation, that we shall use in our proofs; in Appendix \ref{app:lemphi} we prove some technical estimates for almost-bosonic operators; and in Appendix \ref{app:UV} we collect technical estimates on the infrared and ultraviolet regularizations of various expressions appearing in our proofs.

\section{Main result}\label{sec:res}
We consider a system of $N$ interacting, spinning fermions in a cubic box $\Lambda_{L} = [0,L]^{3}$, with periodic boundary conditions. The Hamiltonian of the model acts on $L^{2}(\Lambda_{L}; \mathbb{C}^{2})^{\otimes N}$, and it is given by:
\begin{equation}
H_{N} = -\sum_{i=1}^{N} \Delta_{x_{i}} + \sum_{i<j = 1}^{N} V(x_{i} - x_{j})\;,
\end{equation}
with $\Delta_{x_{i}}$ the Laplacian acting on the $i$-th particle, and $V$ the pair interaction potential. We shall suppose that $V$ is the `periodization' on $\Lambda_{L}$ of a potential $V_{\infty}$ on $\mathbb{R}^{3}$, compactly supported and regular enough:
\begin{equation}
V(x-y) = \frac{1}{L^{3}} \sum_{p\in \frac{2\pi}{L}\mathbb{Z}^{3}} e^{ip\cdot (x - y)} \hat V_{\infty}(p)\;,
\end{equation}
with $\hat V_{\infty}(p) = \int_{\mathbb{R}^{3}} dx\, e^{-ip\cdot x} V_{\infty}(x)$. We shall denote by $N_{\sigma}$ the number of particles with a given spin $\sigma \in \{ \uparrow, \downarrow \}$, and we shall set $N = N_{\uparrow} + N_{\downarrow}$. We shall require the wave function of the system, on which the Hamiltonian acts, to be antisymmetric separately in the first $N_{\uparrow}$ variables, and in the second $N_{\downarrow}$ variables. That is, the space of allowed wave functions is $\frak{h}(N_{\uparrow}, N_{\downarrow}) := L^{2}_{\text{a}}(\Lambda_{L}^{N_{\uparrow}}) \otimes L^{2}_{\text{a}}(\Lambda_{L}^{N_{\downarrow}})$,  with $L^{2}_{\text{a}}(\Lambda_{L}^{N_{\sigma}}) = L^{2}(\Lambda_{L})^{\wedge N_{\sigma}}$ the antisymmetric sector of $L^{2}(\Lambda_{L})^{\otimes N_{\sigma}}$.  We will focus on the ground state energy of the system:
\begin{equation}
E_{L}(N_{\uparrow}, N_{\downarrow}) = \inf_{\Psi \in \frak{h}(N_{\uparrow}, N_{\downarrow})} \frac{\langle \Psi, H_{N} \Psi\rangle}{\langle \Psi, \Psi \rangle}\;.
\end{equation}
By translation invariance of the Hamiltonian, the energy is extensive in the system size. Thus, let us define the ground state energy density as:
\begin{equation}
e_{L}(\rho_{\uparrow}, \rho_{\downarrow}) := \frac{E_{L}(N_{\uparrow}, N_{\downarrow})}{L^{3}}\;.
\end{equation}
Let $\rho_{\sigma} = N_{\sigma} / L^{3}$ be the density of particles with spin $\sigma$, and let $\rho = \rho_{\uparrow} + \rho_{\downarrow}$ be the total density. We shall be interested in the thermodynamic limit, meaning $N_{\sigma}, L\to \infty$, with $\rho_{\sigma}$ fixed. The existence of the limit, and the independence of the limit from the choice of the boundary conditions, is well-known \cite{R, Ro}. We shall focus on the dilute regime, corresponding to $\rho \ll 1$. The next theorem is our main result.
\begin{theorem}\label{thm:main} Let $V\in L^{1}(\Lambda_{L})$, compactly supported, $V\geq 0$. There exists $L_{0}>0$ large enough such that for $L\geq L_{0}$ the following holds. There exists $k_{0} > 0$ such that, for $k\geq k_{0}$ and $V\in C^{k}(\Lambda_{L})$:
\begin{equation}\label{eq:main}
e_{L}(\rho_{\uparrow}, \rho_{\downarrow}) = \frac{3}{5}(6\pi^{2})^{\frac{2}{3}} (\rho_{\uparrow}^{\frac{5}{3}} + \rho_{\downarrow}^{\frac{5}{3}}) + 8\pi a \rho_{\uparrow} \rho_{\downarrow} + r_{L}(\rho_{\uparrow}, \rho_{\downarrow})\;,
\end{equation}
where $a$ is the scattering length of the potential $V$, and for some constant $C$ only dependent on $V$:
\begin{equation}\label{eq:xi1xi2}
-C\rho^{2 + \xi_{2}} \leq r_{L}(\rho_{\uparrow}, \rho_{\downarrow}) \leq C\rho^{2 + \xi_{1}}
\end{equation}
with $\xi_{1} = \frac{2}{9}$ and $\xi_{2} = \frac{1}{9}$.
\end{theorem}
\begin{remark}
\begin{itemize}
\item[(i)] The result is not new. As discussed in the introduction, the first proof of (\ref{eq:main}) has been given by Lieb, Seiringer, Solovej in \cite{LSS}. The extension to positive temperature has been obtained in \cite{Se}. The analogue of (\ref{eq:main}) for the Hubbard model follows from the combination of the upper bound of \cite{G} and the lower bound of \cite{SY}. 

\item[(ii)] With respect to \cite{LSS}, our result comes with improved error estimates; in \cite{LSS}, $\xi_{1} = 2/27$ and $\xi_{2} = 1/39$.

\item[(iii)] With respect to \cite{LSS}, our method is restricted to more regular interaction potentials. This restriction is technical: the regularity of the potential will be used to prove the smallness of various error terms involving particle excitations with quasi-momenta greater than $\rho^{-\beta}$, for some $\beta > 0$. The work \cite{LSS} also covers the case of hard spheres, which we cannot consider at the moment. We believe that a larger class of potentials could be treated via approximation arguments, but we have not tried to improve the result is this direction. 

\item[(iv)] We think that our method gives a new perspective on dilute Fermi gases, that might allow to import the ideas developed for Bose gases in the last years. This could be useful for the computation of higher order corrections to the ground state energy, or for the determination of the excitation spectrum. As mentioned in the introduction, an outstanding open problem is to prove the fermionic analogue of the Lee-Huang-Yang formula, due to Huang and Yang \cite{HY}, which predicts the next order correction to the ground state energy density, of order $\rho^{7/3}$ (corresponding to $\xi_{1} = \xi_{2} = \frac{1}{3}$ in (\ref{eq:xi1xi2})).
\end{itemize}
\end{remark}
The rest of the paper is devoted to the proof of Theorem \ref{thm:main}.
\section{Second quantization}\label{sec:fock}
\subsection{Fock space representation}
In the following, it will be convenient to work in a setting in which the number of particles is not fixed. To this end, we define the fermionic Fock space as:
\begin{equation}
\mathcal{F} = \bigoplus_{n\geq 0} \mathcal{F}^{(n)}\;,\qquad \mathcal{F}^{(n)} = L^{2}(\Lambda_{L}; \mathbb{C}^{2})^{\wedge n}
\end{equation}
with the understanding that $\mathcal{F}^{(0)} = \mathbb{C}$. Thus, a given element $\psi \in \mathcal{F}$ is an infinite sequence of fermionic wave functions, $\psi = (\psi^{(0)}, \psi^{(1)}, \ldots, \psi^{(n)}, \ldots)$ with $\psi^{(n)} \in \mathcal{F}^{(n)}$, $\psi^{(n)} \equiv \psi^{(n)}((x_{1}, \sigma_{1}), \ldots, (x_{n}, \sigma_{n}))$ and $(x,\sigma) \in \Lambda_{L} \times \{\uparrow, \downarrow\}$. An important example of vector in the Fock space is the vacuum vector $\Omega$, describing the zero particle state:
\begin{equation}
\Omega = (1, 0,0, \ldots, 0,\ldots )\;.
\end{equation}
Next, it is convenient to introduce the fermionic creation/annihilation operators, as follows. For $f\in L^{2}(\Lambda_{L}; \mathbb{C}^{2}) \simeq L^{2}(\Lambda_{L}) \oplus L^{2}(\Lambda_{L})$, $f = (f_{\uparrow}, f_{\downarrow})$, the fermionic annihilation operator $a(f) : \mathcal{F}^{(n)} \to \mathcal{F}^{(n-1)}$ and creation operator $a^{*}(f): \mathcal{F}^{(n)} \to \mathcal{F}^{(n+1)}$ are defined as:
\begin{eqnarray}\label{eq:aastar}
&&(a(f) \psi)^{(n)}((x_{1}, \sigma_{1}), \ldots, (x_{n}, \sigma_{n})) = \sqrt{n+1} \sum_{\sigma = \uparrow\downarrow} \int_{\Lambda_{L}} dx\, \overline{f_{\sigma}(x)} \psi^{(n+1)}((x, \sigma), (x_{1}, \sigma_{1}), \ldots, (x_{n}, \sigma_{n})) \nonumber\\
&&(a^{*}(f) \psi)^{(n)}((x_{1}, \sigma_{1}), \ldots, (x_{n}, \sigma_{n})) \nonumber\\&& = \frac{1}{\sqrt{n}} \sum_{j=1}^{n} (-1)^{j+1} f_{\sigma_{j}}(x_{j}) \psi^{(n-1)}((x_{1}, \sigma_{1}), \ldots, (x_{j-1}, \sigma_{j-1}), (x_{j+1}, \sigma_{j+1}), \ldots, (x_{n}, \sigma_{n}))\;.
\end{eqnarray} 
The above definitions are complemented by the requirement that the operator $a(f)$ annihilates the Fock space vacuum, $a(f)\Omega = 0$. The definition (\ref{eq:aastar}) implies that $a(f)^{*} = a^{*}(f)$, and that:
\begin{equation}\label{eq:CAR}
\{ a(f), a(g) \} = \{ a^{*}(f), a^{*}(g) \} = 0\;,\qquad \{ a(f), a^{*}(g) \} = \langle f, g \rangle_{L^{2}(\Lambda_{L}; \mathbb{C}^{2})}
\end{equation}
where $\langle f, g \rangle_{L^{2}(\Lambda_{L}; \mathbb{C}^{2})} = \sum_{\sigma = \uparrow\downarrow} \int_{\Lambda_{L}} dx\, \overline{f_{\sigma}(x)} g_{\sigma}(x)$. As a consequence of the canonical anticommutation relations (\ref{eq:CAR}) we have:
\begin{equation}\label{eq:bdferm}
\| a(f) \| \leq \|f\|_{L^{2}(\Lambda_{L}; \mathbb{C}^{2})}\;,\qquad \|a^{*}(g)\| \leq \|g\|_{L^{2}(\Lambda_{L};\mathbb{C}^{2})}\;.
\end{equation}
It will also be convenient to represent the creation/annihilation operators in terms of the operator-valued distributions $a^{*}_{x,\sigma}$, $a_{x,\sigma}$,
\begin{equation}
a(f) = \sum_{\sigma = \uparrow\downarrow} \int_{\Lambda_{L}} dx\, a_{x,\sigma} \overline{f_{\sigma}(x)}\;,\qquad a^{*}(g) = \sum_{\sigma = \uparrow\downarrow} \int_{\Lambda_{L}} dx\, a^{*}_{x,\sigma} g_{\sigma}(x)\;,
\end{equation}
where, formally, $a_{x,\sigma} = a(\delta_{x,\sigma})$. We used the notation $\delta_{x,\sigma}(y, \sigma') = \delta_{\sigma, \sigma'} \delta(x-y)$, with $\delta_{\sigma, \sigma'}$ the Kronecker delta, and $\delta(x-y)$ the Dirac delta distribution, periodic over $\Lambda_{L}$:
\begin{equation}\label{eq:perDir}
\delta(x-y) = \frac{1}{L^{3}} \sum_{k\in \frac{2\pi}{L}\mathbb{Z}^{3}} e^{ik\cdot (x-y)}\;.
\end{equation}
It will also be convenient to introduce momentum-space fermionic creation and annihilation operators. Let $f_{k}(x) = L^{-3/2} e^{ik\cdot x}$, for $k\in (2\pi / L) \mathbb{Z}^{3}$. Then:
\begin{equation}\label{eq:Fou}
\hat a_{k,\sigma} \equiv a_{\sigma}(f_{k}) = \frac{1}{L^{\frac{3}{2}}}\int_{\Lambda_{L}} dx\, a_{x,\sigma} e^{-ik\cdot x}\;,\qquad \hat a^{*}_{k,\sigma} = a_{\sigma}(f_{k})^{*}\;.
\end{equation}
These relations can be inverted as follows, for all $x\in \Lambda_{L}$:
\begin{equation}
a_{x,\sigma} = \frac{1}{L^{\frac{3}{2}}} \sum_{k\in \frac{2\pi}{L}\mathbb{Z}^{3}} e^{ik\cdot x} \hat a_{k,\sigma}\;.
\end{equation}
We then define the number operator $\mathcal{N}$ as:
\begin{equation}\label{eq:calN}
\mathcal{N} = \sum_{\sigma = \uparrow \downarrow}\int_{\Lambda_{L}} dx\, a^{*}_{x,\sigma} a_{x,\sigma} \equiv \sum_{\sigma = \uparrow\downarrow} \sum_{k\in \frac{2\pi}{L}\mathbb{Z}^{3}} \hat a^{*}_{k,\sigma} \hat a_{k,\sigma}\;.
\end{equation}
The operator $\mathcal{N}$ counts the number of particles in a given sector of the fermionic Fock space, $(\mathcal{N} \psi)^{(n)} = n\psi^{(n)}$. We shall also define the number operator associated to particles with a given spin $\sigma$ as:
\begin{equation}
\mathcal{N}_{\sigma} = \int_{\Lambda_{L}} dx\, a^{*}_{x,\sigma} a_{x,\sigma}\;.
\end{equation}
We shall say that $\psi\in \mathcal{F}$ is an $N$-particle state if $\mathcal{N} \psi = N \psi$. Also, we shall say that an $N$-particle state $\psi$, with $N = N_{\uparrow} + N_{\downarrow}$, has $N_{\sigma}$ particles with spin $\sigma$ if $\mathcal{N}_{\sigma} \psi = N_{\sigma}\psi$. We shall denote by $\mathcal{F}^{(N_{\uparrow}, N_{\downarrow})} \subset \mathcal{F}$ the set of such states. 
Let us now rewrite the ground state energy of the system in the language of second quantization. We define the second-quantized Hamiltonian as:
\begin{equation}
\mathcal{H} = \sum_{\sigma = \uparrow\downarrow} \int_{\Lambda_{L}} dx\, \nabla_{x} a^{*}_{x,\sigma} \nabla_{x} a_{x,\sigma} + \frac{1}{2}\sum_{\sigma, \sigma' = \uparrow\downarrow} \int_{\Lambda_{L} \times \Lambda_{L}} dxdy\, V(x-y) a^{*}_{x,\sigma} a^{*}_{y,\sigma'} a_{y,\sigma'} a_{x,\sigma}\;.
\end{equation}
It is not difficult to check that $(\mathcal{H} \psi)^{(n)} = H_{n} \psi^{(n)}$, for $n\geq 1$. By the spin-independence of the $n$-particle Hamiltonian $H_{n}$, we have:
\begin{equation}\label{eq:gstate}
E_{L}(N_{\uparrow}, N_{\downarrow}) = \inf_{\psi \in \mathcal{F}^{(N_{\uparrow}, N_{\downarrow})}} \frac{\langle \psi, \mathcal{H} \psi \rangle}{\langle \psi, \psi \rangle}\;.
\end{equation}
Eq. (\ref{eq:gstate}) is a convenient starting point for our analysis.
\subsection{Fermionic Bogoliubov transformations}
\subsubsection{The free Fermi gas}\label{sec:FFG}
A simple upper bound for the ground state energy is obtained taking as a trial state the Slater determinant that minimizes the kinetic energy, for given $N_{\uparrow}$, $N_{\downarrow}$. We shall refer to this state as the {\it free Fermi gas} (FFG). Explicitly,
\begin{equation}\label{eq:psiFFG}
\Psi_{\text{FFG}}\big( (x_{i}, \uparrow\}_{i=1}^{N_{\uparrow}}, \{ y_{j}, \downarrow \}_{j=1}^{N_{\downarrow}} \big) = \frac{1}{\sqrt{N_{\uparrow}!}} \frac{1}{\sqrt{N_{\downarrow}!}} (\det f^{\uparrow}_{k_{i}}(x_{j}))_{1\leq i, j \leq N_{\uparrow}} (\det f^{\downarrow}_{k_{i}}(y_{j}))_{1 \leq i, j\leq N_{\downarrow}}\;,
\end{equation}
where $f^{\sigma}_{k}(x) \equiv f_{k}(x)$, with $f_{k}(x) = L^{-\frac{3}{2}} e^{ik\cdot x}$ and $k\in \mathcal{B}_{F}^{\sigma}$, with $\mathcal{B}_{F}^{\sigma}$ the {\it Fermi ball:}
\begin{equation}
\mathcal{B}_{F}^{\sigma} = \Big\{ k\in \frac{2\pi}{L} \mathbb{Z}^{3} \, \Big|\, |k| \leq k_{F}^{\sigma} \Big\}\;,
\end{equation}
where the Fermi momentum $k_{F}^{\sigma}$ is chosen so that $N_{\sigma} = |\mathcal{B}_{F}^{\sigma}|$. Of course, this is not possible for all values of $N_{\sigma}$. We shall only consider values of $N_{\sigma}$ such that the Fermi ball is completely filled; {\it i.e.}, for which there exists $k_{F}^{\sigma}$ so that $N_{\sigma} = |\mathcal{B}_{F}^{\sigma}|$. This is not a loss of generality, by the existence of the thermodynamic limit, and since the densities $\rho_{\sigma} = N_{\sigma} / L^{3}$ obtained in this way are dense in $\mathbb{R}_{+}$ as $L\to \infty$. Notice that, for fixed density $\rho_{\sigma}$, $k_{F}^{\sigma} = (6\pi^{2})^{1/3} \rho_{\sigma}^{1/3}+ o(1)$ as $L\to \infty$.

No repetition occurs in the momenta involved in the definition of each determinant in the right-hand side of (\ref{eq:psiFFG}); otherwise, the wave function would be exactly zero, by antisymmetry (Pauli principle). The state (\ref{eq:psiFFG}) turns out to be equal to the fermionic ground state of the total kinetic energy operator $-\sum_{j=1}^{N}\Delta_{x_j}$, on $\frak{h}(N_{\uparrow}, N_{\downarrow})$. The total kinetic energy density of such state is:
\begin{eqnarray}\label{eq:kinetic}
\frac{\langle \Psi_{\text{FFG}}, -\sum_{j=1}^{N} \Delta_{x_{j}} \Psi_{\text{FFG}} \rangle}{L^{3}} &=& \frac{1}{L^{3}} \sum_{\sigma} \sum_{k\in \mathcal{B}_{F}^{\sigma}} |k|^{2} \nonumber\\ 
&=& \frac{3}{5}(6\pi^{2})^{\frac{2}{3}} (\rho_{\uparrow}^{\frac{5}{3}} + \rho_{\downarrow}^{\frac{5}{3}}) + O(L^{-1})\;,
\end{eqnarray}
where the error term denotes contribution bounded as $CL^{-1}$ for $L$ large enough (it is the error term arising from replacing the Riemann sum in the first line by an integral). The same state allows to obtain a simple upper bound for the ground state energy of the interacting system. It is convenient to introduce the fully antisymmetrized version of the $\Psi_{\text{FFG}}$, as:
\begin{equation}\label{eq:AFFS}
\Phi_{\text{FFG}}(x_{1}, \ldots, x_{N}) = \frac{1}{\sqrt{N!}} \det (f^{\sigma_{i}}_{i}(x_{j}))_{1\leq i,j\leq N}\;,
\end{equation}
with the understanding that $\sigma_{i} = \uparrow$ for $i\in [1, N_{\uparrow}]$ and $\sigma_{i} = \downarrow$ for $i\in [N_{\uparrow} + 1, N]$. The orbitals satisfy the orthogonality condition $\langle f^{\sigma}_{k}, f^{\sigma'}_{k'}\rangle = \delta_{\sigma,\sigma'} \delta_{k,k'}$. In the Fock space language, the state (\ref{eq:AFFS}) can also be represented as (up to an overall sign):
\begin{equation}
\Phi_{\text{FFG}} = \prod_{\sigma = \uparrow,\downarrow} \prod_{k\in \mathcal{B}_{F}^{\sigma}} \hat a^{*}_{k,\sigma}\Omega\;.
\end{equation}
Being the Hamiltonian spin independent:
\begin{equation}
\langle \Phi_{\text{FFG}}, H_{N} \Phi_{\text{FFG}}\rangle = \langle \Psi_{\text{FFG}}, H_{N} \Psi_{\text{FFG}}\rangle\;.
\end{equation}
Eq. (\ref{eq:AFFS}) is an example of quasi-free state, for which all correlation functions can be computed starting from the reduced one-particle density matrix:
\begin{eqnarray}\label{eq:omegazz}
\omega_{\sigma, \sigma'}(x,x') &:=& \langle \Phi_{\text{FFG}}, a^{*}_{x',\sigma'} a_{x,\sigma}\Phi_{\text{FFG}}\rangle\nonumber\\
&=& \delta_{\sigma, \sigma'} \sum_{k\in \mathcal{B}^{\sigma}_{F}} \frac{1}{L^{3}}\, e^{ik\cdot (x-x')}\;.
\end{eqnarray}
Eq. (\ref{eq:omegazz}) defines the integral kernel of an operator $\omega: L^{2}(\Lambda_{L}; \mathbb{C}^{2}) \to L^{2}(\Lambda_{L}; \mathbb{C}^{2})$, such that $\omega = \omega^{2} = \omega^{*}$, $\tr\, \omega = N$. In Fourier space, $\hat \omega_{\sigma, \sigma}(k)$ is the characteristic function of the Fermi ball $\mathcal{B}_{F}^{\sigma}$. All higher order density matrices of the system can be computed starting from $\omega$, via the fermionic Wick rule. In particular, the energy of $\Phi_{\text{FFG}}$ only depends on $\omega$:
\begin{equation}
\langle \Phi_{\text{FFG}}, H_{N} \Phi_{\text{FFG}}\rangle = E_{\text{HF}}(\omega)
\end{equation}
where $E_{\text{HF}}(\omega)$ is the Hartree-Fock energy functional:
\begin{equation}
E_{\text{HF}}(\omega) = -\tr\, \Delta \omega + \frac{1}{2} \sum_{\sigma, \sigma' = \uparrow\downarrow}\int_{\Lambda_{L} \times \Lambda_{L}} dxdy\, V(x-y) ( \omega_{\sigma, \sigma}(x;x) \omega_{\sigma', \sigma'}(y;y) - |\omega_{\sigma, \sigma'}(x;y)|^{2})\;.
\end{equation}
The first term reproduces the kinetic energy of the free Fermi gas, Eq. (\ref{eq:kinetic}). The second term, called the direct term, only depends on the density of the system, $\omega_{\sigma, \sigma}(x;x) = \rho_{\sigma}$:
\begin{equation}\label{eq:direct}
\frac{1}{2} \sum_{\sigma, \sigma' = \uparrow, \downarrow}\int_{\Lambda_{L} \times \Lambda_{L}} dxdy\, V(x-y) \omega_{\sigma, \sigma}(x;x) \omega_{\sigma', \sigma'}(y;y) = \frac{L^{3}}{2} \sum_{\sigma, \sigma'} \hat V(0) \rho_{\sigma} \rho_{\sigma'}\;.
\end{equation}
The last term, called the exchange term, can be computed at leading order in the density:
\begin{eqnarray}\label{eq:exchange}
- \frac{1}{2}\sum_{\sigma, \sigma' = \uparrow, \downarrow} \int_{\Lambda_{L} \times \Lambda_{L}} dxdy\, V(x-y) |\omega_{\sigma, \sigma'}(x;y)| &=& - \frac{1}{2}\sum_{\sigma, \sigma'}  \frac{\delta_{\sigma,\sigma'}}{L^{3}} \sum_{k,k' \in \mathcal{B}^{\sigma}_{F}} \hat V (k - k') \nonumber\\
&=& - \frac{1}{2}\sum_{\sigma} L^{3} \hat V(0) \rho^{2}_{\sigma} + O(\rho^{7/3})\;,
\end{eqnarray}
where we used that if $k,k'\in \mathcal{B}_{F}^{\sigma}$ then $|k - k'|\leq C\rho^{1/3}$, which implies $\hat V(k-k') = \hat V(0) + O(\rho^{1/3})$. By the variational principle we get, putting (\ref{eq:kinetic}), (\ref{eq:direct}), (\ref{eq:exchange}) together, for $L$ large enough:
\begin{eqnarray}\label{eq:enHF}
\frac{E_{L}(N_{\uparrow}, N_{\downarrow})}{L^{3}} &\leq& \frac{E_{\text{HF}}(\omega)}{L^{3}}\nonumber\\
&=&  \frac{3}{5}(6\pi^{2})^{\frac{2}{3}} (\rho_{\uparrow}^{\frac{5}{3}} + \rho_{\downarrow}^{\frac{5}{3}}) + \hat V(0) \rho_{\uparrow} \rho_{\downarrow} + O(\rho^{\frac{7}{3}})\;.
\end{eqnarray}
In Eq. (\ref{eq:enHF}), the effect of the interaction is only visible via the average of the potential, $\hat V(0) = \int_{\Lambda_{L}} dx\, V(x)$. The mismatch between (\ref{eq:main}) and (\ref{eq:enHF}) will be due to the correlations between the particles in the many-body ground state, which are absent in the free Fermi gas.
\subsubsection{Fermionic Bogoliubov transformation}\label{sec:bogHcorr}
In this section we shall introduce a suitable unitary transformation in Fock space, that will allow us to efficiently compare the many-body ground state energy with the energy of the free Fermi gas. 

Given the reduced one-particle density matrix of the free Fermi gas, $\omega_{\sigma, \sigma'} = \delta_{\sigma, \sigma'}\sum_{k\in \mathcal{B}^{\sigma}_{F}} |f_{k}\rangle \langle f_{k}|$, we define the operators $u: L^{2}(\Lambda_{L}; \mathbb{C}^{2}) \to L^{2}(\Lambda_{L}; \mathbb{C}^{2})$ and $v: L^{2}(\Lambda_{L};\mathbb{C}^{2}) \to L^{2}(\Lambda_{L}; \mathbb{C}^{2})$ as:
\begin{equation}
u_{\sigma, \sigma'}(x,y) = \delta_{\sigma, \sigma'}\delta(x-y) - \omega_{\sigma, \sigma'}(x;y)\;,\qquad v_{\sigma, \sigma'}(x;y) = \delta_{\sigma, \sigma'}\sum_{k\in \mathcal{B}^{\sigma}_{F}} |\overline{f_{k}}\rangle \langle f_{k}|\;.
\end{equation} 
The symbol $\delta(x-y)$ denotes the periodic Dirac delta distribution on $\Lambda_{L}$, see Eq. (\ref{eq:perDir}). Clearly, 
\begin{equation}\label{eq:uvprop}
u \overline{v} = 0\;,\qquad \overline{v} v = \omega\;.
\end{equation}
By the Shale-Stinespring theorem, see \cite{Sol} for a pedagogical introduction to the topic, there exists a unitary operator $R: \mathcal{F} \to \mathcal{F}$ such that the following holds.
\begin{itemize}
\item[(i)] The vector $R\Omega$ is an $N$-particle state, which reproduces the Slater determinant $\Phi_{\text{FFG}}$, Eq. (\ref{eq:AFFS}):
\begin{equation}\label{eq:defbogi}
(R\Omega)^{(n)} = 0\quad \text{unless $n = N$, in which case}\quad  (R\Omega)^{(N)}  = \Phi_{\text{FFG}}\;.
\end{equation}
\item[(ii)] The map $R:\mathcal{F} \to \mathcal{F}$ implements the following transformation in Fock space:
\begin{equation}\label{eq:bogii}
R^{*} a(f) R = a(uf) + a^{*}(\overline{v} \overline{f})\;,\qquad \text{for all $f\in L^{2}(\Lambda_{L}; \mathbb{C}^{2})$}\;.
\end{equation}
Equivalently,
\begin{equation}\label{eq:RaR}
R^{*} a_{x,\sigma}R = a_{\sigma}(u_{x}) + a^{*}_{\sigma}(\overline{v}_{x})\;,
\end{equation}
where, setting $u_{\sigma,\sigma} \equiv u_{\sigma}$, $v_{\sigma,\sigma}\equiv v_{\sigma}$, and $u_{x}(y) \equiv u(x;y)$, $v_{x}(y) \equiv v(x;y)$:
\begin{equation}
a_{\sigma}(u_{x}) = \int dx\, a_{x,\sigma} \overline{u_{\sigma}(x;y)}\;,\qquad a^{*}_{\sigma}(\overline{v}_{x}) = \int dy\, a^{*}_{y,\sigma} \overline{v_{\sigma}(y;x)}\;.
\end{equation}
Let $\hat a_{k,\sigma} = a_{\sigma}(f_{k})$, recall Eq. (\ref{eq:Fou}). Then, the transformation (\ref{eq:bogii}) reads:
\begin{equation}\label{eq:phol}
R^{*} \hat a_{k,\sigma} R = \left\{ \begin{array}{cc}  \hat a_{k,\sigma} & \text{if $k\notin \mathcal{B}^{\sigma}_{F}$} \\ \hat a^{*}_{k,\sigma} & \text{if $k\in \mathcal{B}^{\sigma}_{F}$.}\end{array}\right.
\end{equation}
Thus, Eq. (\ref{eq:bogii}) can be seen as implementing a particle-hole transformation. By the unitarity of $R$, Eq. (\ref{eq:phol}) also implies that $R^{*} \hat a_{k,\sigma} R = R \hat a_{k,\sigma} R^{*}$.
\end{itemize}
The operator $R$ is a Bogoliubov transformation, see \cite{BLS, Sol} for reviews. The property $(i)$ immediately implies:
\begin{equation}
\langle R\Omega, \mathcal{H} R\Omega\rangle = E_{\text{HF}}(\omega)\;. 
\end{equation}
Instead, property $(ii)$ allows to compare the energy of any state in the fermionic Fock space, with the energy of the free Fermi gas. This is the content of the next proposition. From now on, we shall simply write $\sum_{\sigma}$ for $\sum_{\sigma = \uparrow\downarrow}$ and $\int dx$ for $\int_{\Lambda_{L}}dx$.
\begin{proposition}\label{prp:conj} Let $\psi \in \mathcal{F}$ be a normalized state, such that $\langle \psi, \mathcal{N}_{\sigma} \psi \rangle = N_{\sigma}$ and $N = N_{\uparrow} + N_{\downarrow}$. Then:
\medskip

\noindent{(i)} The following identity holds true:
\begin{equation}\label{eq:bog0}
\langle \psi, \mathcal{H} \psi \rangle = E_{\text{HF}}(\omega) + \langle R^{*}\psi, \mathbb{H}_{0} R^{*}\psi \rangle + \langle R^{*}\psi, \mathbb{X} R^{*}\psi \rangle + \langle R^{*} \psi, \mathbb{Q} R^{*} \psi \rangle
\end{equation}
where: the operators $\mathbb{H}_{0}$, $\mathbb{X}$ are given by:
\begin{eqnarray}
\mathbb{H}_{0} &=& \sum_{k,\sigma} | |k|^{2} - \mu_{\sigma} |  \hat a^{*}_{k,\sigma} \hat a_{k,\sigma}\;,\qquad \mu_{\sigma} = k_{F}^{\sigma 2}\;,\\ \mathbb{X} &=& \sum_{\sigma}\int dxdy\, V(x-y) \omega_{\sigma}(x-y) ( a^{*}_{\sigma}(u_{x}) a_{\sigma}(u_{y}) - a^{*}_{\sigma}(\overline{v}_{y}) a_{\sigma}(\overline{v}_{x}))\;. \nonumber
\end{eqnarray}
The operator $\mathbb{Q}$ can be written as $\mathbb{Q} = \sum_{i=1}^{4} \mathbb{Q}_{i}$ with:
\begin{eqnarray}\label{eq:Q0}
\mathbb{Q}_{1} &=& \frac{1}{2} \sum_{\sigma, \sigma'} \int dxdy\, V(x-y) a_{\sigma}^{*}(u_{x})a_{\sigma'}^{*}(u_{y})a_{\sigma'}(u_{y})a_{\sigma}(u_{x})\\
\mathbb{Q}_{2} &=& \frac{1}{2}\sum_{\sigma, \sigma'}\int dxdy\, V(x-y) \Big[ a^{*}_{\sigma}(u_{x}) a^{*}_{\sigma}(\overline{v}_{x}) a_{\sigma'}(\overline{v}_{y}) a_{\sigma'}(u_{y})\nonumber\\&& - 2a_{\sigma}^{*}(u_{x})a_{\sigma'}^{*}(\overline v_{y})a_{\sigma'}(\overline v_{y})a_{\sigma}(u_{x}) + a_{\sigma'}^{*}(\overline v_{y})a_{\sigma}^{*}(\overline v_{x})a_{\sigma}(\overline v_{x})a_{\sigma'}(\overline v_{y})\Big]\nonumber\\
\mathbb{Q}_{3} &=& -\sum_{\sigma, \sigma'} \int dxdy\, V(x-y) \Big[ a^{*}_{\sigma}(u_{x}) a^{*}_{\sigma'}(u_{y}) a^{*}_{\sigma}(\overline{v}_{x}) a_{\sigma'}(u_{y}) - a^{*}_{\sigma}(u_{x}) a^{*}_{\sigma'}(\overline{v}_{y}) a^{*}_{\sigma}(\overline{v}_{x}) a_{\sigma'}(\overline{v}_{y})\Big] + \text{h.c.} \nonumber\\
\mathbb{Q}_{4} &=& \frac{1}{2} \sum_{\sigma, \sigma'} \int dxdy\, V(x-y) a^{*}_{\sigma}(u_{x}) a^{*}_{\sigma'}(u_{y}) a^{*}_{\sigma'}(\overline{v}_{y}) a^{*}_{\sigma}(\overline{v}_{x})  + \text{h.c.}\nonumber
\end{eqnarray}
\noindent{(ii)} The following inequality holds true:
\begin{equation}\label{eq:bog00}
\langle \psi, \mathcal{H} \psi \rangle \geq E_{\text{HF}}(\omega) + \langle R^{*}\psi, \mathbb{H}_{0} R^{*}\psi \rangle + \langle R^{*}\psi, \mathbb{X} R^{*}\psi \rangle + \langle R^{*} \psi, \widetilde{\mathbb{Q}} R^{*} \psi \rangle\;,
\end{equation}
where $\widetilde{\mathbb{Q}} = \sum_{i=1}^{4} \widetilde{\mathbb{Q}}_{i}$ and:
\begin{eqnarray}\label{eq:tildeQs}
\widetilde{\mathbb{Q}}_{1} &=& \frac{1}{2} \sum_{\sigma\neq \sigma'} \int dxdy\, V(x-y) a_{\sigma}^{*}(u_{x})a_{\sigma'}^{*}(u_{y})a_{\sigma'}(u_{y})a_{\sigma}(u_{x})\\
\widetilde{\mathbb{Q}}_{2} &=& \frac{1}{2}\sum_{\sigma\neq \sigma'}\int dxdy\, V(x-y) \Big[ a^{*}_{\sigma}(u_{x}) a^{*}_{\sigma}(\overline{v}_{x}) a_{\sigma'}(\overline{v}_{y}) a_{\sigma'}(u_{y})\nonumber\\&& - 2a_{\sigma}^{*}(u_{x})a_{\sigma'}^{*}(\overline v_{y})a_{\sigma'}(\overline v_{y})a_{\sigma}(u_{x}) + a_{\sigma'}^{*}(\overline v_{y})a_{\sigma}^{*}(\overline v_{x})a_{\sigma}(\overline v_{x})a_{\sigma'}(\overline v_{y})\Big]\nonumber\\
\widetilde{\mathbb{Q}}_{3} &=& -\sum_{\sigma\neq \sigma'} \int dxdy\, V(x-y) \Big[ a^{*}_{\sigma}(u_{x}) a^{*}_{\sigma'}(u_{y}) a^{*}_{\sigma}(\overline{v}_{x}) a_{\sigma'}(u_{y}) - a^{*}_{\sigma}(u_{x}) a^{*}_{\sigma'}(\overline{v}_{y}) a^{*}_{\sigma}(\overline{v}_{x}) a_{\sigma'}(\overline{v}_{y})\Big] + \text{h.c.} \nonumber\\
\widetilde{\mathbb{Q}}_{4} &=& \frac{1}{2} \sum_{\sigma\neq \sigma'} \int dxdy\, V(x-y) a^{*}_{\sigma}(u_{x}) a^{*}_{\sigma'}(u_{y}) a^{*}_{\sigma'}(\overline{v}_{y}) a^{*}_{\sigma}(\overline{v}_{x})  + \text{h.c.}\nonumber
\end{eqnarray}
\end{proposition}
\begin{proof} $(i)$ To prove this identity we transformed each fermionic operator according to (\ref{eq:RaR}) and we put the resulting expression into normal order, using the canonical anticommutation relations (\ref{eq:CAR}) and the properties (\ref{eq:uvprop}). The details of the computation have been given already in a number of places and hence will be omitted; see for instance \cite{BPS, BJPSS, HPR, BNPSS, BNPSS2}.
\medskip

\noindent{$(ii)$} We use that:
\begin{eqnarray}\label{eq:nospin}
&&\langle \psi, \mathcal{H} \psi \rangle \nonumber\\
&& = \sum_{\sigma} \int dx\, \| \nabla a_{x,\sigma} \psi \|^{2} + \frac{1}{2} \sum_{\sigma \neq \sigma'} \int dxdy\, V(x-y) \| a_{x,\sigma} a_{y,\sigma'} \psi \|^{2} + \frac{1}{2} \sum_{\sigma} \int dxdy\, V(x-y) \| a_{x,\sigma} a_{y,\sigma} \psi\|^{2} \nonumber\\
&&\geq \sum_{\sigma} \int dx\, \| \nabla a_{x,\sigma} \psi \|^{2} + \frac{1}{2} \sum_{\sigma \neq \sigma'} \int dxdy\, V(x-y) \| a_{x,\sigma} a_{y,\sigma'} \psi \|^{2}\;,
\end{eqnarray}
and then we repeat the proof of $(i)$ for the right-hand side of (\ref{eq:nospin}).
\end{proof}
As we will prove, the terms $\langle R^{*}\psi, \mathbb{H}_{0} R^{*}\psi \rangle$, $\langle R^{*} \psi, \mathbb{Q} R^{*} \psi \rangle$ give a contribution to the ground state energy of order $L^{3}\rho^{2}$, which will allow to reconstruct the scattering length in the final result (\ref{eq:main}). Before proving this, let us establish some useful estimates for the various terms arising in (\ref{eq:bog0}), that will allow us to identify terms that are subleading with respect to $L^{3} \rho^{2}$.
\begin{proposition}\label{prp:bogbd} Under the assumptions of Theorem \ref{thm:main}, the following holds.
\begin{itemize}
\item[a)] The operator $\mathbb{X}$ satisfies the bound:
\begin{equation}\label{eq:Xbd}
|\langle \psi, \mathbb{X} \psi \rangle| \leq C\rho \langle \psi, \mathcal{N} \psi \rangle\;.
\end{equation}
\item[b)] The operators $\mathbb{Q}_{1}$, $\widetilde{\mathbb{Q}}_{1}$ are nonnegative.
\item[c)] The operators $\mathbb{Q}_{2}$, $\widetilde{\mathbb{Q}}_{2}$ satisfy the bounds:
\begin{equation}\label{eq:Q2bd}
| \langle \psi, \mathbb{Q}_{2}\psi\rangle | \leq C\rho\langle \psi, \mathcal{N} \psi \rangle\;,\qquad | \langle \psi, \mathbb{\widetilde{Q}}_{2}\psi\rangle | \leq C\rho\langle \psi, \mathcal{N} \psi \rangle\;.
\end{equation}
\item[d)] The operators $\mathbb{Q}_{3}$, $\widetilde{\mathbb{Q}}_{3}$ satisfy the bounds, for any $\alpha \geq 0$:
\begin{equation}\label{eq:bdQ3}
|\langle \psi, \mathbb{Q}_{3} \psi \rangle| \leq \rho^{\alpha}\langle \psi, \mathbb{Q}_{1} \psi \rangle + C\rho^{1-\alpha} \langle \psi, \mathcal{N} \psi \rangle\;,\quad |\langle \psi, \mathbb{\widetilde{Q}}_{3} \psi \rangle| \leq \rho^{\alpha}\langle \psi, \mathbb{\widetilde{Q}}_{1} \psi \rangle + C\rho^{1-\alpha} \langle \psi, \mathcal{N} \psi \rangle\;.
\end{equation}
Furthermore, suppose that $\psi$ is a Fock space vector such that $\psi^{(n)} = 0$ unless $n = 4k$ for $k\in \mathbb{N}$. Then:
\begin{equation}\label{eq:Q3canc}
\langle \psi, \mathbb{Q}_{3} \psi \rangle =0\;,\qquad \langle \psi, \widetilde{\mathbb{Q}}_{3} \psi \rangle = 0\;.
\end{equation}
\item[e)] The operators $\mathbb{Q}_{4}$, $\widetilde{\mathbb{Q}}_{4}$ satisfy the bounds, for any $\delta > 0$:
\begin{equation}\label{eq:bdQ4}
| \langle \psi, \mathbb{Q}_{4} \psi\rangle | \leq \delta \langle \psi, \mathbb{Q}_{1} \psi\rangle + \frac{C}{\delta} \rho^{2} L^{3} \| \psi \|^{2}\;,\qquad | \langle \psi, \mathbb{\widetilde{Q}}_{4} \psi\rangle | \leq \delta \langle \psi, \widetilde{\mathbb{Q}}_{1} \psi\rangle + \frac{C}{\delta} \rho^{2} L^{3} \| \psi \|^{2}\;.
\end{equation}
\end{itemize}
\end{proposition}
\begin{proof} We shall only prove the statements for the $\mathbb{Q}_{i}$ operators; the analogous statements for the $\widetilde{\mathbb{Q}}_{i}$ operators are proven in exactly the same way.
\medskip

\noindent\underline{Proof of $a).$} We have, using the notation $\| \omega_{\sigma, x} \|_{\infty} = \sup_{y} | \omega_{\sigma}(x;y)|$:
\begin{eqnarray}\label{eq:a0}
| \langle \psi, \mathbb{X} \psi \rangle | &\leq& \sum_{\sigma} \|\omega_{\sigma,x}\|_{\infty}\int dxdy\, V(x-y) \big( \| a_{\sigma}(u_{x}) \psi\|^{2} + \| a_{\sigma}(\overline{v}_{x}) \psi\|^{2} \big) \nonumber\\
&\leq& \rho \|V\|_{1} \langle \psi, \mathcal{N} \psi \rangle
\end{eqnarray}
where we used that $\| \omega_{\sigma, x} \|_{\infty} \leq \rho_{\sigma}$, and that:
\begin{equation}
\int dx\, \| a_{\sigma}(u_{x}) \psi \|^{2} = \sum_{k \notin \mathcal{B}_{F}^{\sigma}} \langle \psi, \hat a^{*}_{k,\sigma} \hat a_{k,\sigma}\psi \rangle\;,\qquad \int dx\, \| a_{\sigma}(\overline{v}_{x}) \psi \|^{2} = \sum_{k \in \mathcal{B}_{F}^{\sigma}} \langle \psi, \hat a^{*}_{k,\sigma} \hat a_{k,\sigma}\psi \rangle\;,
\end{equation}
from which the final bound in (\ref{eq:a0}) immediately follows (recall the expression for the number operator, (\ref{eq:calN})).
\medskip

\noindent\underline{Proof of $b)$.} We have:
\begin{equation}
\langle \psi, \mathbb{Q}_{1} \psi \rangle = \frac{1}{2} \sum_{\sigma, \sigma'} \int dxdy\, V(x-y) \| a_{\sigma}(u_{x}) a_{\sigma'}(u_{y}) \psi\|^{2} \geq 0\;,
\end{equation}
where we used that $V(x-y) \geq 0$.
\medskip

\noindent\underline{Proof of $c)$.} We have:
\begin{eqnarray}
&&\Big|\sum_{\sigma, \sigma'}\int dxdy\, V(x-y) \langle \psi, a^{*}_{\sigma}(u_{x}) a^{*}_{\sigma}(\overline{v}_{x}) a_{\sigma'}(\overline{v}_{y}) a_{\sigma'}(u_{y}) \psi \rangle \Big| \nonumber\\
&&\qquad \leq \sum_{\sigma,\sigma'} \int dxdy\, V(x-y) \|v_{\sigma, x}\|_{2} \|v_{\sigma', y}\|_{2}  \| a_{\sigma}(u_{x}) \psi \| \| a_{\sigma'}(u_{y}) \psi\| \nonumber\\
&&\qquad \leq C\rho \sum_{\sigma, \sigma'} \int dxdy\, V(x-y) \| a_{\sigma}(u_{x}) \psi \|^{2} \nonumber\\
&& \qquad \leq C \rho \| V\|_{1} \langle \psi, \mathcal{N} \psi \rangle\;.
\end{eqnarray}
The second inequality follows from Cauchy-Schwarz inequality and:
\begin{equation}
\|v_{\sigma, x}\|_{2}^{2} = \int dy\, |v_{\sigma}(x;y)|^{2} = \omega_{\sigma}(x;x) = \rho_{\sigma}\;.
\end{equation}
The last inequality follows from:
\begin{equation}
\sum_{\sigma} \int dx\, \| a_{\sigma}(u_{x}) \psi \|^{2} \leq \langle \psi, \mathcal{N} \psi \rangle\;.
\end{equation}
The other two terms in the definition of $\mathbb{Q}_{2}$ are estimated in exactly the same way.
\medskip

\noindent{\underline{Proof of $d)$.}} Consider the first contribution to $\mathbb{Q}_{3}$. We write:
\begin{eqnarray}
&&\Big| \sum_{\sigma, \sigma'} \int dxdy\, V(x-y) \langle \psi, a^{*}_{\sigma}(u_{x}) a^{*}_{\sigma'}(u_{y}) a^{*}_{\sigma}(\overline{v}_{x}) a_{\sigma'}(u_{y})\psi \rangle\Big| \nonumber\\
&&\leq \frac{\rho^{\alpha}}{8}\sum_{\sigma, \sigma'}\int dxdy\, V(x-y) \| a_{\sigma}(u_{x}) a_{\sigma'}(u_{y}) \psi \|^{2}  + \frac{2}{\rho^{\alpha}}  \sum_{\sigma, \sigma'}\int dxdy\, V(x-y) \| a^{*}_{\sigma}(\overline{v}_{x}) a_{\sigma'}(u_{y}) \psi \|^{2} \nonumber\\
&&\leq \frac{\rho^{\alpha}}{4} \langle \psi, \mathbb{Q}_{1} \psi \rangle + C \|V\|_{1} \rho^{1-\alpha} \langle \psi, \mathcal{N} \psi \rangle\;,
\end{eqnarray}
by Cauchy-Schwarz inequality and $\|v_{\sigma, x}\|_{2} \leq \rho^{\frac{1}{2}}$. Consider now the second contribution to $\mathbb{Q}_{3}$. We have:
\begin{eqnarray}
&&\Big| \sum_{\sigma, \sigma'} \int dxdy\, V(x-y) \langle \psi,  a^{*}_{\sigma}(u_{x}) a^{*}_{\sigma'}(\overline{v}_{y}) a^{*}_{\sigma}(\overline{v}_{x}) a_{\sigma'}(\overline{v}_{y}) \psi \rangle \Big|\nonumber\\
&&\qquad \leq \sum_{\sigma, \sigma'} \int dxdy\, V(x-y) \| v_{\sigma', y} \|_{2} \| v_{\sigma, x} \| \| a_{\sigma}(u_{x}) \psi  \| \| a_{\sigma'}(\overline{v}_{y}) \psi \|\nonumber\\
&&\qquad \leq C\rho \|V\|_{1} \langle \psi, \mathcal{N} \psi\rangle\;,
\end{eqnarray}
again by Cauchy-Schwarz inequality and $\|v_{\sigma, x}\|_{2} \leq \rho^{\frac{1}{2}}$. The remaining terms in the definition of $\mathbb{Q}_{3}$ are estimated in exactly the same way. Let us now prove the identities (\ref{eq:Q3canc}). Consider the first; the proof of the second is identical. Let $\psi$ be such that $\psi^{(n)} = 0$ unless $n=4k$ for $k\in \mathbb{N}$. Let $\varphi = \mathbb{Q}_{3} \psi$. Then, $\varphi$ is a Fock space vector such that $\varphi^{(n)} = 0$ unless $n = 4k + 2$ for $k\in \mathbb{N}$. Since $4k + 2$ is not a multiple of $4$, $\langle \psi, \varphi \rangle = \sum_{n\geq 0} \langle \psi^{(n)}, \varphi^{(n)} \rangle = 0$.
\medskip

\noindent\underline{Proof of $e)$.} Consider the first contribution to $\mathbb{Q}_{4}$. We write, by Cauchy-Schwarz inequality, for $\delta > 0$:
\begin{eqnarray}\label{eq:Q4Q1}
&&\Big| \sum_{\sigma, \sigma'} \int dxdy\, V(x-y) \langle \psi, a^{*}_{\sigma}(u_{x}) a^{*}_{\sigma'}(u_{y}) a^{*}_{\sigma'}(\overline{v}_{y}) a^{*}_{\sigma}(\overline{v}_{x}) \psi \rangle \Big| \nonumber\\
&&\leq \sum_{\sigma, \sigma'} \int dxdy\, V(x-y) \Big[ \frac{\delta}{2}\| a_{\sigma}(u_{x}) a_{\sigma'}(u_{y}) \psi \|^{2} + \frac{1}{2\delta} \| a^{*}_{\sigma'}(\overline{v}_{y}) a^{*}_{\sigma}(\overline{v}_{x}) \psi \|^{2} \Big]\nonumber\\
&&\leq \delta \langle \psi, \mathbb{Q}_{1} \psi \rangle + \frac{C}{\delta} \rho^{2} L^{3} \|\psi\|^{2}\;.
\end{eqnarray}
The other contribution to $\mathbb{Q}_{4}$ is bounded in the same way. This concludes the proof of Proposition \ref{prp:bogbd}.
\end{proof}
In order to make good use of the above estimates, we need a priori information on the size of the expectation of the number operator, on states that are close enough to the ground state of the system. We shall refer to these states as approximate ground states
\begin{definition}[Approximate ground state.]\label{def:appgs} Let $\psi \in \mathcal{F}$ be a normalized state, such that $\langle \psi, \mathcal{N}_{\sigma} \psi \rangle = N_{\sigma}$ and $N = N_{\uparrow} + N_{\downarrow}$. Suppose that:
\begin{equation}\label{eq:H0as}
\Big| \langle \psi, \mathcal{H} \psi \rangle - \sum_{\sigma = \uparrow\downarrow} \sum_{k\in \mathcal{B}_{F}^{\sigma}} |k|^{2} \Big| \leq CL^{3} \rho^{2}\;.
\end{equation}
Then, we shall say that $\psi$ is an approximate ground state of $\mathcal{H}$.
\end{definition}
We will first get an a priori estimate on the relative kinetic energy operator $\mathbb{H}_{0}$. Afterwards, we will show how to get information on the number operator from this a priori bound. 
\begin{lemma}[A priori estimate for $\mathbb{H}_{0}$.]\label{lem: apH0} Under the assumptions of Theorem \ref{thm:main}, the following holds. Suppose that $\psi$ is an approximate ground state. Then:
\begin{equation}\label{eq:aprH0}
\langle R^{*}\psi, \mathbb{H}_{0} R^{*}\psi \rangle \leq C L^{3}\rho^{2}\;.
\end{equation}
\end{lemma}
\begin{proof} By the positivity of the interaction,
\begin{eqnarray}\label{eq:lowH}
\langle \psi, \mathcal{H} \psi \rangle &\geq& \sum_{\sigma =\uparrow\downarrow}\sum_{k\in \frac{2\pi}{L}\mathbb{Z}^{3}} |k|^{2}\langle \psi, \hat a^{*}_{k,\sigma} \hat a_{k,\sigma} \psi \rangle \\
&=&  \sum_{\sigma =\uparrow\downarrow}\sum_{k\in \frac{2\pi}{L}\mathbb{Z}^{3}} |k|^{2}\langle R^{*}\psi, R^{*}\hat a^{*}_{k,\sigma} \hat a_{k,\sigma} R R^{*} \psi \rangle\nonumber\\
&=&  \sum_{\sigma = \uparrow\downarrow} \sum_{k\in \mathcal{B}_{F}^{\sigma}} |k|^{2} +  \sum_{\sigma = \uparrow\downarrow} \sum_{k \notin \mathcal{B}_{F}^{\sigma}} |k|^{2} \langle R^{*}\psi, \hat a^{*}_{k,\sigma} \hat a_{k,\sigma} R^{*} \psi\rangle - \sum_{\sigma = \uparrow\downarrow} \sum_{k \in \mathcal{B}_{F}^{\sigma}} |k|^{2} \langle R^{*}\psi, \hat a^{*}_{k,\sigma} \hat a_{k,\sigma} R^{*}\psi \rangle\;,\nonumber
\end{eqnarray}
where the last step follows from (\ref{eq:phol}). We then rewrite the last two terms as:
\begin{eqnarray}
&&\sum_{\sigma = \uparrow\downarrow} \sum_{k \notin \mathcal{B}_{F}^{\sigma}} (|k|^{2} - \mu_{\sigma}) \langle R^{*}\psi, \hat a^{*}_{k,\sigma} \hat a_{k,\sigma} R^{*}\psi \rangle - \sum_{\sigma = \uparrow\downarrow} \sum_{k \in \mathcal{B}_{F}^{\sigma}} (|k|^{2} - \mu_{\sigma}) \langle R^{*}\psi, \hat a^{*}_{k,\sigma} \hat a_{k,\sigma} R^{*}\psi \rangle\nonumber\\
&&\qquad  + \sum_{\sigma = \uparrow\downarrow} \mu_{\sigma} \Big[ \sum_{k \notin \mathcal{B}_{F}^{\sigma}} \langle R^{*}\psi, \hat a^{*}_{k,\sigma} \hat a_{k,\sigma} R^{*}\psi \rangle - \sum_{k \in \mathcal{B}_{F}^{\sigma}} \langle R^{*}\psi, \hat a^{*}_{k,\sigma} \hat a_{k,\sigma} R^{*}\psi \rangle\Big] \nonumber\\
&&\qquad \equiv \langle R^{*}\psi , \mathbb{H}_{0} R^{*}\psi \rangle + \sum_{\sigma = \uparrow\downarrow} \mu_{\sigma}\Big[ \sum_{k\in \frac{2\pi}{L} \mathbb{Z}^{3}}\langle \psi, \hat a^{*}_{k,\sigma} \hat a_{k,\sigma} \psi \rangle - N_{\sigma}\Big]\;.
\end{eqnarray}
To reconstruct the kinetic energy operator $\mathbb{H}_{0}$, we used that if $k\notin \mathcal{B}_{F}^{\sigma}$ then $|k|^{2} - \mu_{\sigma} \geq 0$, while if $k\in \mathcal{B}_{F}^{\sigma}$ then $|k|^{2} - \mu_{\sigma} \leq 0$. To obtain the term in the square brackets, we used again the properties of the Bogoliubov transformation (\ref{eq:phol}). The term in the brackets vanishes, by the assumptions on the state. The bound (\ref{eq:lowH}), combined with the assumption (\ref{eq:H0as}), implies the final claim. 
\end{proof}
We are now ready to prove an a priori estimate on the number operator. To do so, the following lemma will play an important role.
\begin{lemma}\label{lem:apriori} Let $\alpha \geq \frac{2}{3}$. The following bound holds true:
\begin{equation}\label{eq:apriori}
\langle \psi, \mathcal{N} \psi \rangle \leq CL^{3} \rho^{\frac{1}{3} + \alpha} \|\psi\|^{2} + \frac{1}{\rho^{\alpha}}\langle \psi, \mathbb{H}_{0} \psi \rangle\;.
\end{equation}
\end{lemma}
\begin{proof}
We write:
\begin{eqnarray}
\mathcal{N} &=& \sum_{\sigma} \sum_{k} \hat a^{*}_{\sigma, k} \hat a_{\sigma, k}  = \sum_{\sigma} \sum_{k: |k^{2} - \mu_{\sigma}| \leq \rho^{\alpha}} \hat a^{*}_{\sigma, k} \hat a_{\sigma, k} + \sum_{\sigma} \sum_{k: |k^{2} - \mu_{\sigma}| > \rho^{\alpha}} \hat a^{*}_{\sigma, k} \hat a_{\sigma, k}\nonumber\\
&\equiv& \mathcal{N}^{<} + \mathcal{N}^{>}\;.
\end{eqnarray}
For the first term, we use that:
\begin{equation}
\langle \psi, \mathcal{N}^{<} \psi \rangle = \sum_{\sigma} \sum_{k: |k^{2} - \mu_{\sigma}|  \leq \rho^{\alpha}} \| \hat a_{\sigma, k} \psi \|^{2} \leq CL^{3} \rho^{\frac{1}{3} + \alpha} \|\psi\|^{2}\;.
\end{equation}
For the second term we use that:
\begin{eqnarray}\label{eq:Nbig}
\langle \psi, \mathcal{N}^{>} \psi \rangle &=&  \sum_{\sigma} \sum_{k: |k^{2} - \mu_{\sigma}|  > \rho^{\alpha}} \| \hat a_{\sigma, k} \psi \|^{2} \nonumber\\ 
&\leq& \frac{1}{\rho^{\alpha}} \sum_{\substack{k,\sigma \\ |k^{2} - \mu_{\sigma}|  > \rho^{\alpha}}} |k^{2} - \mu_{\sigma}| \| \hat a_{\sigma, k} \psi \|^{2} \leq \frac{1}{\rho^{\alpha}}\langle \psi, \mathbb{H}_{0} \psi \rangle\;.
\end{eqnarray}
This concludes the proof. 
\end{proof}
\begin{corollary}[A priori estimate for $\mathcal{N}$.]\label{cor:Nbd} Under the assumptions of Lemma \ref{lem: apH0}, the following holds:
\begin{equation}\label{eq:bdcalN}
\langle R^{*}\psi, \mathcal{N} R^{*}\psi \rangle \leq CL^{3} \rho^{\frac{7}{6}}\;.
\end{equation}
\end{corollary}
\begin{proof} The bound follows from Eqs. (\ref{eq:apriori}), (\ref{eq:aprH0}), after optimizing over $\alpha$.
\end{proof}
\begin{remark}[Condensation estimate.] It is well-known that the estimate (\ref{eq:bdcalN}) can be used to control the difference between the reduced one-particle density matrix of $\psi$ and the reduced one-particle density matrix of the free Fermi gas, see {\it e.g.} \cite{BPS}. Let $\gamma^{(1)}_{\psi}$ be the reduced one-particle density matrix of $\psi$,
\begin{equation}
\gamma^{(1)}_{\sigma, \sigma'}(x;y) = \langle \psi, a^{*}_{y,\sigma'} a_{x,\sigma} \psi \rangle\;.
\end{equation}
Then, the bound (\ref{eq:bdcalN}) implies that, for an approximate ground state $\psi$:
\begin{equation}\label{eq:condest}
\tr\, \gamma^{(1)}_{\psi} (1 - \omega) \leq CL^{3} \rho^{\frac{7}{6}}\;.
\end{equation}
This `condensation estimate' is not new: it was an important ingredient of the analysis of \cite{LSS}. One of the reasons for our improved error estimates in Theorem \ref{thm:main} is that we will be able to improve the bound (\ref{eq:condest}), see Remark \ref{rem:cond}.
\end{remark}

We conclude this section by discussing an a priori estimate for the operator $\mathbb{Q}_{1}$, arising after conjugating the Hamiltonian with the fermionic Bogoliubov transformation.
\begin{lemma}[A priori estimate for $\mathbb{Q}_{1}$.] \label{lem: est priori Q1} Under the assumptions of Theorem \ref{thm:main}, the following is true. Suppose that $\psi$ an approximate ground state. Then:
\begin{equation}\label{eq:aprQ1}
\langle R^{*}\psi, \mathbb{Q}_{1} R^{*}\psi \rangle \leq CL^{3} \rho^{2}\;, \qquad \langle R^{*}\psi, \widetilde{\mathbb{Q}}_{1} R^{*}\psi \rangle \leq CL^{3} \rho^{2}\;.
\end{equation}
\end{lemma}
\begin{proof}
From the estimates for $\mathbb{X}$, $\mathbb{Q}_2$, $\mathbb{Q}_3$, Eqs. \eqref{eq:Xbd}, \eqref{eq:Q2bd} and \eqref{eq:bdQ3}, we get:
\begin{equation}
	\langle \psi, \mathcal{H}\psi\rangle \geq E_{\text{HF}}(\omega)  + \langle R^{*}\psi, \mathbb{H}_0 R^{*}\psi\rangle +(1-C\rho^\alpha)\langle R^{*}\psi, \mathbb{Q}_1 R^{*}\psi\rangle + \langle R^{*}\psi, \mathbb{Q}_4 R^{*}\psi\rangle + \mathcal{E}(\psi),
\end{equation}
with
\begin{equation}\label{eq:estNop}
	|\mathcal{E}(\psi)|\leq C \rho^{1-\alpha}\langle R^{*} \psi, \mathcal{N} R^{*}\psi\rangle.
\end{equation}
Eq. \eqref{eq:estNop} together with the bound $\pm \mathbb{Q}_{4} \leq \delta \mathbb{Q}_{1} + (C/\delta) L^{3}\rho^{2}$, Eq. (\ref{eq:bdQ4}), imply, taking {\it e.g.} $\alpha = 1/12$:
\begin{equation}
\langle \psi, \mathcal{H} \psi \rangle \geq E_{\text{HF}}(\omega) + \langle R^{*}\psi, \mathbb{H}_{0} R^{*}\psi \rangle + (1 - C\delta) \langle R^{*}\psi, \mathbb{Q}_{1} R^{*}\psi \rangle - \frac{C}{\delta}L^{3}\rho^{2}\;.
\end{equation}
Taking $\delta >0$ small enough, the final claim follows from assumption (\ref{eq:H0as}), from the explicit expression of $E_{\text{HF}}(\omega)$, Eq. (\ref{eq:enHF}), and from the positivity of $\mathbb{H}_{0}$.  The inequality for $\widetilde{\mathbb{Q}}_{1}$ follows immediately, since $\mathbb{Q}_{1} \geq \widetilde{\mathbb{Q}}_{1}$.
\end{proof}
\section{The correlation structure}\label{sec:T}
\subsection{Heuristics}\label{sec:heu}
Here we shall give the intuition behind the method developed in the rest of the paper. Recall the expression for the many-body energy, after conjugating with the Bogoliubov transformation:
\begin{equation}\label{eq:en}
\langle \psi, \mathcal{H} \psi \rangle = E_{\text{HF}}(\omega) + \langle R^{*}\psi, (\mathbb{H}_{0} + \mathbb{Q}_{1} + \mathbb{Q}_{4}) R^{*}\psi \rangle + \mathcal{E}_{\text{1}}(\psi)\;.
\end{equation}
If $\psi$ is an approximate ground state, the error term $\mathcal{E}_{1}(\psi)$ is subleading with respect to $\rho^{2}$, as a consequence of the estimates proven in the previous section. For the sake of the following heuristic discussion, we shall neglect it. The operator $\mathbb{Q}_{4}$ can be rewritten as:
\begin{equation}
\mathbb{Q}_{4} = \frac{1}{2}\sum_{\sigma,\sigma'} \frac{1}{L^{3}}\sum_{p} \hat V(p) \hat b_{p,\sigma} \hat b_{-p,\sigma'} + \text{h.c.}\;,
\end{equation}
with:
\begin{equation}
\hat b_{p,\sigma} = \int dx\, e^{ip\cdot x} a_{\sigma}(u_{x}) a_{\sigma}(\overline{v}_{x}) = \sum_{\substack{k: k+p\notin \mathcal{B}^{\sigma}_{F} \\ k\in \mathcal{B}^{\sigma}_{F}}} \hat a_{k+p,\sigma} \hat a_{k,\sigma}\;.
\end{equation}
The $b$, $b^{*}$ operators  turn out to behave as `bosonic' operators, if evaluated on states with few particles. To begin, notice that, denoting by $\delta_{k,k'}$ the Kronecker delta:
\begin{eqnarray}
[ \hat b_{p,\sigma}, \hat b^{*}_{q,\sigma'} ] = \delta_{p,q} \delta_{\sigma, \sigma'} |\mathcal{B}_{F}^{\sigma}| - \delta_{\sigma, \sigma'}\sum_{\substack{k,k'\in \mathcal{B}_{F}^{\sigma} \\ k+p, k'+q \notin \mathcal{B}_{F}^{\sigma}}} ( \hat a^{*}_{k'+q,\sigma'} \hat a_{k+p,\sigma} \delta_{k,k'} + \hat a^{*}_{k',\sigma'} \hat a_{k,\sigma} \delta_{k+p, k'+q}) 
\end{eqnarray}
and $[ \hat b_{p,\sigma}, \hat b_{q,\sigma'} ] = 0$. In particular, on states $\psi$ that contain `few' particles, $L^{-3}\langle \psi, \mathcal{N}\psi \rangle = o(\rho)$:
\begin{equation}\label{eq:CCR}
L^{-3}\langle \psi, [ \hat b_{p,\sigma}, \hat b_{q,\sigma'}^{*} ] \psi \rangle = \delta_{p,q}\delta_{\sigma\sigma'} \rho_{\sigma} + o(\rho)\;,\qquad \langle \psi, [ \hat b_{p,\sigma}, \hat b_{q,\sigma'} ] \psi \rangle = 0\;.
\end{equation}
Eqs. (\ref{eq:CCR}) suggest that, on states with `few' particles, the operators $b_{p,\sigma}$, $b^{*}_{q,\sigma'}$ satisfy approximate canonical commutation relations. Therefore, $\mathbb{Q}_{4}$ is quadratic on these pseudo-bosonic operators, which suggests that one might attempt to evaluate its energetic contribution to the ground state energy via diagonalization. Unfortunately, the $\mathbb{H}_{0}$, $\mathbb{Q}_{1}$ operators do not have this structure.  Nevertheless, if evaluated on a suitable class of states, they behave as quadratic bosonic operators, as we shall see below. For instance, consider $\mathbb{H}_{0}$. Let us rescale the $b$ operators so that they satisfy (approximate) canonical commutation relations, by setting $\hat c_{p,\sigma} = \rho_{\sigma}^{-1/2} \hat b_{p,\sigma}$. One has:
\begin{eqnarray}
[\mathbb{H}_{0}, \hat c^{*}_{q,\sigma}] &=& \rho_{\sigma}^{-\frac{1}{2}} \sum_{\substack{ k: k+q\notin \mathcal{B}^{\sigma}_{F} \\ k\in \mathcal{B}^{\sigma}_{F} }} ( | |k+q|^{2} - \mu_{\sigma} | + ||k|^{2} - \mu_{\sigma}| ) \hat a_{k+q,\sigma} \hat a_{k,\sigma}\nonumber\\
&=& \rho_{\sigma}^{-\frac{1}{2}}\sum_{\substack{ k: k+q\notin \mathcal{B}^{\sigma}_{F} \\ k\in \mathcal{B}^{\sigma}_{F} }} ( |k+q|^{2} - |k|^{2}) \hat a_{k+q,\sigma} \hat a_{k,\sigma}\;.
\end{eqnarray}
Being $k$ inside the Fermi ball, $|k|^{2}\leq C\rho_{\sigma}^{2/3}$. Thus, for $|q| \gg \rho_{\sigma}^{1/3}$, it makes sense to approximate:
\begin{equation}
[\mathbb{H}_{0}, c^{*}_{q,\sigma}] \simeq |q|^{2} c^{*}_{q,\sigma}\;.
\end{equation}
Let:
\begin{equation}
\mathbb{K}_{\text{B}} = \frac{1}{L^{3}}\sum_{p,\sigma} |p|^{2} \hat c^{*}_{p,\sigma}\hat c_{p,\sigma}\;.
\end{equation}
Considering the $c$ operators as true bosonic operators, we see that $\mathbb{K}_{\text{B}}$ satisfies the same (approximate) commutation relation as $\mathbb{H}_{0}$. This suggests that, on states with few bosons, the operators $\mathbb{H}_{0}$ and $\mathbb{K}_{\text{B}}$ act similarly. For instance, consider a state with one boson, $\hat c^{*}_{q,\sigma}\Omega$. Then:
\begin{equation}
\mathbb{H}_{0} \hat c^{*}_{q,\sigma} \Omega = [\mathbb{H}_{0}, \hat c^{*}_{q,\sigma}] \Omega \simeq [ \mathbb{K}_{\text{B}}, \hat c^{*}_{q,\sigma} ] \Omega = \mathbb{K}_{\text{B}} \hat c^{*}_{q,\sigma}\Omega\;.
\end{equation}
More generally, it is reasonable to expect that, on states $R^{*}\psi$ with few particles:
\begin{equation}
\langle R^{*}\psi, \mathbb{H}_{0} R^{*}\psi\rangle \simeq \langle R^{*}\psi, \mathbb{K}_{\text{B}} R^{*}\psi \rangle\;.
\end{equation}
Finally, consider now the $\mathbb{Q}_{1}$ operator. Again, $\mathbb{Q}_{1}$ is not quadratic in the pseudo-bosons. To understand its action in terms of pseudo-bosons, we rewrite it as:
\begin{eqnarray}\label{eq:Q4com}
\mathbb{Q}_{1} &=& \frac{1}{2}\sum_{\sigma, \sigma'} \int dxdy\, V(x-y) a_{\sigma}^{*}(u_{x}) a_{\sigma'}^{*}(u_{y})a_{\sigma'}(u_{y})a_{\sigma}(u_{x}) \\
&=& \frac{1}{2} \sum_{\sigma, \sigma'}\int dxdy\, V(x-y) a_{\sigma}^{*}(u_{x}) a_{\sigma}(u_{x}) a_{\sigma'}^{*}(u_{y}) a_{\sigma'}(u_{y})\nonumber\\
&& - \frac{1}{2}\sum_{\sigma, \sigma'} \int dxdy\, V(x-y) u_{\sigma}(y,x) \delta_{\sigma, \sigma'} a_{\sigma}^{*}(u_{x}) a_{\sigma'}(u_{y})\nonumber\\
&\equiv& \frac{1}{2L^{3}} \sum_{p} \hat V(p) D_{p} D_{-p} - \frac{1}{2L^{3}} \sum_{p} \hat V(p) E_{p}\;,\nonumber
\end{eqnarray}
where:
\begin{equation}
D_{p} = \sum_{\substack{ k: k\notin \mathcal{B}_{F}^{\sigma} \\ k - p \notin \mathcal{B}_{F}^{\sigma}}} \hat a^{*}_{k} \hat a_{k-p}\;, \qquad E_{p} = \sum_{\substack{ k: k\notin \mathcal{B}_{F}^{\sigma} \\ k - p \notin \mathcal{B}_{F}^{\sigma}}} \hat a^{*}_{k} \hat a_{k}\;.
\end{equation}
A simple computation shows that:
\begin{eqnarray}\label{eq:Q4com2}
[ D_{p}, \hat c^{*}_{q,\sigma'}] &=& \rho_{\sigma}^{-\frac{1}{2}}\sum_{\substack{ k: k\in \mathcal{B}^{\sigma}_{F} \\ k + q \notin \mathcal{B}^{\sigma}_{F} \\ k+q-p \notin \mathcal{B}^{\sigma}_{F}}} a^{*}_{k,\sigma} a^{*}_{k+q-p,\sigma'} \\
&\simeq& \rho_{\sigma}^{-\frac{1}{2}}\sum_{\substack{ k: k\in \mathcal{B}^{\sigma}_{F} \\ k+q-p \notin \mathcal{B}^{\sigma}_{F}}} a^{*}_{k,\sigma} a^{*}_{k+q-p,\sigma'} \equiv \hat c^{*}_{q-p}\;.\nonumber
\end{eqnarray}
In the last step we neglected the constraint $k+q\notin \mathcal{B}^{\sigma}_{F}$: this is reasonable, if $|q|\gg \rho_{\sigma}^{1/3}$. Thus, considering the $c$ operators as true bosons, we see that Eq. (\ref{eq:Q4com2}) is the same commutation relation satisfied by replacing $D_{p}$ with the operator:
\begin{equation}
\mathbb{G}_{p} = \frac{1}{L^{3}}\sum_{\sigma, q} \hat c^{*}_{q-p,\sigma} \hat c_{q,\sigma}\;.
\end{equation}
In the same spirit, it is not difficult to see that, for $|q|\gg \rho_{\sigma}^{1/3}$:
\begin{equation}
[ E_{p}, \hat c^{*}_{q,\sigma'} ] \simeq \hat c^{*}_{q,\sigma'}\;,
\end{equation}
which is the same commutation relation satisfied replacing $E_{p}$ by $\mathbb{G}_{0}$. All in all, we expect that, on states $R^{*}\psi$ with few pseudo-bosonic excitations particles, with momenta $|q|\gg \rho^{1/3}$:
\begin{eqnarray}
\langle \psi, \mathcal{H} \psi \rangle &\simeq& E_{\text{HF}}(\omega) + \frac{1}{L^{3}}\sum_{p,\sigma} |p|^{2} \langle R^{*}\psi, \hat c^{*}_{p,\sigma} \hat c_{p,\sigma} R^{*}\psi \rangle\nonumber\\&& + \frac{1}{2L^{3}}\sum_{p,\sigma,\sigma'} \rho_{\sigma}^{\frac{1}{2}} \rho_{\sigma'}^{\frac{1}{2}}\hat V(p) \langle R^{*}\psi, ( \hat c_{p,\sigma}\hat c_{-p,\sigma'} + \text{h.c.}) R^{*}\psi \rangle\nonumber\\
&& + \frac{1}{2L^{9}}\sum_{\substack{p, q, q' \\ \sigma, \sigma'}} \hat V(p) \langle R^{*}\psi, \hat c^{*}_{q-p, \sigma} \hat c_{q, \sigma} \hat c^{*}_{q'+p,\sigma'} \hat c_{q', \sigma'} R^{*}\psi \rangle\nonumber\\&& - \frac{1}{2 L^{6}} \sum_{p,q,\sigma} \hat V(p) \langle R^{*}\psi, \hat c^{*}_{q,\sigma} \hat c_{q,\sigma} R^{*}\psi \rangle\;.
\end{eqnarray}
Now, suppose that $R^{*}\psi = T\xi$, with $\xi$ `close' to the Fock space vacuum and:
\begin{equation}
T = \exp \Big\{\frac{1}{L^{3}}\sum_{p} \rho_{\uparrow}^{\frac{1}{2}} \rho_{\downarrow}^{\frac{1}{2}} \hat \varphi(p) \hat c_{p,\uparrow} \hat c_{-p,\downarrow} - \text{h.c.} \Big\}\;,
\end{equation}
for some even, real function $\hat \varphi(p)$, to be chosen in a moment. Treating the $c$ operators as true bosons, the operator $T$ implements a bosonic Bogoliubov transformation. It acts as:
\begin{equation}
T^{*} c_{q,\sigma} T = c_{q,\sigma} - \rho_{\sigma}^{\frac{1}{2}} \rho_{-\sigma}^{\frac{1}{2}} \hat \varphi(q) c^{*}_{-q,-\sigma} + o(\rho^{2})\;.
\end{equation}
The state $T\Omega$ is a bosonic quasi-free state, and its energy can be computed via the bosonic Wick rule. We have:
\begin{eqnarray}
\frac{1}{L^{3}}\langle T\Omega, \hat c^{*}_{p,\sigma} \hat c_{p,\sigma} T\Omega \rangle &=& \rho_{\sigma} \rho_{-\sigma} \hat \varphi(p)^{2} + o(\rho^{2}) \\ 
\frac{1}{L^{3}}\langle T\Omega, \hat c_{p,\sigma} \hat c_{-p,\sigma'} T\Omega \rangle &=& -\delta_{\sigma,-\sigma'} \rho_{\sigma}^{\frac{1}{2}} \rho_{-\sigma}^{\frac{1}{2}}\varphi(p) + o(\rho)\nonumber\\
\frac{1}{L^{6}} \langle T\Omega, \hat c^{*}_{q-p, \sigma} \hat c_{q, \sigma} \hat c^{*}_{q'+p,\sigma'} \hat c_{q', \sigma'} T\Omega \rangle &=& \delta_{p,0} \delta_{q,q'} \delta_{\sigma, \sigma'} \rho_{\sigma} \rho_{-\sigma}\varphi(q)^{2} + \delta_{q, q'+p} \delta_{\sigma, \sigma'} \rho_{\sigma} \rho_{-\sigma}\varphi(q)^{2} + o(\rho^{2})\;.\nonumber
\end{eqnarray}
Supposing that $R^{*}\psi \simeq T \Omega$ one has, neglecting all $o(\rho^{2})$ terms:
\begin{eqnarray}\label{eq:heue0}
\langle R^{*}\psi, \mathcal{H} R^{*}\psi \rangle &\simeq& E_{\text{HF}}(\omega) + \sum_{p,\sigma} \rho_{\sigma}\rho_{-\sigma} |p|^{2} \hat \varphi(p)^{2} - \sum_{p, \sigma} \rho_{\sigma} \rho_{-\sigma} \hat V(p) \hat \varphi(p)\nonumber\\
&& + \frac{1}{2L^{3}} \sum_{p,q, \sigma} \rho_{\sigma} \rho_{-\sigma} \hat V(p) (\hat \varphi(q-p) \hat \varphi(-q) + \hat \varphi(q-p)^{2}) \nonumber\\
&& - \frac{1}{2L^{3}} \sum_{p,q,\sigma} \rho_{\sigma}\rho_{-\sigma} \hat V(p) \hat \varphi(q)^{2} \equiv E_{\text{HF}}(\omega) + 2L^{3}\rho_{\uparrow}\rho_{\downarrow} e(\varphi)
\end{eqnarray}
where:
\begin{equation}\label{eq:heue}
e(\varphi) := \frac{1}{L^{3}}\sum_{p} \Big(|p|^{2} \hat \varphi(p)^{2}  - \hat V(p) \hat \varphi(p) + \frac{1}{2} (\hat V * \hat \varphi)(p) \hat \varphi(p)\Big)\;.
\end{equation}
We are interested in the value of $\varphi$ that gives the smallest possible energy. The equation for the minimizer is:
\begin{equation}\label{eq:scat}
2 |p|^{2} \hat \varphi(p) - \hat V(p) + (\hat V* \hat \varphi)(p) = 0\;,
\end{equation}
{\it i.e.} the zero energy scattering equation (in a periodic box). Plugging the solution of this equation in (\ref{eq:heue}) we get:
\begin{equation}
e(\varphi) = -\frac{1}{2L^{3}}\sum_{p} \hat V(p) \hat \varphi(p)\;;
\end{equation}
therefore, neglecting $o(\rho^{2})$ terms:
\begin{equation}
\frac{\langle R^{*}\psi, \mathcal{H} R^{*}\psi \rangle}{L^{3}} \simeq \frac{E_{\text{HF}}(\omega)}{L^{3}} - \rho_{\uparrow} \rho_{\downarrow} \frac{1}{L^{3}}\sum_{p}\hat V(p) \hat \varphi(p)\;.
\end{equation}
In conclusion, recalling the expression (\ref{eq:enHF}) for $E_{\text{HF}}(\omega)$, for $L$ large enough:
\begin{equation}
\frac{\langle R^{*}\psi, \mathcal{H} R^{*}\psi \rangle}{L^{3}} = \frac{3}{5}(6\pi^{2})^{\frac{2}{3}} (\rho_{\uparrow}^{\frac{5}{3}} + \rho_{\downarrow}^{\frac{5}{3}}) + \rho_{\uparrow}\rho_{\downarrow} \Big( \hat V(0) - \frac{1}{L^{3}} \sum_{p} \hat V(p) \hat \varphi(p) \Big) + o(\rho^{2})\;.
\end{equation}
The term in parenthesis can be rewritten as, as $L\to \infty$:
\begin{equation}
\int dx\, V(x) (1 - \varphi(x)) \equiv 8\pi a\;,
\end{equation}
where $a$ is the scattering length of the potential $V$. This reproduces the final result (\ref{eq:main}). Even at the heuristic level, however, there is a problem: the operators $\mathbb{H}_{0}$, $\mathbb{Q}_{1}$, $\mathbb{Q}_{4}$ can be represented in terms of bosons {\it provided} they act on states with few bosons, with momenta $|q|\gg \rho^{1/3}$. The state $T\Omega$ is given by a superposition of states by an even number of bosons, with momenta in the support of $\hat \varphi(p)$. To enforce the momentum contraint, we would like the function $\hat \varphi(p)$ to be supported for $|p|\gg \rho^{1/3}$. Equivalently, we would like to regularize $\varphi(x)$, so that it is supported on a ball of radius $1\ll R\ll \rho^{-1/3}$. Let $\varphi_{\infty}$ be the solution of the scattering equation in a ball $B\equiv B_{\rho^{-\gamma}}(0) \subset \mathbb{R}^{3}$ centered at zero, with radius $\rho^{-\gamma}$ with $0\leq \gamma \leq 1/3$ and Neumann boundary conditions (see Appendix \ref{sec:scat}):
\begin{equation}\label{eq:scat2}
-\Delta (1- \varphi_{\infty}) + \frac{1}{2}V_{\infty}(1 - \varphi_{\infty}) = \lambda_{\gamma}(1-\varphi_{\infty})\;,\qquad \varphi_{\infty} = \nabla \varphi_{\infty} = 0\quad \text{on $\partial B$,}
\end{equation}
with $|\lambda_{\gamma}| \leq C\rho^{3\gamma}$. For $x$ away from the support of $V_{\infty}$, the solution of (\ref{eq:scat2}) behaves as:
\begin{equation}
\varphi_{\infty}(x) \sim \frac{a_{\gamma}}{|x|}\;,\qquad 8\pi a_{\gamma} = \int dx\, V_{\infty}(x) (1 - \varphi_{\infty}(x))\;.
\end{equation}
The function $\varphi_{\infty}$ is extended to $\mathbb{R}^{3}$ by setting $\varphi_{\infty}(x) = 0$ for $x\notin B$. To make it compatible with the periodic boundary conditions on $\Lambda_{L}$, we shall take $\varphi$ as the periodization of $\varphi_{\infty}$, 
\begin{equation}\label{eq:phiper}
\varphi(x) = \sum_{n\in \mathbb{Z}^{3}} \varphi_{\infty}(x + n_{1} e_{1} L + n_{2} e_{2} L + n_{3} e_{3} L)\;.
\end{equation}
Equivalently:
\begin{equation}\label{eq:per}
\varphi(x) = \frac{1}{L^{3}} \sum_{p\in \frac{2\pi}{L}\mathbb{Z}^{3}} e^{ip\cdot x} \hat \varphi_{\infty}(p)\;,
\end{equation}
where $\hat \varphi_{\infty}(p) = \int_{\mathbb{R}^{3}} dx\, e^{-ip\cdot x} \varphi_{\infty}(x)$. Plugging the solution of this equation in $e(\varphi)$, one obtains, as $L\to \infty$, using that $\varphi(x) \to  \varphi_{\infty}(x)$ pointwise:
\begin{equation}
e(\varphi) = -\frac{1}{2} \int \frac{dp}{(2\pi)^{3}}\, \hat V_{\infty}(p) \hat \varphi_{\infty}(p) + \rho^{3\gamma} \int \frac{dp}{(2\pi)^{3}}\,\hat \varphi_{\infty}(p) (1 - \hat \varphi_{\infty}(p))\;.
\end{equation}
It is well-known that $|a - a_{\gamma}| \leq C\rho^{\gamma}$, see Lemma \ref{lem:scat}. Thus, the first term combined with the interaction energy of the free Fermi gas reproduces the scattering length, while the last term is bounded by:
\begin{equation}
\Big|\rho^{3\gamma} \int \frac{dp}{(2\pi)^{3}}\, \hat \varphi_{\infty}(p) (1 - \hat \varphi_{\infty}(p))\Big| \leq \rho^{3\gamma} \varphi_{\infty}(0) + \rho^{3\gamma} \| \varphi_{\infty} \|_{2}^{2} \leq C\rho^{2\gamma}\;,
\end{equation}
where we used that $\varphi_{\infty}(x) \sim |x|^{-1}$ for large $|x|$, and that $\varphi_{\infty}(x)$ is compactly supported. This amounts to a small correction to the ground state energy, which does not affect the $\rho^{2}$ term.
\begin{remark}
\begin{itemize}
\item[(i)] As the above heuristics suggests, the choice $\gamma = 1/3$ is expected to be the correct one in order to compute the ground state energy density with a $O(\rho^{7/3})$ precision; this is the order of magnitude of the next correction to the ground state energy after $8\pi a\rho_{\uparrow} \rho_{\downarrow}$, \cite{HY}.
\item[(ii)] As mentioned in the introduction, similar bosonization ideas have been recently used in order to prove the validity of the random phase approximation, for the ground state energy of interacting fermionic systems in the mean-field regime \cite{BNPSS, BNPSS2}. There, the emergent bosonic degrees of freedom correspond to particle-hole excitations that are localized around suitable `patches' on the Fermi surface (whose radius grows proportionally to $N^{1/3}$), around which the kinetic energy operator $\mathbb{H}_{0}$ can be approximated by a linear dispersion for the bosonic modes.
\end{itemize}
\end{remark}

\subsection{Definition of the correlation structure}\label{sec:cor}
Here we shall give a precise definition of the unitary operator $T$, introduced in the previous section. Before doing this, let us fix some notation, that will be used in the rest of the paper.
\medskip

\noindent{\bf Notations.}
\begin{itemize}
\item We shall denote by $\chi: \mathbb{R}_{+} \to [0,1]$ a smooth, nonincreasing cutoff function such that $\chi(t) = 1$ for $t\leq 1$ and $\chi(t) = 0$ for $t\geq 2$. We shall also use the notation $\chi^{c} = 1 - \chi$.
\item We shall denote by $C$ a general constant, possibly dependent on $V$, whose value might change from line to line.
\item We shall denote by $C_{\beta}$ a general prefactor of the form $C \rho^{-c\beta}$. We will not keep track of such $\rho$-dependence of the bounds. At the end, by taking the interaction potential regular enough, we will be able to consider $\beta > 0$ arbitrarily small.
\item Unless otherwise stated, we shall use the notation $\|\cdot \|_{p}$ for $\|\cdot \|_{L^{p}(\Lambda_{L})}$.
\item We shall use the notations $\sum_{\sigma}$ for $\sum_{\sigma = \uparrow\downarrow}$, $\int dx$ for $\int_{\Lambda_{L}}dx$ and $\sum_{k}$ for $\sum_{k\in \frac{2\pi}{L}\mathbb{Z}^{3}}$.
\item We shall denote by $|\cdot|$ the usual Euclidean distance on $\mathbb{R}^{3}$, and by $|\cdot|_{L}$ the distance on the torus: $|x - y|_{L} = \min_{n\in \mathbb{Z}^{3}} | x - y - n_{1} e_{1} - n_{2}e_{2} - n_{3} e_{3}  |$, with $\{e_{i}\}$ the standard orthonormal basis of $\mathbb{R}^{3}$.
\item We shall use the notation $u_{x}$, $v_{x}$ to denote the functions $y\mapsto u_{x}(y) \equiv u(x;y)$, $y\mapsto v_{x}(y)\equiv v(x;y)$.
\item We shall denote by $\frak{e}_{L}$ a general finite size correction, subleading with respect to $L^{3}$.
\end{itemize}
Let us introduce regularized versions of the operators $u$ and of $v$, introduced in Section \ref{sec:bogHcorr}. We define:
\begin{equation}\label{eq:reguv}
\vr_{\sigma, \sigma'} = \frac{\delta_{\sigma, \sigma'}}{L^3} \sum_{k} \hat v_{\sigma}^{\text{r}}(k) |\overline{f_{k}}\rangle \langle f_{k}|\;,\qquad \ur_{\sigma,\sigma'} =  \frac{\delta_{\sigma,\sigma'}}{L^{3}} \sum_{k} \hat u_{\sigma}^{\text{r}}(k) |f_{k}\rangle \langle f_{k}|\;,
\end{equation}
where $\hat v_{\sigma}^{\text{r}}(k)$, $\hat u_{\sigma}^{\text{r}}(k)$ are smooth, radial functions with the following properties. Let $\alpha = \frac{1}{3} + \frac{\epsilon}{3}$, with $\epsilon>0$, and let $\beta>0$, to be chosen later on:
\begin{equation}\label{eq:urvrdef}
\hat v_{\sigma}^{\text{r}}(k) = \left\{ \begin{array}{cc} 1 & \text{for $|k| < k_{F}^{\sigma} - \rho_{\sigma}^{\alpha}$} \\ 0 & \text{for $|k|\geq k_{F}^{\sigma}$} \end{array} \right.\qquad \hat u^{\text{r}}_{\sigma}(k) = \left\{ \begin{array}{cc} 0 & \text{for $|k|\leq k^{\sigma}_{F}$} \\ 1 & \text{for $2k^{\sigma}_{F}\leq |k|\leq \frac{3}{2}\rho_{\sigma}^{-\beta}$} \\ 0 & \text{for $|k| \geq 2\rho_{\sigma}^{-\beta}$.} \end{array} \right.
\end{equation}
Concretely, we choose:
\begin{equation}\label{eq:uvdef0}
\hat v_{\sigma}^{\text{r}}(k) = \chi\Big( \frac{|k| - (k_{F}^{\sigma} - 2\rho_{\sigma}^{\alpha})}{\rho_{\sigma}^{\alpha}} \Big)\;,\qquad \hat u^{\text{r}}_{\sigma}(k) = \chi(\rho_{\sigma}^{\beta} |k|) \chi^{c}\Big( \frac{|k|}{k_{F}^{\sigma}} \Big)\;.
\end{equation}
Notice that these regularized operators preserve the important orthogonality relation in Eq. (\ref{eq:uvprop}):
\begin{equation}
u^{\text{r}} \overline{v}^{\text{r}} = 0\;.
\end{equation}
We shall also denote by $\omegar$ the regularized version of the $\omega$, $\omegar = \overline{\vr} \vr$. The next proposition collects some useful bounds for these regularized objects.
\begin{proposition}[Bounds for the regularized kernels.]\label{prp:decreg} The following estimates hold true, for all $n\in \mathbb{N}$ and for $L$ large enough:
\begin{equation}\label{eq:decuv}
\| u_{x,\sigma}^{\text{r}} \|_{2} \leq C \rho_{\sigma}^{-\frac{3\beta}{2}}\;,\qquad \| v_{x,\sigma}^{\text{r}} \|_{2} \leq \rho_{\sigma}^{\frac{1}{2}}\;,\qquad \| \omega^{\text{r}}_{x,\sigma} \|_{1} \leq C\rho_{\sigma}^{-\frac{\epsilon}{3}}\;,\qquad \| u^{\text{r}}_{x,\sigma} \|_{1} \leq C\;.
\end{equation}
\end{proposition}
\begin{proof} The first two estimates of (\ref{eq:decuv}) easily follow from (\ref{eq:urvrdef}). Consider now the last two. Let us prove the estimate for $\omega^{\text{r}}(x,y)$ (we shall omit the spin label for simplicity). Let $\omega_{\infty}^{\text{r}}(x,y)$ be the infinite volume limit of $\omega^{\text{r}}(x,y)$. We have, performing the angular integration:
\begin{eqnarray}
\omega_{\infty}^{\text{r}}(x,y) &=& \int d^{3}k\, \hat \omega^{\text{r}}(k) e^{ik\cdot (x-y)} \nonumber\\
&=& \frac{4\pi}{|x-y|} \int dt\, t \hat \omega^{\text{r}}(t) \sin (t |x-y|)\;,
\end{eqnarray}
where we used that $\hat \omega^{\text{r}}(k)\equiv \hat \omega^{\text{r}}(|k|)$. Recall that $\hat \omega^{\text{r}}(|k|)$ is a smooth, compactly supported function, given by $(\hat v^{\text{r}}(k))^{2}$, Eq. (\ref{eq:uvdef0}). Using that $\sin (t |x-y|) = -|x-y|^{-1} \partial_{t} \cos (t|x-y|)$ and integrating by parts, we get:
\begin{equation}\label{eq:xy2}
\omega_{\infty}^{\text{r}}(x,y) = \frac{4\pi}{|x-y|^{2}} \int dt\, ( \hat \omega^{\text{r}}(t) + t \partial_{t} \hat \omega^{\text{r}}(t) ) \cos (t |x-y|)\;.
\end{equation}
That is:
\begin{equation}\label{eq:esto1}
|\omega_{\infty}^{\text{r}}(x,y)| \leq \frac{C\rho^{\frac{1}{3}}}{|x-y|^{2}}\;.
\end{equation}
Integrating by parts two more times, we get:
\begin{eqnarray}\label{eq:esto2}
| \omega^{\text{r}}_{\infty}(x;y) | &\leq& \frac{C}{|x-y|^{4}} \int dt\, ( | \partial_{t}^{2} \hat \omega^{\text{r}}(t)| + t |\partial_{t}^{3} \hat \omega^{\text{r}}(t)  | ) \nonumber\\
&\leq& \frac{C\rho^{\frac{1}{3} - 2\alpha}}{|x-y|^{4}}\;;
\end{eqnarray}
we used that $\alpha \geq 1/3$, to conclude that the last term in the right-hand side of the first line dominates over the first. Therefore, putting together (\ref{eq:esto1}), (\ref{eq:esto2}):
\begin{eqnarray}\label{eq:13a}
\| \omega^{\text{r}}_{\infty,x} \|_{1} &\leq& \int_{|y|\leq \rho^{-\alpha}} dy\,  \frac{C\rho^{\frac{1}{3}}}{|y|^{2}} + \int_{|y| > \rho^{-\alpha}} dy\,\frac{C\rho^{\frac{1}{3} - 2\alpha}}{|y|^{4}} \nonumber\\
&\leq& C\rho^{\frac{1}{3} - \alpha}  \equiv C\rho^{-\frac{\epsilon}{3}}\;,
\end{eqnarray}
recall that $\alpha = \frac{1}{3} + \frac{\epsilon}{3}$. To prove the estimate for $\| \omega^{\text{r}}_{x} \|_{1}$, we write, for $n\geq 0$ large enough, independent of $L$, and for $L$ large enough:
\begin{eqnarray}
\| \omega^{\text{r}}_{x} \|_{1} &\leq& \| \omega^{\text{r}}_{x} \chi(|\cdot|_{L} \leq \rho^{-n}) \|_{1} + \| \omega^{\text{r}}_{x} \chi(|\cdot|_{L} > \rho^{-n}) \|_{1} \nonumber\\
&\leq& \| \omega^{\text{r}}_{\infty, x} \chi(|\cdot|_{L} \leq \rho^{-n}) \|_{1} + \| (\omega^{\text{r}}_{x} - \omega^{\text{r}}_{\infty,x})  \chi(|\cdot|_{L} \leq \rho^{-n}) \|_{1} + \| \omega^{\text{r}}_{x} \chi(|\cdot|_{L} > \rho^{-n}) \|_{1} \nonumber\\
&\leq& C\rho^{-\frac{\epsilon}{3}}\;.
\end{eqnarray}
To prove the last inequality, we used that, for any fixed $x$, $|\omega^{\text{r}}_{\infty}(x) - \omega^{\text{r}}(x)| \leq C/L$. Also, we used that, by the smoothness of $\hat \omega^{\text{r}}$, $\omega^{\text{r}}(x;y)$ decays faster than any power in $|x-y|_{L}$ (nonuniformly in $\rho$), which allows to control the last term in the second line. The proof of the last estimate in Eq. (\ref{eq:decuv}) is completely analogous, and we shall omit the details; the reason for the uniform bound in $\rho$ is that the infrared part of $\hat u^{\text{r}}(k)$ is smoothened on scale $\rho^{\frac{1}{3}}$ instead of $\rho^{\frac{1}{3} + \frac{\epsilon}{3}}$.
\end{proof}
As hinted by the heuristic discussion in Section \ref{sec:heu}, in the following an important role will be played by the solution of the scattering equation (\ref{eq:scat2}). We shall denote by $\varphi$ the periodization of $\varphi_{\infty}$ over $\Lambda_{L}$, recall Eq. (\ref{eq:per}).
\begin{definition}[The correlation structure] Let $\gamma \geq 0$, $\lambda \in [0,1]$. We define the unitary operator $T_{\lambda}: \mathcal{F} \to \mathcal{F}$ as:
\begin{equation}\label{eq:defT}
T_{\lambda} := e^{\lambda (B - B^{*})}\;,\qquad B := \int dzdz'\, \varphi(z-z') a_{\uparrow}(u^{\text{r}}_{z}) a_{\uparrow}(\overline{v}^{\text{r}}_{z}) a_{\downarrow}(u^{\text{r}}_{z'}) a_{\downarrow}(\overline{v}^{\text{r}}_{z'})\;.
\end{equation}
We shall also set $T \equiv T_{1}$.
\end{definition}
The operator $T$ is a regularized version of the operator introduced in Section \ref{sec:heu}. The reason for the smoothing of $u^{\text{r}}$, $v^{\text{r}}$ is that the function $|\omega(x-y)|$ is not integrable uniformly in $L$. As we shall see, to bound the errors neglected in the analysis of Section \ref{sec:heu} we will have to control volume integrations, which will be possible thanks to the improved decay of $|\omega^{\text{r}}(x;y)|$, $|u^{\text{r}}(x;y)|$. Also, the ultraviolet cutoff in the definition of $u^{\text{r}}$ makes the operators involved in the definition of $T$ bounded. By (\ref{eq:bdferm}), recalling (\ref{eq:decuv}):
\begin{equation}
\| a(u^{\text{r}}_{x}) \| \leq \| u^{\text{r}}_{x} \|_{2} \leq C\rho^{-\frac{3\beta}{2}} \leq C_{\beta}\;,\qquad \| a(\overline{v}^{\text{r}}_{x}) \| \leq \| \overline{v}^{\text{r}}_{x} \|_{2} \leq C\rho^{\frac{1}{2}}\;.
\end{equation}
Finally, notice that the $B$ operator is extensive: this follows from the fact that the $z, z'$ variables satisfy $|z - z'|_{L} \leq \rho^{-\gamma}$, due to the compact support of $\varphi$ in $\Lambda_{L}$. In particular, thanks to the estimate $\|\varphi_{\infty}\|_{L^{1}(\mathbb{R}^{3})} \leq C\rho^{-2\gamma}$, see Appendix \ref{sec:scat}, and recalling (\ref{eq:phiper}), we have:
\begin{equation}
\| \varphi \|_{L^{1}(\Lambda_{L})} \leq C\| \varphi_{\infty} \|_{L^{1}(\mathbb{R}^{3})} \leq C\rho^{-2\gamma}\;.
\end{equation}
As it will be clear at the end of our analysis, all these regularizations will not affect the computation of the ground state energy density at order $\rho^{2}$. They will however determine the size of the subleading error terms.
\section{Bounds on interpolating states}\label{sec:prop}
\subsection{Introduction}
In this section we shall propagate the a priori bounds on $\mathcal{N}$, $\mathbb{H}_{0}$, $\mathbb{Q}_{1}$ proven in Section \ref{sec:bogHcorr}, over the states
\begin{equation}
\xi_{\lambda} := T^{*}_{\lambda} R^{*}\psi\;,
\end{equation}
with $\psi$ being an approximate ground state, in the sense of Definition \ref{def:appgs}. One of the main results of the section will be that, for $5/18\leq \gamma \leq 1/3$, the following bounds hold true. Let $\lambda \in [0,1]$. Then:
\begin{equation}\label{eq:propapri}
\langle \xi_{\lambda}, \mathcal{N} \xi_{\lambda} \rangle \leq C L^{\frac{3}{2}} \rho^{\frac{1}{6}} \| \mathbb{H}_{0}^{\frac{1}{2}} \xi_{1} \| + CL^{3} \rho^{2 - \gamma}\;,\quad \langle \xi_{\lambda}, \mathbb{H}_{0} \xi_{\lambda} \rangle \leq CL^{3} \rho^{2}\;,\quad \langle \xi_{\lambda}, \mathbb{Q}_{1} \xi_{\lambda} \rangle \leq CL^{3} \rho^{2}\;.
\end{equation}
In particular, these bounds show that the a priori estimates of $\mathbb{H}_{0}$, $\mathbb{Q}_{1}$ on $\xi_{0} = R^{*}\psi$, recall (\ref{eq:aprH0}), (\ref{eq:aprQ1}), do not deterioriate over $\xi_{\lambda} = T^{*}_{\lambda} R^{*} \psi$, for $\lambda \in [0,1]$. Along the way, we shall prove a number of auxiliary results, that will play an important role in the computation of the ground state energy of the system, in Sections \ref{sec:lwbd}, \ref{sec:upper}. More precisely, we will prove that:
\begin{equation}
\frac{d}{d\lambda} \langle \xi_{\lambda}, \mathbb{H}_{0} \xi_{\lambda} \rangle = -\langle \xi_{\lambda}, \mathbb{T}_{1} \xi_{\lambda}\rangle + \mathcal{E}_{\mathbb{H}_{0}}(\xi_{\lambda})\;,\qquad \frac{d}{d\lambda} \langle \xi_{\lambda}, \mathbb{Q}_{1} \xi_{\lambda} \rangle = -\langle \xi_{\lambda}, \mathbb{T}_{2} \xi_{\lambda}\rangle + \mathcal{E}_{\mathbb{Q}_{1}}(\xi_{\lambda})
\end{equation}
where $ \mathcal{E}_{\mathbb{H}_{0}}(\xi_{\lambda})$, $\mathcal{E}_{\mathbb{Q}_{1}}(\xi_{\lambda})$ are two error terms, subleading with respect to $L^{3} \rho^{2}$, and $\mathbb{T}_{1}$, $\mathbb{T}_{2}$ are given by:
\begin{eqnarray}
\mathbb{T}_{1} &=& -\int dxdy\, \theta(|x-y|_{L} < \rho^{-\gamma})(-2\Delta \varphi)(x-y) a_{\uparrow}(u^{\text{r}}_{x}) a_{\uparrow}(\overline{v}^{\text{r}}_{x}) a_{\downarrow}(u^{\text{r}}_{y}) a_{\downarrow}(\overline{v}^{\text{r}}_{y}) + \text{h.c.}\nonumber\\
\mathbb{T}_{2} &=& -\int dxdy\, V(x-y) \varphi(x-y) a_{\uparrow}(\overline{v}^{\text{r}}_{x}) a_{\uparrow}(u_{x}) a_{\downarrow}(\overline{v}^{\text{r}}_{y}) a_{\downarrow}(u_{y})  + \text{h.c.}\;,
\end{eqnarray}
with $\theta(\cdot)$ the characteristic function. The important point to notice here is that $\mathbb{T}_{1}$, $\mathbb{T}_{2}$ have the same `bosonic' structure as $\widetilde{\mathbb{Q}}_{4}$. In particular, we will prove the following cancellation, using the fact that the function $1 - \varphi$ solves the scattering equation:
\begin{equation}\label{eq:scateqcanc}
\langle \xi_{\lambda}, (\mathbb{T}_{1} + \mathbb{T}_{2} + \widetilde{\mathbb{Q}}_{4}) \xi_{\lambda} \rangle = o(L^{3}\rho^{2})\;.
\end{equation}
The results (\ref{eq:propapri})-(\ref{eq:scateqcanc}) are the main technical ingredients needed in order to prove our main result, Theorem \ref{thm:main}, in Sections \ref{sec:lwbd}, \ref{sec:upper}.
\subsection{Propagation of the estimates: preliminaries}
To begin, let us start by propagating the a priori estimate for $\mathcal{N}$. The propagation estimate we shall obtain below is not optimal; it will be improved at the end of the section. Nevertheless, Proposition \ref{prp:propN} will be enough to propagate the a priori bounds for $\mathbb{H}_{0}$ and for $\mathbb{Q}_{1}$, which is our first task.
\begin{proposition}[Propagation estimate for $\mathcal{N}$.]\label{prp:propN} Let $\psi \in \mathcal{F}$ such that $\|\psi\| = 1$. Let $\gamma \leq 1/2$. Then, the following bound holds, for $\lambda \in [0,1]$:
\begin{equation}\label{eq:propN}
|\partial_{\lambda}\langle \xi_{\lambda}, \mathcal{N} \xi_{\lambda} \rangle| \leq C\langle \xi_{\lambda}, \mathcal{N}  \xi_{\lambda}  \rangle + C L^3 \rho^{2- \gamma}\;.
\end{equation} 
\end{proposition}
\begin{corollary} Let $\psi$ be an approximate ground state. Then:
\begin{equation}\label{eq:prop76}
\langle \xi_{\lambda}, \mathcal{N} \xi_{\lambda}  \rangle \leq CL^{3} \rho^{\frac{7}{6}}\;.
\end{equation}
\end{corollary}
\begin{proof} Eq. (\ref{eq:prop76}) immediately follows from (\ref{eq:propN}) and Gr\"onwall lemma, recalling the a priori estimate $\langle \xi_{0}, \mathcal{N} \xi_{0} \rangle \leq CL^{3} \rho^{\frac{7}{6}}$, Eq. (\ref{eq:bdcalN}).
\end{proof}
Let us now discuss the proof of Proposition \ref{prp:propN}. The proof is based on the following lemma.
\begin{lemma}\label{lem:bosonbd} Let $g\in L^{1}(\Lambda_{L}) \cap L^{2}(\Lambda_{L})$, and let:
\begin{equation}\label{eq:bphi}
b_{z,\sigma} = a_{\sigma}(u^{\text{r}}_{z}) a_{\sigma}(\overline{v}^{\text{r}}_{z})\;,\qquad b_{\sigma}(g_{z}) = \int dz'\, \overline{g(z-z')} b_{z',\sigma}\;.
\end{equation}
Then:
\begin{equation}\label{eq:bnorms}
\| b_{\sigma}(g_{z}) \| \leq C_{\beta}\rho^{\frac{1}{2}}\|g\|_{1}\;,\qquad \| b^{*}_{\sigma}(g_{z}) \| \leq C_{\beta}\rho^{\frac{1}{2}}\|g\|_{1}
\end{equation}
and:
\begin{equation}\label{eq:estbosb}
\| b^{*}_{\sigma}(g_{z}) \psi \| \leq \| b_{\sigma}(g_{z}) \psi \| + C\rho^{\frac{1}{2}}\|g\|_{2} \|\psi\|\;.
\end{equation}
\end{lemma}
\begin{proof} Consider the first of (\ref{eq:bnorms}), the proof of the second is identical. The inequality simply follows from the boundedness of the fermionic operators:
\begin{eqnarray}
\| b_{\sigma}(g_{z})\psi \| &\leq& \int dz' |g(z-z')| \| a_{\sigma}(u^{\text{r}}_{z'}) a_{\sigma}(\overline{v}^{\text{r}}_{z'}) \psi \| \nonumber\\
&\leq& \int dz' |g(z-z')| \|u^{\text{r}}_{z'}\|_{2} \| \overline{v}^{\text{r}}_{z'} \|_{2} \|\psi\| \nonumber\\
&\leq& C_{\beta} \|g\|_{1} \rho^{\frac{1}{2}} \|\psi\|\;.
\end{eqnarray}
Let us now prove (\ref{eq:estbosb}). It is convenient to rewrite the first norm as:
\begin{equation}\label{eq:bosonbound}
\| b^{*}_{\sigma}(g_{z}) \xi_{\lambda}\|^{2} = \| b_{\sigma}(g_{z}) \xi_{\lambda}\|^{2} - \langle \xi_{\lambda}, [ b^{*}_{\sigma}(g_{z}), b_{\sigma}(g_{z}) ]\xi_{\lambda} \rangle\;.
\end{equation}
Let us compute the commutator. We get:
\begin{eqnarray}
&&[ b^{*}_{\sigma}(g_{z}), b_{\sigma}(g_{z}) ] = \int dxdy\,g(z-x) \overline{g(z-y)} [  a^{*}_{\sigma}(\overline{v}^{\text{r}}_{x})a^{*}_{\sigma}(u^{\text{r}}_{x}), a_{\sigma}(u^{\text{r}}_{y}) a_{\sigma}(\overline{v}^{\text{r}}_{y}) ] \nonumber\\
&&\qquad = \int dxdy\,g(z-x) \overline{g(z-y)} \Big( a^{*}_{\sigma}(\overline{v}^{\text{r}}_{x}) [ a^{*}_{\sigma}(u^{\text{r}}_{x}), a_{\sigma}(u^{\text{r}}_{y}) a_{\sigma}(\overline{v}^{\text{r}}_{y}) ] + [  a^{*}_{\sigma}(\overline{v}^{\text{r}}_{x}), a_{\sigma}(u^{\text{r}}_{y}) a_{\sigma}(\overline{v}^{\text{r}}_{y}) ] a^{*}_{\sigma}(u^{\text{r}}_{x})\Big)\nonumber\\
&&\qquad = \int dxdy\,g(z-x) \overline{g(z-y)} \Big( a^{*}_{\sigma}(\overline{v}^{\text{r}}_{x}) (u^{\text{r}})^{2}_{\sigma}(x;y) a_{\sigma}(\overline{v}^{\text{r}}_{y}) - a_{\sigma}(u^{\text{r}}_{y}) \omega^{\text{r}}_{\sigma}(x;y) a^{*}_{\sigma}(u^{\text{r}}_{x}) \Big)\;.
\end{eqnarray}
Putting the last term into normal order, we get:
\begin{eqnarray}
[ b^{*}_{\sigma}(g_{z}), b_{\sigma}(g_{z}) ] &=& -\int dxdy\,g(z-x) \overline{g(z-y)} \omega^{\text{r}}_{\sigma}(x;y) (u^{\text{r}})^{2}_{\sigma}(y;x) \\
&& + \int dxdy\,g(z-x) \overline{g(z-y)} \Big( a^{*}_{\sigma}(\overline{v}^{\text{r}}_{x}) (u^{\text{r}})^{2}_{\sigma}(x;y) a_{\sigma}(\overline{v}^{\text{r}}_{y}) + a^{*}_{\sigma}(u^{\text{r}}_{x}) \omega^{\text{r}}_{\sigma}(x;y)  a_{\sigma}(u^{\text{r}}_{y})\Big)\;.\nonumber
\end{eqnarray}
The last two terms are nonnegative. In fact:
\begin{eqnarray}
&&\int dxdy\,g(z-x) \overline{g(z-y)} \langle \xi_{\lambda}, a^{*}_{\sigma}(\overline{v}^{\text{r}}_{x}) (u^{\text{r}})^{2}_{\sigma}(x;y) a_{\sigma}(\overline{v}^{\text{r}}_{y}) \xi_{\lambda}\rangle\nonumber\\
&&\qquad = \int dr\, \Big\langle \xi_{\lambda}, \Big( \int dx\, g(z-x) a^{*}_{\sigma}(\overline{v}^{\text{r}}_{x}) u_{\sigma}^{\text{r}}(x;r) \Big)\Big( \int dy\, \overline{g(z-y)} a_{\sigma}(\overline{v}^{\text{r}}_{y}) u_{\sigma}^{\text{r}}(r;y) \Big) \xi_{\lambda} \Big\rangle\nonumber\\
&&\qquad \geq 0\;,
\end{eqnarray}
and similarly:
\begin{eqnarray}
&&\int dxdy\, g(z-x) \overline{g(z-y)} \langle \xi_{\lambda}, a^{*}_{\sigma}(u^{\text{r}}_{x}) \omega^{\text{r}}_{\sigma}(x;y)  a_{\sigma}(u^{\text{r}}_{y}) \xi_{\lambda}\rangle \nonumber\\
&&\qquad = \int dr\, \Big\langle \xi_{\lambda}, \Big( \int dx\, g(z-x) a^{*}_{\sigma}(u^{\text{r}}_{x}) \overline{v}_{\sigma}^{\text{r}}(x;r) \Big)\Big( \int dy\, \overline{g(z-y)} a_{\sigma}(u^{\text{r}}_{y}) v_{\sigma}^{\text{r}}(r;y) \Big) \xi_{\lambda} \Big\rangle\nonumber\\
&&\qquad \geq 0\;.
\end{eqnarray}
Therefore:
\begin{equation}\label{eq:bbos1}
\| b^{*}_{\uparrow}(g_{z}) \xi_{\lambda}\|^{2} \leq \| b_{\uparrow}(g_{z}) \xi_{\lambda}\|^{2} + \int dxdy\,g(z-x) \overline{g(z-y)} \omega^{\text{r}}_{\sigma}(x;y) (u^{\text{r}})^{2}_{\sigma}(y;x)\;.
\end{equation}
The last term can be estimated as, using that $(u^{\text{r}})^{2}$ satisfies the same estimates as $u^{\text{r}}$, Eqs. (\ref{eq:decuv}):
\begin{eqnarray}\label{eq:bbos2}
\Big|\int dxdy\,g(z-x) \overline{g(z-y)} \omega^{\text{r}}_{\sigma}(x;y) (u^{\text{r}})^{2}_{\sigma}(y;x) \Big| &\leq& \rho \int dxdy\, |g(z-x)|^{2} |(u^{\text{r}})^{2}_{\sigma}(y;x)| \nonumber\\
&\leq& C\rho\|g\|_{2}^{2}\;.
\end{eqnarray}
In conclusion:
\begin{equation}
\| b^{*}_{\uparrow}(g_{z}) \xi_{\lambda}\|^{2} \leq \| b_{\uparrow}(g_{z}) \xi_{\lambda}\|^{2} + C\rho\|g\|_{2}^{2}\;.
\end{equation}
This concludes the proof.
\end{proof}
We are now ready to prove Proposition \ref{prp:propN}.
\begin{proof}(of Proposition \ref{prp:propN}.) We compute:
\begin{eqnarray}\label{eq:partN}
\partial_{\lambda} \langle \xi_{\lambda}, \mathcal{N} \xi_{\lambda} \rangle &=& -\langle \xi_{\lambda}, [ \mathcal{N}, (B - B^{*}) ] \xi_{\lambda} \rangle \\
&=& -4\int dzdz'\, \varphi(z-z') \langle \xi_{\lambda}, a_{\uparrow}(u^{\text{r}}_{z}) a_{\uparrow}(\overline{v}^{\text{r}}_{z}) a_{\downarrow}(u^{\text{r}}_{z'}) a_{\downarrow}(\overline{v}^{\text{r}}_{z'}) \xi_{\lambda} \rangle + \text{c.c..}\nonumber
\end{eqnarray}
Using the notation (\ref{eq:bphi}), we rewrite (\ref{eq:partN}) as:
\begin{equation}
\partial_{\lambda} \langle \xi_{\lambda}, \mathcal{N} \xi_{\lambda} \rangle = -4\int dz\, \langle \xi_{\lambda}, b_{\uparrow}(\varphi_{z}) b_{\downarrow,z} \xi_{\lambda}  \rangle + \text{c.c.}\;.
\end{equation}
We then estimate:
\begin{equation}\label{eq:bstarb}
|\langle \xi_{\lambda}, b_{\uparrow}(\varphi_{z}) b_{\downarrow,z} \xi_{\lambda}  \rangle|\leq \| b^{*}_{\uparrow}(\varphi_{z}) \xi_{\lambda}\| \| b_{\downarrow,z} \xi_{\lambda} \|\;;
\end{equation}
by Lemma \ref{lem:bosonbd}, using that $\|\varphi\|_{2} \leq C\rho^{-\frac{\gamma}{2}}$:
\begin{equation}\label{eq:bzbz}
|\langle \xi_{\lambda}, b_{\uparrow}(\varphi_{z}) b_{\downarrow,z} \xi_{\lambda}  \rangle|\leq \| b_{\uparrow}(\varphi_{z}) \xi_{\lambda}\| \| b_{\downarrow,z} \xi_{\lambda} \| + C\rho^{\frac{1}{2} - \frac{\gamma}{2}} \| b_{\downarrow,z} \xi_{\lambda} \|\;.
\end{equation}
Consider the first term in (\ref{eq:bzbz}). We have, using that $\|\varphi\|_{1} \leq C\rho^{-2\gamma}$:
\begin{eqnarray}
\int dz\, \| b_{\uparrow}(\varphi_{z}) \xi_{\lambda}\| \| b_{\downarrow,z} \xi_{\lambda} \| &\leq& C\rho\int dzdz'\, \varphi(z-z') (\| a_{\uparrow}(u^{\text{r}}_{z}) \|^{2} + \| a_{\downarrow}(u^{\text{r}}_{z'}) \|^{2})\nonumber\\
&\leq& C\rho^{1-2\gamma} \langle \xi_{\lambda}, \mathcal{N}\xi_{\lambda}\rangle\;.
\end{eqnarray}
Consider now the second term in (\ref{eq:bzbz}). We have:
\begin{eqnarray}\label{eq:csineq}
C\rho^{\frac{1}{2} - \frac{\gamma}{2}}\int dz\, \| b_{\downarrow,z} \xi_{\lambda} \| &\leq& CL^{\frac{3}{2}} \rho^{1 - \frac{\gamma}{2}} \| \mathcal{N}^{\frac{1}{2}}\xi_{\lambda} \| \nonumber\\ &\leq& CL^{3}\rho^{2 - \gamma} + \langle \xi_{\lambda}, \mathcal{N}\xi_{\lambda}\rangle\;.
\end{eqnarray}
All together:
\begin{equation}
\partial_{\lambda} \langle \xi_{\lambda}, \mathcal{N} \xi_{\lambda} \rangle \leq C\rho^{1-2\gamma} \langle \xi_{\lambda}, \mathcal{N}\xi_{\lambda}\rangle + CL^{3} \rho^{2 - \gamma} + C\langle \xi_{\lambda}, \mathcal{N}\xi_{\lambda}\rangle\;.
\end{equation}
The final claim follows from $\gamma \leq 1/2$. This concludes the proof.
\end{proof}
The next lemma will allow us to bound recurrent expressions in our computations. We shall use the short-hand notations $\partial^{n} \overline{v}^{\text{r}}_{x}$, $\partial^{n} u_{x}^{\text{r}}$ to denote the functions (in Eq. (\ref{eq:foll}) $x$ is fixed, $y$ is the argument of the functions):
\begin{equation}\label{eq:foll}
\frac{\partial^{n}}{\partial y_{i_{1}} \cdots \partial y_{i_n}} \overline{v}^{\text{r}}(y;x)\;,\qquad \frac{\partial^{n}}{\partial y_{i_{1}} \cdots \partial y_{i_n}} u^{\text{r}}(y;x)\;,
\end{equation}
for some choice of indices $i_{1}, \ldots, i_{n}$ (their values will be inessential for the bounds).
\begin{lemma}\label{lem:phi} Under the assumptions of Theorem \ref{thm:main}, the following holds. Let $n_{2} \in \mathbb{N}$, and let $\psi$ be an approximate ground state. Let $\gamma \leq 7/18$. Then:
\begin{eqnarray}
\Big| \int dxdy\, \varphi(x-y) \langle \xi_{\lambda}, a_{\uparrow}(u^{\text{r}}_{x}) a_{\uparrow}(\partial^{n_{2}}\overline{v}^{\text{r}}_{x}) a_{\downarrow}(u^{\text{r}}_{y}) a_{\downarrow}(\overline{v}^{\text{r}}_{y}) \xi_{\lambda} \rangle\Big| &\leq& C L^{\frac{3}{2}} \rho^{1 - \frac{\gamma}{2} + \frac{n_{2}}{3}} \| \mathcal{N}^{\frac{1}{2}} \xi_{\lambda} \|\nonumber\\
\Big| \int dxdy\, \varphi(x-y) \langle \xi_{\lambda}, a_{\uparrow}(\partial u^{\text{r}}_{x}) a_{\uparrow}(\partial^{n_{2}}\overline{v}^{\text{r}}_{x}) a_{\downarrow}(u^{\text{r}}_{y}) a_{\downarrow}(\overline{v}^{\text{r}}_{y}) \xi_{\lambda} \rangle\Big| &\leq& CL^{\frac{3}{2}} \rho^{1 - \frac{\gamma}{2} + \frac{n_{2}}{3}} ( \| \mathbb{H}_{0}^{\frac{1}{2}}\xi_{\lambda}\| + \rho^{\frac{1}{3}} \|\mathcal{N}^{\frac{1}{2}}\xi_{\lambda} \| )\;.\nonumber
\end{eqnarray}
\end{lemma}
We refer the reader to Appendix \ref{app:lemphi} for the proof. We shall use Lemma \ref{lem:phi} to propagate the a priori estimates on $\xi_{0} = R^{*}\psi$ for $\mathbb{H}_{0}$, $\mathbb{Q}_{1}$ and $\widetilde{\mathbb{Q}}_{1}$, Eqs. (\ref{eq:aprH0}), (\ref{eq:aprQ1}), to the interpolating states $\xi_{\lambda} = T^{*}_{\lambda} R^{*}\psi$, via a Gr\"onwall-type argument. To do so, the next proposition will play an important role.
\begin{proposition}[Propagation estimate for $\mathbb{H}_{0}$, $\mathbb{Q}_{1}$ - Part 1.]\label{prp:ders} Let $\gamma \leq 7/18$. Under the assumptions of Theorem \ref{thm:main} the following is true:
\begin{eqnarray}\label{eq:contrH0}
&&\frac{d}{d\lambda}\langle \xi_{\lambda}, \mathbb{H}_{0} \xi_{\lambda} \rangle = -\langle \xi_{\lambda}, \mathbb{T}_{1} \xi_{\lambda}\rangle + \mathcal{E}_{\mathbb{H}_{0}}(\xi_{\lambda})\;,\quad \frac{d}{d\lambda} \langle \xi_{\lambda}, \mathbb{Q}_{1} \xi_{\lambda} \rangle = -\langle \xi_{\lambda}, \mathbb{T}_{2} \xi_{\lambda} \rangle + \mathcal{E}_{\mathbb{Q}_{1}}(\xi_{\lambda})\;,\nonumber\\&&\quad\qquad\qquad\qquad\qquad \frac{d}{d\lambda} \langle \xi_{\lambda}, \widetilde{\mathbb{Q}}_{1} \xi_{\lambda} \rangle = -\langle \xi_{\lambda}, \mathbb{T}_{2} \xi_{\lambda} \rangle + \mathcal{E}_{\widetilde{\mathbb{Q}}_{1}}(\xi_{\lambda})
\end{eqnarray}
where:
\begin{eqnarray}\label{eq:Tdef}
\mathbb{T}_{1} &=& -\int dxdy\, \theta(|x-y|_{L} < \rho^{-\gamma})(-2\Delta \varphi)(x-y) a_{\uparrow}(u^{\text{r}}_{x}) a_{\uparrow}(\overline{v}^{\text{r}}_{x}) a_{\downarrow}(u^{\text{r}}_{y}) a_{\downarrow}(\overline{v}^{\text{r}}_{y}) + \text{h.c.}\nonumber\\
\mathbb{T}_{2} &=& -\int dxdy\, V(x-y) \varphi(x-y) a_{\uparrow}(\overline{v}^{\text{r}}_{x}) a_{\uparrow}(u_{x}) a_{\downarrow}(\overline{v}^{\text{r}}_{y}) a_{\downarrow}(u_{y})  + \text{h.c.,}
\end{eqnarray}
and the error terms are bounded as, for $0\leq \eta < \min\{\gamma, \frac{1}{3}\}$:
\begin{eqnarray}\label{eq:errH0}
| \mathcal{E}_{\mathbb{H}_{0}}(\xi_{\lambda}) | &\leq& CL^{\frac{3}{2}} \rho^{\frac{4}{3} - \frac{\gamma}{2}} ( \| \mathbb{H}_{0}^{\frac{1}{2}}\xi_{\lambda}\| + \rho^{\frac{1}{3}} \|\mathcal{N}^{\frac{1}{2}}\xi_{\lambda} \| )\nonumber\\
|\mathcal{E}_{\widetilde{\mathbb{Q}}_{1}}(\xi_{\lambda})| &\leq& C_{\beta}\rho^{1 - 2\gamma - \eta} \| \widetilde{\mathbb{Q}}^{\frac{1}{2}}_{1} \xi_{\lambda} \| ( \| \mathbb{H}_{0}^{\frac{1}{2}} \xi_{\lambda} \| + \rho^{\frac{5\eta}{2}} \| \mathcal{N}^{\frac{1}{2}} \xi_{\lambda} \| ) + CL^{\frac{3}{2}} \rho^{2 - 2\gamma} \| \widetilde{\mathbb{Q}}^{\frac{1}{2}}_{1} \xi_{\lambda} \|\nonumber\\
|\mathcal{E}_{\mathbb{Q}_{1}}(\xi_{\lambda})| &\leq& C_{\beta}\rho^{1 - 2\gamma - \eta} \| \widetilde{\mathbb{Q}}^{\frac{1}{2}}_{1} \xi_{\lambda} \| ( \| \mathbb{H}_{0}^{\frac{1}{2}} \xi_{\lambda} \| + \rho^{\frac{5\eta}{2}} \| \mathcal{N}^{\frac{1}{2}} \xi_{\lambda} \| ) + CL^{\frac{3}{2}} \rho^{2 - 2\gamma} \| \mathbb{Q}^{\frac{1}{2}}_{1} \xi_{\lambda} \|\;.
\end{eqnarray}
\end{proposition}
\begin{proof}
\noindent{\underline{Derivative of $\langle \xi_{\lambda}, \mathbb{H}_{0} \xi_{\lambda} \rangle$}.} We compute:
\begin{equation}
\frac{d}{d\lambda} \langle \xi_{\lambda}, \mathbb{H}_{0} \xi_{\lambda} \rangle = -\langle \xi_{\lambda}, [ \mathbb{H}_{0}, B ] \xi_{\lambda} \rangle + \text{c.c.}
\end{equation}
with
\begin{equation}\label{eq:HB}
[ \mathbb{H}_{0}, B ] = \int dzdz'\, \varphi(z-z')[ \mathbb{H}_{0}, a_{\uparrow}(u^{\text{r}}_{z}) a_{\uparrow}(\overline{v}^{\text{r}}_{z}) a_{\downarrow}(u^{\text{r}}_{z'}) a_{\downarrow}(\overline{v}^{\text{r}}_{z'}) ]\;.
\end{equation}
We write the commutator as:
\begin{eqnarray}
[ \mathbb{H}_{0}, a_{\uparrow}(u^{\text{r}}_{z}) a_{\uparrow}(\overline{v}^{\text{r}}_{z}) a_{\downarrow}(u^{\text{r}}_{z'}) a_{\downarrow}(\overline{v}^{\text{r}}_{z'}) ] &=& [ \mathbb{H}_{0}, a_{\uparrow}(u_{z}^{\text{r}}) a_{\uparrow}(\overline{v}^{\text{r}}_{z}) ] a_{\downarrow}(u^{\text{r}}_{z'}) a_{\downarrow}(\overline{v}^{\text{r}}_{z'}) \nonumber\\
&& + a_{\uparrow}(u_{z}^{\text{r}}) a_{\uparrow}(\overline{v}^{\text{r}}_{z}) [ \mathbb{H}_{0}, a_{\downarrow}(u^{\text{r}}_{z'}) a_{\downarrow}(\overline{v}^{\text{r}}_{z'}) ] \;,
\end{eqnarray}
where, recalling that $\mathbb{H}_{0} = \sum_{\sigma} | |k|^{2} - \mu_{\sigma} | \hat a^{*}_{k,\sigma} \hat a_{k,\sigma}$, with $\mu_{\sigma} = k_{F}^{\sigma 2}$:
\begin{eqnarray}\label{eq:kin3}
[ \mathbb{H}_{0}, a_{\sigma}(u_{z}^{\text{r}}) a_{\sigma}(\overline{v}^{\text{r}}_{z}) ] &=& -a_{\sigma}((-\Delta - \mu_{\sigma})u_{z}^{\text{r}}) a_{\sigma}(\overline{v}^{\text{r}}_{z}) - a_{\sigma}(u_{z}^{\text{r}}) a_{\sigma}((\Delta + \mu_{\sigma})\overline{v}^{\text{r}}_{z})  \\
&=& \Delta_{z} a_{\sigma}(u^{\text{r}}_{z}) a_{\sigma}(\overline{v}^{\text{r}}_{z}) - 2a_{\sigma}(u_{z}^{\text{r}}) a_{\sigma}(\Delta\overline{v}^{\text{r}}_{z}) - 2a_{\sigma}(\nabla u^{\text{r}}_{z}) a_{\sigma}(\nabla \overline{v}^{\text{r}}_{z})\;.\nonumber
\end{eqnarray}
To write this identity we used that $u^{\text{r}}(z';z) \equiv u^{\text{r}}(z-z')$ and $v^{\text{r}}(z';z) \equiv v^{\text{r}}(z' + z)$; recall the definitions (\ref{eq:reguv}). Hence, $\Delta_{z'} u^{\text{r}}(z';z) = \Delta_{z} u^{\text{r}}(z';z)$ and $\Delta_{z'} v^{\text{r}}(z';z) = \Delta_{z} v^{\text{r}}(z';z)$. We define:
\begin{eqnarray}
\text{I} &=& \int dzdz'\, \varphi(z-z') \Delta_{z} \langle \xi_{\lambda}, a_{\uparrow}(u^{\text{r}}_{z}) a_{\uparrow}(\overline{v}^{\text{r}}_{z}) a_{\downarrow}(u^{\text{r}}_{z'}) a_{\downarrow}(\overline{v}^{\text{r}}_{z'}) \xi_{\lambda}\rangle + (\uparrow, z) \leftrightarrow (\downarrow, z') \nonumber\\
\text{II} &=& -2\int dzdz'\, \varphi(z-z') \langle \xi_{\lambda}, a_{\uparrow}(u_{z}^{\text{r}}) a_{\uparrow}(\Delta\overline{v}^{\text{r}}_{z}) a_{\downarrow}(u^{\text{r}}_{z'}) a_{\downarrow}(\overline{v}^{\text{r}}_{z'})\xi_{\lambda}\rangle + (\uparrow, z) \leftrightarrow (\downarrow, z')  \nonumber\\
\text{III} &=& -2\int dzdz'\, \varphi(z-z') \langle \xi_{\lambda}, a_{\uparrow}(\nabla u^{\text{r}}_{z}) a_{\uparrow}(\nabla \overline{v}^{\text{r}}_{z}) a_{\downarrow}(u^{\text{r}}_{z'}) a_{\downarrow}(\overline{v}^{\text{r}}_{z'}) \xi_{\lambda} \rangle + (\uparrow, z) \leftrightarrow (\downarrow, z')\;.
\end{eqnarray}
To estimate $\text{II}$ and $\text{III}$ we shall use Lemma \ref{lem:phi}. We have:
\begin{eqnarray}\label{eq:II-III}
|\text{II}| &\leq& C L^{\frac{3}{2}} \rho^{\frac{5}{3} - \frac{\gamma}{2}} \| \mathcal{N}^{\frac{1}{2}} \xi_{\lambda} \|\nonumber\\
|\text{III}| &\leq& CL^{\frac{3}{2}} \rho^{\frac{4}{3} - \frac{\gamma}{2}} ( \| \mathbb{H}_{0}^{\frac{1}{2}}\xi_{\lambda}\| + \rho^{\frac{1}{3}} \|\mathcal{N}^{\frac{1}{2}}\xi_{\lambda} \| )\;.
\end{eqnarray}
Consider now $\text{I}$. We rewrite it as:
\begin{equation}\label{eq:kinI}
\text{I} = \int_{|z-z'|_{L} \leq \rho^{-\gamma}} dzdz'\, \varphi(z-z') (\Delta_{z} + \Delta_{z'}) \langle \xi_{\lambda}, a_{\uparrow}(u^{\text{r}}_{z}) a_{\uparrow}(\overline{v}^{\text{r}}_{z}) a_{\downarrow}(u^{\text{r}}_{z'}) a_{\downarrow}(\overline{v}^{\text{r}}_{z'}) \xi_{\lambda}\rangle\;.
\end{equation}
The condition on the integration domain follows from the fact that $\varphi$ is the periodization of $\varphi_{\infty}$, Eq. (\ref{eq:phiper}), which is compactly supported in the ball $B$ of radius $\rho^{-\gamma}$, Eq. (\ref{eq:scat2}). Therefore, using that $\varphi_{\infty} = \nabla \varphi_{\infty} = 0$ on $\partial B$, we can integrate by parts in Eq. (\ref{eq:kinI}), without producing boundary terms. This is the point of our analysis where we take advantage of the Neumann boundary conditions of the scattering equation. We have:
\begin{equation}
\text{I} = \int dzdz'\, \theta(|z-z'|_{L} < \rho^{-\gamma})(2\Delta \varphi)(z-z') \langle \xi_{\lambda}, a_{\uparrow}(u^{\text{r}}_{z}) a_{\uparrow}(\overline{v}^{\text{r}}_{z}) a_{\downarrow}(u^{\text{r}}_{z'}) a_{\downarrow}(\overline{v}^{\text{r}}_{z'}) \xi_{\lambda} \rangle\;,
\end{equation}
with $\theta(\cdot)$ the characteristic function. All in all:
\begin{eqnarray}\label{eq:pause}
&&\frac{d}{d\lambda}\langle \xi_{\lambda}, \mathbb{H}_{0} \xi_{\lambda} \rangle = -\langle \xi_{\lambda}, \mathbb{T}_{1} \xi_{\lambda} \rangle  + \mathcal{E}_{\mathbb{H}_{0}}(\xi_{\lambda})\;,
\end{eqnarray}
where 
\begin{equation}
\mathbb{T}_{1} = \int dxdy\, \theta(|z-z'|_{L} < \rho^{-\gamma}) (2\Delta \varphi)(x-y) a_{\uparrow}(u^{\text{r}}_{x}) a_{\uparrow}(\overline{v}^{\text{r}}_{x}) a_{\downarrow}(u^{\text{r}}_{y}) a_{\downarrow}(\overline{v}^{\text{r}}_{y}) + \text{h.c.} \nonumber\\
\end{equation}
and $\mathcal{E}_{\mathbb{H}_{0}}(\xi_{\lambda})$ collects the error terms, produced by the contributions $\text{II}$ and $\text{III}$. Thanks to the estimates in (\ref{eq:II-III}):
\begin{equation}
|\mathcal{E}_{\mathbb{H}_{0}}(\xi_{\lambda})| \leq CL^{\frac{3}{2}} \rho^{\frac{4}{3} - \frac{\gamma}{2}} ( \| \mathbb{H}_{0}^{\frac{1}{2}}\xi_{\lambda}\| + \rho^{\frac{1}{3}} \|\mathcal{N}^{\frac{1}{2}}\xi_{\lambda} \| )\;.
\end{equation}
This concludes the proof of the first of Eqs. (\ref{eq:contrH0}). 

\medskip

\noindent{\underline{Derivative of $\langle \xi_{\lambda}, \widetilde{\mathbb{Q}}_{1} \xi_{\lambda} \rangle$}.} We compute:
\begin{equation}\label{eq:derQ1}
\frac{d}{d\lambda} \langle \xi_{\lambda}, \widetilde{\mathbb{Q}}_{1} \xi_{\lambda} \rangle = -\langle \xi_{\lambda}, [ \widetilde{\mathbb{Q}}_{1}, B] \xi_{\lambda}\rangle + \text{c.c.,}
\end{equation}
with
\begin{eqnarray}\label{eq:comQ1}
&&[ \widetilde{\mathbb{Q}}_{1}, B ] =\\&& \frac{1}{2} \sum_{\sigma\neq \sigma'} \int dxdydzdz'\, V(x-y) \varphi(z-z') [ a^{*}_{\sigma}(u_{x})a^{*}_{\sigma'}(u_{y}) a_{\sigma'}(u_{y})a_{\sigma}(u_{x}), a_{\uparrow}(u^{\text{r}}_{z}) a_{\uparrow}(\overline{v}^{\text{r}}_{z}) a_{\downarrow}(u^{\text{r}}_{z'}) a_{\downarrow}(\overline{v}^{\text{r}}_{z'})]\;.\nonumber
\end{eqnarray}
We rewrite the commutator as:
\begin{eqnarray}
&&[ a^{*}_{\sigma}(u_{x})a^{*}_{\sigma'}(u_{y}) a_{\sigma'}(u_{y})a_{\sigma}(u_{x}), a_{\uparrow}(u^{\text{r}}_{z}) a_{\uparrow}(\overline{v}^{\text{r}}_{z}) a_{\downarrow}(u^{\text{r}}_{z'}) a_{\downarrow}(\overline{v}^{\text{r}}_{z'})] = \\
&&\qquad -[ a^{*}_{\sigma}(u_{x})a^{*}_{\sigma'}(u_{y}), a_{\uparrow}(u^{\text{r}}_{z}) a_{\downarrow}(u^{\text{r}}_{z'}) ] a_{\uparrow}(\overline{v}^{\text{r}}_{z}) a_{\downarrow}(\overline{v}^{\text{r}}_{z'}) a_{\sigma'}(u_{y})a_{\sigma}(u_{x})\;.\nonumber
\end{eqnarray}
where:
\begin{eqnarray}\label{eq:comuQ1}
\big[ a^{*}_{\sigma}(u_{x}) a^{*}_{\sigma'}(u_{y}) ,  a_{\uparrow}(u^{\text{r}}_{z})  a_{\downarrow}(u^{\text{r}}_{z'}) \big] &=&  a^{*}_{\sigma}(u_{x}) \big( \delta_{\sigma'\uparrow} u^{\text{r}}_{\sigma'}(z;y) a_{\downarrow}(u^{\text{r}}_{z'}) - \delta_{\sigma'\downarrow} u^{\text{r}}_{\sigma'}(z';y) a_{\uparrow}(u^{\text{r}}_{z}) \big)\\
&& - a^{*}_{\sigma'}(u_{y})\big( \delta_{\sigma\uparrow} u^{\text{r}}_{\sigma}(z;x) a_{\downarrow}(u^{\text{r}}_{z'}) - \delta_{\sigma\downarrow} u^{\text{r}}_{\sigma}(z';x) a_{\uparrow}(u^{\text{r}}_{z}) \big)\nonumber\\
&& + \delta_{\sigma\uparrow} \delta_{\sigma'\downarrow} u^{\text{r}}_{\sigma}(z;x) u^{\text{r}}_{\sigma'}(z';y) - \delta_{\sigma\downarrow}\delta_{\sigma'\uparrow} u^{\text{r}}_{\sigma}(z';x) u^{\text{r}}_{\sigma'}(z;y)\;.\nonumber
\end{eqnarray}
Consider the first term in the first line of Eq. (\ref{eq:comuQ1}). All the other terms in the first and second line of (\ref{eq:comuQ1}) give rise to contributions to (\ref{eq:comQ1}) that can be estimated in exactly the same way. We have, recall the definition (\ref{eq:bphi}) of $b_{\sigma}(\varphi_{z})$:
\begin{eqnarray}
&&\int dxdydzdz'\, V(x-y) \varphi(z-z') u^{\text{r}}_{\uparrow}(z;y) \langle \xi_{\lambda}, a^{*}_{\downarrow}(u_{x}) a_{\downarrow}(u^{\text{r}}_{z'}) a_{\uparrow}(\overline{v}^{\text{r}}_{z}) a_{\downarrow}(\overline{v}^{\text{r}}_{z'}) a_{\uparrow}(u_{y})a_{\downarrow}(u_{x})\xi_{\lambda}\rangle\nonumber\\
&&\equiv -\int dxdydz\, V(x-y) u^{\text{r}}_{\uparrow}(z;y) \langle \xi_{\lambda}, a^{*}_{\downarrow}(u_{x})  b_{\downarrow}(\varphi_{z}) a_{\uparrow}(\overline{v}^{\text{r}}_{z}) a_{\uparrow}(u_{y})a_{\downarrow}(u_{x}) \xi_{\lambda}\rangle =: \text{A}\;.
\end{eqnarray}
We then estimate, by Lemma \ref{lem:bosonbd}:
\begin{eqnarray}
|\text{A}| &\leq& \int dxdydz\, V(x-y) |u^{\text{r}}_{\uparrow}(z;y)| \| b^{*}_{\downarrow}(\varphi_{z}) a_{\downarrow}(u_{x}) \xi_{\lambda}\| \| a_{\uparrow}(\overline{v}^{\text{r}}_{z}) a_{\uparrow}(u_{y})a_{\downarrow}(u_{x}) \xi_{\lambda}\| \nonumber\\
&\leq& C\rho^{\frac{1}{2}} \int dxdydz\, V(x-y) |u^{\text{r}}_{\uparrow}(z;y)| \| b_{\downarrow}(\varphi_{z}) a_{\downarrow}(u_{x}) \xi_{\lambda}\| \| a_{\uparrow}(u_{y})a_{\downarrow}(u_{x}) \xi_{\lambda}\| \nonumber\\
&& + C\rho^{1 - \frac{\gamma}{2}} \int dxdydz\, V(x-y) |u^{\text{r}}_{\uparrow}(z;y)| \| a_{\downarrow}(u_{x}) \xi_{\lambda}\| \| a_{\uparrow}(u_{y})a_{\downarrow}(u_{x}) \xi_{\lambda}\|\nonumber\\
&\equiv& \text{A}_{1} + \text{A}_{2}\;.
\end{eqnarray}
To get the second inequality, we used the estimate (\ref{eq:estbosb}), together with $\| a_{\uparrow}(\overline{v}^{\text{r}}_{z}) \| \leq C\rho^{\frac{1}{2}}$. Consider the term $\text{A}_{2}$. By Cauchy-Schwarz inequality, using that $\| u^{\text{r}}_{y} \|_{1} \leq C$:
\begin{equation}\label{eq:A2}
|\text{A}_{2}| \leq C\rho^{1 - \frac{\gamma}{2}} \| \mathcal{N}^{\frac{1}{2}} \xi_{\lambda} \| \| \widetilde{\mathbb{Q}}^{\frac{1}{2}}_{1} \xi_{\lambda} \|\;.
\end{equation}
Consider now the term $\text{A}_{1}$. Here we split $u_{x}$ as $u^{<}_{x} + u^{>}_{x}$, with $\hat u^{<}(k)$ supported for $|k| \leq \rho^{\eta}$, with $\eta < \min\{\gamma , \frac{1}{3}\}$ to be chosen later. Correspondingly:
\begin{eqnarray}
\text{A}_{1} &\leq& C\rho^{\frac{1}{2}} \int dxdydz\, V(x-y) |u^{\text{r}}_{\uparrow}(z;y)| \| b_{\downarrow}(\varphi_{z}) a_{\downarrow}(u^{<}_{x}) \xi_{\lambda}\| \| a_{\uparrow}(u_{y})a_{\downarrow}(u_{x}) \xi_{\lambda}\| \nonumber\\
&& + C\rho^{\frac{1}{2}} \int dxdydz\, V(x-y) |u^{\text{r}}_{\uparrow}(z;y)| \| b_{\downarrow}(\varphi_{z}) a_{\downarrow}(u^{>}_{x}) \xi_{\lambda}\| \| a_{\uparrow}(u_{y})a_{\downarrow}(u_{x}) \xi_{\lambda}\|\nonumber\\
&\equiv& \text{A}_{1;1} + \text{A}_{1;2}\;.
\end{eqnarray}
Consider $\text{A}_{1;1}$. Using that $\| u^{<}_{x} \|_{2} \leq C\rho^{\frac{3\eta}{2}}$, we get:
\begin{eqnarray}\label{eq:A11}
|\text{A}_{1;1}| &\leq& C\rho^{\frac{1}{2} + \frac{3\eta}{2}} \int dxdydz\, V(x-y) |u^{\text{r}}_{\uparrow}(z;y)| \| b_{\downarrow}(\varphi_{z})  \xi_{\lambda}\| \| a_{\uparrow}(u_{y})a_{\downarrow}(u_{x}) \xi_{\lambda}\|\nonumber\\
&\leq& C\rho^{1 + \frac{3\eta}{2}}  \int dxdydzdz'\, V(x-y) |u^{\text{r}}_{\uparrow}(z;y)| \varphi(z-z') \| a_{\downarrow}(u^{\text{r}}_{z}) \xi_{\lambda} \| \| a_{\uparrow}(u_{y})a_{\downarrow}(u_{x}) \xi_{\lambda}\| \nonumber\\
&\leq& C \rho^{1 + \frac{3\eta}{2} - 2\gamma} \| \mathcal{N}^{\frac{1}{2}}\xi_{\lambda} \| \| \widetilde{\mathbb{Q}}_{1}^{\frac{1}{2}} \xi_{\lambda} \|\;,
\end{eqnarray}
where the last step follows from Cauchy-Schwarz inequality. Finally, consider $\text{A}_{1;2}$. We have:
\begin{eqnarray}\label{eq:A12}
|\text{A}_{1;2}| &\leq& C_{\beta}\rho^{1 - 2\gamma} \int dxdydz\, V(x-y) |u^{\text{r}}_{\uparrow}(z;y)| \|a_{\downarrow}(u^{>}_{x}) \xi_{\lambda}\| \| a_{\uparrow}(u_{y})a_{\downarrow}(u_{x}) \xi_{\lambda}\|\nonumber\\
&\leq& C_{\beta}\rho^{1-2\gamma - \eta} \| \mathbb{H}_{0}^{\frac{1}{2}} \xi_{\lambda} \| \| \widetilde{\mathbb{Q}}^{\frac{1}{2}}_{1} \xi_{\lambda} \|\;,
\end{eqnarray}
where the last step follows from Cauchy-Schwarz inequality, combined with:
\begin{equation}
\int dx\, \|a_{\downarrow}(u^{>}_{x}) \xi_{\lambda}\|^{2} \leq C\rho^{-2\eta} \langle \xi_{\lambda}, \mathbb{H}_{0} \xi_{\lambda} \rangle\;.
\end{equation}
This inequality can be proven in a way completely analogous to (\ref{eq:Nbig}). In fact:
\begin{eqnarray}
\int dx\, \|a_{\sigma}(u^{>}_{x}) \xi_{\lambda}\|^{2} &=& \sum_{k} (\hat u_{\sigma}^{>}(k))^{2} \langle \xi_{\lambda}, \hat a^{*}_{k,\sigma} a_{k,\sigma} \xi_{\lambda}\rangle \nonumber\\
&\leq& \rho^{-2\eta} \sum_{k} |k|^{2} (\hat u_{\sigma}^{>}(k))^{2} \langle \xi_{\lambda}, \hat a^{*}_{k,\sigma} a_{k,\sigma} \xi_{\lambda}\rangle\nonumber\\
&\leq& C\rho^{-2\eta} \sum_{k} | |k|^{2} - \mu_{\sigma} | \langle \xi_{\lambda}, \hat a^{*}_{k,\sigma} a_{k,\sigma} \xi_{\lambda}\rangle\nonumber\\
&\leq&C\rho^{-2\eta} \langle \xi_{\lambda}, \mathbb{H}_{0} \xi_{\lambda}\rangle\;,
\end{eqnarray}
where we used that $\hat u_{\sigma}^{>}(k)$ is supported for $|k| \geq \rho^{\eta}$, that $0\leq \hat u_{\sigma}^{>}(k) \leq 1$, and that $\mu_{\sigma} \leq C\rho^{\frac{2}{3}} \ll C\rho^{2\eta}$. Hence:
\begin{equation}\label{eq:bdA}
|\text{A}| \leq C_{\beta}\rho^{1 - 2\gamma - \eta} \| \widetilde{\mathbb{Q}}^{\frac{1}{2}}_{1} \xi_{\lambda} \| ( \| \mathbb{H}_{0}^{\frac{1}{2}} \xi_{\lambda} \| + \rho^{\frac{5\eta}{2}} \| \mathcal{N}^{\frac{1}{2}} \xi_{\lambda} \| )\;.
\end{equation}
The other three terms arising from the first two lines of the right-hand side of (\ref{eq:comuQ1}) can be estimated in exactly the same way.

Next, let us plug the last two terms in the right-hand side of Eq. (\ref{eq:comuQ1}) in Eq. (\ref{eq:comQ1}). We get:
\begin{equation}\label{eq:Q1main}
\text{I}_{\text{main}} = -\int dxdydzdz'\, V(x-y) \varphi(z-z') u^{\text{r}}_{\uparrow}(z;x) u^{\text{r}}_{\downarrow}(z';y) a_{\uparrow}(\overline{v}^{\text{r}}_{z}) a_{\uparrow}(u_{x}) a_{\downarrow}(\overline{v}^{\text{r}}_{z'}) a_{\downarrow}(u_{y})\;. 
\end{equation}
We shall use that $u^{\text{r}}_{\uparrow}(z;x)$ behaves as a delta function, to leading order in $\rho$. More precisely, we write:
\begin{equation}\label{eq:udelta}
u_{\sigma}^{\text{r}}(z;x) = \delta_{\sigma}^{\text{r}}(z;x) - \nu_{\sigma}(z;x)\;,
\end{equation}
where $\nu_{\sigma}(z;x)\equiv \nu_{\sigma}(z-x)$ with Fourier transform given by:
\begin{equation}
\hat \nu_{\sigma}(k) = \left\{ \begin{array}{cc} 1 & \text{for $0\leq |k| \leq k^{\sigma}_{F}$} \\ 0 & \text{for $|k| > 2k_{F}^{\sigma}$,}   \end{array} \right.
\end{equation}
and it smoothly interpolates between $1$ and $0$ for $k_{F}^{\sigma} \leq |k| \leq 2k_{F}^{\sigma}$. The function $\nu_{\sigma, x}(y)$ satisfies the bounds:
\begin{equation}
\| \nu_{\sigma, x} \|_{1} \leq C\;,\qquad \| \nu_{\sigma, x} \|_{\infty} \leq C\rho\;.
\end{equation}
Instead, $\delta_{\sigma}^{\text{r}}(z;x) \equiv \delta^{\text{r}}(z-x)$, with:
\begin{equation}
\hat \delta_{\sigma}^{\text{r}}(k) = \left\{ \begin{array}{cc} 1 & \text{for $0\leq |k| \leq \frac{3}{2}\rho^{-\beta}$} \\ 0 & \text{for $|k| \geq 2\rho^{-\beta}$} \\ \end{array} \right.
\end{equation}
and it smoothly interpolates between $1$ and $0$ in the region $(3/2)\rho^{-\beta} \leq |k| \leq 2\rho^{-\beta}$. The function $\delta^{\text{r}}_{\sigma,x}(y)$ is an approximate Dirac delta function at $x$, such that 
\begin{equation}
\| \delta^{\text{r}}_{\sigma,x} \|_{1} \leq C\;,\qquad \|\delta^{\text{r}}_{\sigma, x}\|_{\infty} \leq C\rho^{-3\beta}\;.
\end{equation}
Let us perform the replacement (\ref{eq:udelta}) in Eq. (\ref{eq:Q1main}). Consider the terms with one $\nu$. We have:
\begin{eqnarray}\label{eq:nu1}
&&\int dxdydzdz'\, V(x-y) \varphi(z-z') |\delta_{\sigma}^{\text{r}}(z;x)| |\nu_{\sigma'}(z';y)| |\langle \xi_{\lambda}, a_{\uparrow}(\overline{v}^{\text{r}}_{z}) a_{\uparrow}(u_{x}) a_{\downarrow}(\overline{v}^{\text{r}}_{z'}) a_{\downarrow}(u_{y}) \xi_{\lambda}\rangle|\nonumber\\
&&\quad \leq C\int dxdy\, V(x-y) \rho^{2-2\gamma} \| a_{\uparrow}(u_{x}) a_{\downarrow}(u_{y}) \xi_{\lambda}\|\nonumber\\
&&\quad \leq CL^{\frac{3}{2}} \rho^{2 - 2\gamma} \| \widetilde{\mathbb{Q}}^{\frac{1}{2}}_{1} \xi_{\lambda} \|\;.
\end{eqnarray}
Next, consider the terms with two $\nu$. Proceeding as above we have:
\begin{eqnarray}\label{eq:nu2}
&&\int dxdydzdz'\, V(x-y) \varphi(z-z') |\nu_{\sigma}(z;x)| |\nu_{\sigma'}(z';y)| |\langle \xi_{\lambda}, a(\overline{v}^{\text{r}}_{z}) a_{\uparrow}(u_{x}) a(\overline{v}^{\text{r}}_{z'}) a_{\downarrow}(u_{y}) \xi_{\lambda}\rangle| \nonumber\\
&& \quad \leq C\int dxdy\, V(x-y) \rho^{2 - 2\gamma} \| a_{\uparrow}(u_{x}) a_{\downarrow}(u_{y}) \xi_{\lambda}\|\nonumber\\
&&\quad \leq CL^{\frac{3}{2}} \rho^{2 - 2\gamma} \| \widetilde{\mathbb{Q}}^{\frac{1}{2}}_{1} \xi_{\lambda} \|\;.
\end{eqnarray}
Therefore, the main contribution to Eq. (\ref{eq:Q1main}) is:
\begin{equation}\label{eq:Q1fin}
-\int dxdydzdz'\, V(x-y) \varphi(z-z') \delta^{\text{r}}_{\uparrow}(z;x) \delta^{\text{r}}_{\downarrow}(z';y) a_{\uparrow}(\overline{v}^{\text{r}}_{z}) a_{\uparrow}(u_{x}) a_{\downarrow}(\overline{v}^{\text{r}}_{z'}) a_{\downarrow}(u_{y})\;.
\end{equation}
The next lemma will allow us to replace the approximate $\delta$ functions with true $\delta$ functions, up to a small error. This is one of the points where we use the regularity of the potential.
\begin{lemma}\label{lem:UV2} Under the assumptions of Proposition \ref{prp:ders}, the following is true:
\begin{eqnarray}\label{eq:phizz}
&&-\int dxdydzdz'\, V(x-y) \varphi(z-z') \delta^{\text{r}}_{\uparrow}(z;x) \delta^{\text{r}}_{\downarrow}(z';y) \langle \xi_{\lambda}, a_{\uparrow}(\overline{v}^{\text{r}}_{z}) a_{\uparrow}(u_{x}) a_{\downarrow}(\overline{v}^{\text{r}}_{z'}) a_{\downarrow}(u_{y}) \xi_{\lambda} \rangle  \nonumber\\
&& = -\int dxdy\, V(x-y) \varphi(x-y) \langle \xi_{\lambda}, a_{\uparrow}(\overline{v}^{\text{r}}_{x}) a_{\uparrow}(u_{x}) a_{\downarrow}(\overline{v}^{\text{r}}_{y}) a_{\downarrow}(u_{y}) \xi_{\lambda} \rangle + \widehat{\mathcal{E}}_{\widetilde{\mathbb{Q}}_{1}}(\xi_{\lambda})\;,
\end{eqnarray}
where, for all $n\geq 4$, taking $V\in C^{k}$ with $k$ large enough:
\begin{equation}
|\widehat{\mathcal{E}}_{\widetilde{\mathbb{Q}}_{1}}(\xi_{\lambda})| \leq C_{n} \rho^{\beta (n - 3)} ( CL^{3} \rho^{2} + \langle \xi_{\lambda}, \widetilde{\mathbb{Q}}_{1} \xi_{\lambda}\rangle )\;.
\end{equation}
\end{lemma}
The proof of Lemma \ref{lem:UV2} is deferred to Appendix \ref{sec:UV2}. 
Putting together (\ref{eq:bdA}), (\ref{eq:nu1}), (\ref{eq:nu2}), (\ref{eq:phizz}) and recalling the definition of $\mathbb{T}_{2}$, Eq. (\ref{eq:Tdef}), the claim follows. The proof of the statement about $\langle \xi_{\lambda},\mathbb{Q}_{1} \xi_{\lambda}\rangle$ is exactly the same, and we shall omit the details. 
\end{proof}
\subsection{Scattering equation cancellation}
The next result will imply an important cancellation, that will be used to propagate the a priori bounds for $\mathbb{H}_{0}$, $\mathbb{Q}_{1}$, and to compute the ground state energy at order $\rho^{2}$.
\begin{proposition}[Scattering equation cancellation]\label{prp:cancscat} Let:
\begin{equation}
\widetilde{\mathbb{Q}}^{\text{r}}_{4} = \int dxdy\, V(x-y) a^{*}_{\uparrow}(u^{\text{r}}_{x}) a^{*}_{\downarrow}(u^{\text{r}}_{y}) a^{*}_{\downarrow}(\overline{v}^{\text{r}}_{y}) a^{*}_{\uparrow}(\overline{v}^{\text{r}}_{x})  + \text{h.c..}
\end{equation}
Let $\gamma \leq 7/18$. Under the assumptions of Theorem \ref{thm:main}, the following holds:
\begin{eqnarray}\label{eq:lambdagamma}
&&\langle \xi_{\lambda}, ( \mathbb{T}_{1} + \mathbb{T}_{2} + \widetilde{\mathbb{Q}}^{\text{r}}_{4}) \xi_{\lambda} \rangle \\
&&= \lambda_{\gamma} \int dxdy\,\theta(|x - y|_{L}< \rho^{-\gamma}) (1-\varphi(x-y)) \langle \xi_{\lambda}, (a_{\uparrow}(\overline{v}^{\text{r}}_{x}) a_{\uparrow}(u^{\text{r}}_{x}) a_{\downarrow}(\overline{v}^{\text{r}}_{y}) a_{\downarrow}(u^{\text{r}}_{y}) + \text{h.c.}) \xi_{\lambda}  \rangle + \mathcal{E}_{\mathbb{T}_{2}}(\xi_{\lambda})\nonumber
\end{eqnarray}
with $|\lambda_{\gamma}| \leq C\rho^{3\gamma}$ and, for $\delta > 0$ and $L$ large enough:
\begin{equation}\label{eq:ETerr}
| \mathcal{E}_{\mathbb{T}_{2}}(\xi_{\lambda}) | \leq CL^{\frac{3}{2}}\rho^{\frac{3}{2}} \| \mathcal{N}^{\frac{1}{2}} \xi_{\lambda} \|\;.\end{equation}
Moreover:
\begin{eqnarray}\label{eq:lambdagamma0}
&&\Big| \lambda_{\gamma} \int dxdy\,\theta(|x - y|_{L} < \rho^{-\gamma}) (1-\varphi(x-y)) \langle \xi_{\lambda}, (a_{\uparrow}(\overline{v}^{\text{r}}_{x}) a_{\uparrow}(u^{\text{r}}_{x}) a_{\downarrow}(\overline{v}^{\text{r}}_{y}) a_{\downarrow}(u^{\text{r}}_{y}) + \text{h.c.}) \xi_{\lambda}  \rangle \Big| \nonumber\\
&&\qquad \leq C L^{\frac{3}{2}} \rho^{1 + \frac{3\gamma}{2}} \| \mathcal{N}^{\frac{1}{2}} \xi_{\lambda} \|\;.
\end{eqnarray}
\end{proposition}
\begin{proof} Let $\mathbb{T}_{2}^{\text{r}}$ be the operator obtained from $\mathbb{T}_{2}$ after replacing all $u$ with $u^{\text{r}}$. The error term $\mathcal{E}_{\mathbb{T}_{2}}(\xi_{\lambda})$ takes into account the difference $\langle \xi_{\lambda}, (\mathbb{T}_{2} - \mathbb{T}_{2}^{\text{r}}) \xi_{\lambda}\rangle$. We postpone its estimate (\ref{eq:ETerr}) to Appendix \ref{app:ab}.

Recall the notation $b_{x,\sigma} = a_{\sigma}(u^{\text{r}}_{x}) a_{\sigma}(\overline{v}^{\text{r}}_{x})$, Eq. (\ref{eq:bphi}). We then have, using that by the compact support of the potential $V(x-y) \equiv V(x-y) \theta(|x-y|_{L} < \rho^{-\gamma})$:
\begin{eqnarray}
&&\langle \xi_{\lambda}, ( \mathbb{T}_{1} + \mathbb{T}^{\text{r}}_{2} + \widetilde{\mathbb{Q}}^{\text{r}}_{4}) \xi_{\lambda} \rangle \\
&&=\int dxdy\, \theta(|x-y|_{L}<\rho^{-\gamma})\big(2\Delta \varphi(x-y) + V(x-y)(1 - \varphi(x-y))\big) \langle \xi_{\lambda}, b_{x,\uparrow} b_{y,\downarrow} \xi_{\lambda}  \rangle + \text{c.c.}\nonumber\\
&& \equiv 2\int dxdy\, \theta(|x-y|_{L}<\rho^{-\gamma}) \Big(-\Delta f(x-y) + \frac{1}{2}V(x-y) f(x-y)\Big) \langle \xi_{\lambda}, b_{x,\uparrow} b_{y,\downarrow} \xi_{\lambda}  \rangle + \text{c.c.,} \nonumber
\end{eqnarray}
with $f = 1 - \varphi$. Recall that $\varphi$ is the periodization of $\varphi_{\infty}$, the solution of the Neumann problem (\ref{eq:scat2}), over $\Lambda_{L}$: $\varphi(x) = \sum_{n \in \mathbb{Z}^{3}} \varphi_{\infty}(x + n_{1} e_{1} L + n_{2} e_{2} L + n_{3} e_{3} L)$. The function $\varphi_{\infty}(x)$ is compactly supported in $\mathbb{R}^{3}$, with support in $B = \{ x\in \mathbb{R}^{3} \mid |x| \leq \rho^{-\gamma} \}$. Thus, up to a boundary term we can replace $f$ with $f_{\infty} = 1 - \varphi_{\infty}$:
\begin{eqnarray}
&&\int_{\Lambda_{L} \times \Lambda_{L}} dxdy\, \theta(|x-y|_{L}<\rho^{-\gamma}) \Big(-\Delta f(x-y) + \frac{1}{2}V(x-y) f(x-y)\Big) \langle \xi_{\lambda}, b_{x,\uparrow} b_{y,\downarrow} \xi_{\lambda}  \rangle\\
&&\quad = \int_{\tilde \Lambda_{L} \times \tilde \Lambda_{L}} dxdy\, \theta(|x-y|_{L}<\rho^{-\gamma}) \Big(-\Delta f_{\infty}(x-y) + \frac{1}{2}V_{\infty}(x-y) f_{\infty}(x-y)\Big) \langle \xi_{\lambda}, b_{x,\uparrow} b_{y,\downarrow} \xi_{\lambda}  \rangle + \frak{e}_{L}\nonumber
\end{eqnarray}
where $\tilde \Lambda_{L} = [\rho^{-\gamma}, L - \rho^{-\gamma}]^{3}$, and $\frak{e}_{L}$ is a boundary term:
\begin{equation}\label{eq:eLbd}
| \frak{e}_{L} | \leq CL^{2} \rho^{-4\gamma} ( \| \theta(|\cdot|_{L} < \rho^{-\gamma})\Delta \varphi \|_{\infty} + \|\varphi\|_{\infty} ) \rho^{-3\beta + 1}\;.
\end{equation}
Therefore, using that $f_{\infty}$ solves the scattering equation in a ball of radius $\rho^{-\gamma}$, Eq. (\ref{eq:scat2}), we easily get:
\begin{eqnarray}\label{eq:scatTTQ}
&&\langle \xi_{\lambda}, ( \mathbb{T}_{1} + \mathbb{T}^{\text{r}}_{2} + \widetilde{\mathbb{Q}}^{\text{r}}_{4}) \xi_{\lambda} \rangle \\
&& = 2\lambda_{\gamma} \int_{\tilde \Lambda_{L} \times \tilde \Lambda_{L}} dxdy\, \theta(|x - y|_{L} < \rho^{-\gamma})f_{\infty}(x-y) \langle \xi_{\lambda}, (b_{x,\uparrow} b_{y,\downarrow} + \text{h.c.}) \xi_{\lambda} \rangle + \frak{e}_{L} \nonumber\\
&& = 2\lambda_{\gamma} \int_{\Lambda_{L} \times \Lambda_{L}} dxdy\, \theta(|x - y|_{L} < \rho^{-\gamma})f(x-y) \langle \xi_{\lambda}, (b_{x,\uparrow} b_{y,\downarrow} +  \text{h.c.})\xi_{\lambda}  \rangle + \frak{e}_{L}\;,\nonumber
\end{eqnarray}
up to a redefinition of the boundary term (still satisfying (\ref{eq:eLbd})). To conclude, we estimate the integral using Lemma \ref{lem:bosonbd}, setting $g(x-y) = \lambda_{\gamma} \theta(|x - y|_{L} < \rho^{-\gamma})f(x-y)$. We shall use that:
\begin{equation}
\| g \|_{1} \leq C\;,\qquad \| g \|_{2} \leq C\rho^{\frac{3\gamma}{2}}\;.
\end{equation}
We get, thanks to Lemma \ref{lem:bosonbd}:
\begin{eqnarray}
&&\Big|\int dxdy\, g(x-y) \langle \xi_{\lambda}, a_{\uparrow}(\overline{v}^{\text{r}}_{x}) a_{\uparrow}(u^{\text{r}}_{x}) a_{\downarrow}(\overline{v}^{\text{r}}_{y}) a_{\downarrow}(u^{\text{r}}_{y}) \xi_{\lambda}  \rangle\Big| \nonumber\\
&&\qquad\qquad \leq \int dy\, \| b^{*}_{\uparrow}(g_{y}) \xi_{\lambda} \| \| b_{\downarrow,y} \xi_{\lambda} \|\nonumber\\
&&\qquad\qquad \leq \int dy\, \| b_{\uparrow}(g_{y}) \xi_{\lambda} \| \| b_{\downarrow,y} \xi_{\lambda} \| + C\rho \|g\|_{2} \int dy\, \| a_{\downarrow}(u_{y}) \xi_{\lambda}\|\nonumber\\
&& \qquad\qquad \leq C\rho\int dxdy\, |g(x-y)| \| a_{\uparrow}(u_{x}) \xi_{\lambda} \| \| a_{\downarrow}(u_{y}) \xi_{\lambda} \| + CL^{\frac{3}{2}} \rho^{1 + \frac{3\gamma}{2}} \| \mathcal{N}^{\frac{1}{2}}\xi_{\lambda} \|\;.
\end{eqnarray}
Thus, by Cauchy-Schwarz inequality, using that $\| g \|_{1} \leq C$:
\begin{eqnarray}
&&\Big|\int dxdy\, g(x-y) \langle \xi_{\lambda}, a_{\uparrow}(\overline{v}^{\text{r}}_{x}) a_{\uparrow}(u^{\text{r}}_{x}) a_{\downarrow}(\overline{v}^{\text{r}}_{y}) a_{\downarrow}(u^{\text{r}}_{y}) \xi_{\lambda}  \rangle\Big| \nonumber\\
&&\qquad \qquad \leq C\rho \langle \xi_{\lambda}, \mathcal{N}\xi_{\lambda} \rangle + CL^{\frac{3}{2}} \rho^{1 + \frac{3\gamma}{2}} \| \mathcal{N}^{\frac{1}{2}}\xi_{\lambda} \|\nonumber\\
&&\qquad \qquad \leq CL^{\frac{3}{2}} \rho^{1 + \frac{3\gamma}{2}} \| \mathcal{N}^{\frac{1}{2}}\xi_{\lambda} \|\;,
\end{eqnarray}
where in the last step we used the propagation of the a priori estimate for the number operator, and the assumption $\gamma \leq 7/18$. This concludes the proof.
\end{proof}
\subsection{Propagation of the estimates}
We now have all the ingredients needed in order to propagate the a priori estimates for $\mathbb{H}_{0}$ and for $\mathbb{Q}_{1}$. 
\begin{proposition}[Propagation estimate for $\mathbb{H}_{0}$, $\mathbb{Q}_{1}$ - Part 2.]\label{prp:prop2} Let $5/18\leq \gamma \leq 1/3$. Under the assumptions of Theorem \ref{thm:main} the following is true. For all $\lambda \in [0,1]$:
\begin{equation}\label{eq:derA}
\langle \xi_{\lambda}, \mathbb{H}_{0}\xi_{\lambda}\rangle \leq CL^{3} \rho^{2}\;,\qquad \langle \xi_{\lambda}, \mathbb{Q}_{1} \xi_{\lambda} \rangle \leq CL^{3} \rho^{2}\;,\qquad  \langle \xi_{\lambda}, \widetilde{\mathbb{Q}}_{1} \xi_{\lambda} \rangle \leq CL^{3} \rho^{2}\;.
\end{equation}
\end{proposition}
\begin{proof} The last bound immediately follows from the second, using that $ \widetilde{\mathbb{Q}}_{1} \leq \mathbb{Q}_{1}$. Let us prove the first two bounds. Using Eqs. (\ref{eq:contrH0}) we get:
\begin{eqnarray}
\frac{d}{d\lambda} \langle \xi_{\lambda}, (\mathbb{H}_{0} + \mathbb{Q}_{1}) \xi_{\lambda} \rangle &=& -\langle \xi_{\lambda}, (\mathbb{T}_{1} + \mathbb{T}_{2}) \xi_{\lambda}\rangle + \mathcal{E}_{\mathbb{H}_{0}}(\xi_{\lambda}) + \mathcal{E}_{\mathbb{Q}_{1}}(\xi_{\lambda})\\
&=& -\langle \xi_{\lambda}, (\mathbb{T}_{1} + \mathbb{T}_{2} + \widetilde{\mathbb{Q}}^{\text{r}}_{4}) \xi_{\lambda}\rangle + \langle \xi_{\lambda}, \widetilde{\mathbb{Q}}^{\text{r}}_{4} \xi_{\lambda}\rangle + \mathcal{E}_{\mathbb{H}_{0}}(\xi_{\lambda}) + \mathcal{E}_{\mathbb{Q}_{1}}(\xi_{\lambda})\;.\nonumber
\end{eqnarray}
As proven in Appendix \ref{app:ab}:
\begin{equation}
| \langle \xi_{\lambda}, (\widetilde{\mathbb{Q}}_{4} - \widetilde{\mathbb{Q}}_{4}^{\text{r}}) \xi_{\lambda}\rangle | \leq \frac{C}{\delta}\rho^{2 + \frac{\epsilon}{3}} + \delta \langle \xi_{\lambda}, \widetilde{\mathbb{Q}}_{1} \xi_{\lambda}\rangle\;.
\end{equation}
Also, thanks to (\ref{eq:bdQ4}):
\begin{equation}
\pm \widetilde{\mathbb{Q}}_{4} \leq \delta \widetilde{\mathbb{Q}}_{1} + (C/\delta) L^{3} \rho^{2}\;.
\end{equation}
Therefore, using that $\widetilde{\mathbb{Q}}_{1} \leq \mathbb{Q}_{1}$:
\begin{equation}\label{eq:derHQ}
\frac{d}{d\lambda} \langle \xi_{\lambda}, ( \mathbb{H}_{0} + \mathbb{Q}_{1} ) \xi_{\lambda}\rangle \leq -\langle \xi_{\lambda}, (\mathbb{T}_{1} + \mathbb{T}_{2} + \widetilde{\mathbb{Q}}^{\text{r}}_{4}) \xi_{\lambda}\rangle + C\langle \xi_{\lambda}, \mathbb{Q}_{1} \xi_{\lambda} \rangle + CL^{3} \rho^{2} + \mathcal{E}_{\mathbb{H}_{0}}(\xi_{\lambda}) + \mathcal{E}_{\mathbb{Q}_{1}}(\xi_{\lambda})\;.
\end{equation}
To estimate the various terms, we shall use the bounds of Propositions \ref{prp:ders}, \ref{prp:cancscat}. We have, for $5/18\leq \gamma \leq 1/3$, from Eqs. (\ref{eq:ETerr}), (\ref{eq:lambdagamma0}):
\begin{eqnarray}\label{eq:TTQgro}
|\langle \xi_{\lambda}, (\mathbb{T}_{1} + \mathbb{T}_{2} + \widetilde{\mathbb{Q}}^{\text{r}}_{4}) \xi_{\lambda}\rangle| &\leq& CL^{\frac{3}{2}}\rho^{1 + \frac{3\gamma}{2}} \| \mathcal{N}^{\frac{1}{2}} \xi_{\lambda} \|\nonumber\\
&\leq& CL^{3}\rho^{2}\;,
\end{eqnarray}
where we used the propagation of the a priori estimate for number operator, Eq. (\ref{eq:propN}). The bound (\ref{eq:TTQgro}) is the only point where we need the lower bound on $\gamma$.

Consider now $\mathcal{E}_{\mathbb{H}_{0}}(\xi_{\lambda})$. We have, from the first bound in (\ref{eq:errH0}):
\begin{eqnarray}\label{eq:EH0gro}
\mathcal{E}_{\mathbb{H}_{0}}(\xi_{\lambda}) &\leq& CL^{\frac{3}{2}} \rho^{\frac{7}{6}} \| \mathbb{H}_{0}^{\frac{1}{2}} \xi_{\lambda} \| + CL^{\frac{3}{2}} \rho^{\frac{3}{2}} \|\mathcal{N}^{\frac{1}{2}} \xi_{\lambda}\| \nonumber\\
&\leq& C\langle \xi_{\lambda}, \mathbb{H}_{0} \xi_{\lambda} \rangle + CL^{3} \rho^{2}\;,
\end{eqnarray}
where we used again the propagation of the apriori estimate for number operator (\ref{eq:prop76}). Also, from the third of (\ref{eq:errH0}), it is not difficult to see that, for $\frac{4}{3}\gamma - \frac{7}{18} < \eta < \gamma$ (which is a nonempty set for $\eta$, since $\gamma \leq 1/3$), and for $\beta$ small enough:
\begin{equation}\label{eq:EQ1gro}
\mathcal{E}_{\mathbb{Q}_{1}}(\xi_{\lambda}) \leq C \langle \xi_{\lambda}, (\mathbb{H}_{0} + \mathbb{Q}_{1}) \xi_{\lambda}\rangle + CL^{3} \rho^{2}\;.
\end{equation}
The final claim follows from Eqs. (\ref{eq:derHQ})-(\ref{eq:EQ1gro}), together with Gr\"onwall lemma.
\end{proof}
Let us rewrite the bounds of Propositions \ref{prp:ders}, \ref{prp:cancscat}, using the propagation of the a priori estimates (\ref{eq:derA}). We have, for $5/18\leq \gamma \leq 1/3$ and $\eta < \gamma$:
\begin{eqnarray}\label{eq:bdsfin1}
| \mathcal{E}_{\mathbb{H}_{0}}(\xi_{\lambda}) | &\leq& CL^{3} \rho^{\frac{7}{3} - \frac{\gamma}{2}} + CL^{\frac{3}{2}} \rho^{\frac{5}{3} - \frac{\gamma}{2}} \|\mathcal{N}^{\frac{1}{2}} \xi_{\lambda}\|\nonumber\\
|\mathcal{E}_{\widetilde{\mathbb{Q}}_{1}}(\xi_{\lambda})| &\leq& C_{\beta} L^{3} \rho^{3 - 2\gamma - \eta} + C_{\beta} L^{\frac{3}{2}} \rho^{2 + \frac{3\eta}{2} - 2\gamma} \| \mathcal{N}^{\frac{1}{2}}\xi_{\lambda} \| \nonumber\\
|\mathcal{E}_{\mathbb{Q}_{1}}(\xi_{\lambda})| &\leq& C_{\beta} L^{3} \rho^{3 - 2\gamma - \eta} + C_{\beta} L^{\frac{3}{2}} \rho^{2 + \frac{3\eta}{2} - 2\gamma} \| \mathcal{N}^{\frac{1}{2}}\xi_{\lambda} \|\;,
\end{eqnarray}
and:
\begin{equation}\label{eq:bdsfin2}
|\langle \xi_{\lambda}, ( \mathbb{T}_{1} + \mathbb{T}_{2} + \widetilde{\mathbb{Q}}^{\text{r}}_{4}) \xi_{\lambda} \rangle| \leq C L^{\frac{3}{2}} \rho^{1 + \frac{3\gamma}{2}} \| \mathcal{N}^{\frac{1}{2}} \xi_{\lambda} \|\;.
\end{equation}
The reason why we kept the dependence on the number operator is that the estimate on $\|\mathcal{N}^{\frac{1}{2}} \xi_{\lambda}\|$ obtained propagating the a priori bound for $\mathcal{N}$ on $\xi_{0}$ is not optimal, in contrast to the estimates for $\mathbb{H}_{0}$ and $\mathbb{Q}_{1}$. We shall conclude the section by proving an improved version of the bound for the number operator.
\begin{proposition}[Improved a priori estimate for the number operator.]\label{prp:Nimpro} Let $\gamma \leq 1/2$. We have:
\begin{equation}\label{eq:Nimpro}
\langle \xi_{\lambda}, \mathcal{N} \xi_{\lambda} \rangle \leq C L^{\frac{3}{2}} \rho^{\frac{1}{6}} \| \mathbb{H}_{0}^{\frac{1}{2}} \xi_{1} \| + CL^{3} \rho^{2 - \gamma}\;.
\end{equation}
\end{proposition}
\begin{proof} 
Lemma \ref{prp:propN}, together with Gronwall lemma, immediately implies:
\begin{equation}
\langle \xi_{\lambda}, \mathcal{N} \xi_{\lambda} \rangle \leq \langle \xi_{1}, \mathcal{N} \xi_{1} \rangle + CL^{3} \rho^{2 - \gamma}\;.
\end{equation}
To estimate $\langle \xi_{1}, \mathcal{N} \xi_{1} \rangle$ in terms of the kinetic energy, we use Lemma \ref{lem:apriori}:
\begin{equation}
\langle \xi_{\lambda}, \mathcal{N} \xi_{\lambda} \rangle \leq CL^{3} \rho^{\frac{1}{3} + \alpha} + \frac{1}{\rho^{\alpha}}\langle \xi_{1}, \mathbb{H}_{0} \xi_{1} \rangle + CL^{3} \rho^{2 - \gamma}\;.
\end{equation}
We choose $\alpha$ such that:
\begin{equation}
\rho^{\frac{1}{3} + \alpha} = \max\Big\{ \frac{1}{L^{3} \rho^{\alpha}} \langle \xi_{1}, \mathbb{H}_{0}\xi_{1} \rangle, \rho^{2 - \gamma} \Big\}\;.
\end{equation} 
By doing so, we get:
\begin{equation}
\langle \xi_{\lambda}, \mathcal{N} \xi_{\lambda} \rangle \leq C L^{\frac{3}{2}} \rho^{\frac{1}{6}} \| \mathbb{H}_{0}^{\frac{1}{2}} \xi_{1} \| + CL^{3} \rho^{2 - \gamma}\;.
\end{equation} 
\end{proof}
\noindent{\bf Improved bounds for the error terms.} To conclude the section, let us reexpress the bounds (\ref{eq:bdsfin1}), (\ref{eq:bdsfin2}) in view of the improved estimate (\ref{eq:Nimpro}). Take $5/18\leq \gamma \leq 1/3$, in order to be able to use the propagation of the a priori estimates (\ref{prp:prop2}). The bound (\ref{eq:Nimpro}) implies:
\begin{equation}\label{eq:interp1}
\| \mathcal{N}^{\frac{1}{2}} \xi_{\lambda} \| \leq C L^{\frac{3}{4}} \rho^{\frac{1}{12}} \| \mathbb{H}_{0}^{\frac{1}{2}} \xi_{1} \|^{\frac{1}{2}} + C L^{\frac{3}{2}} \rho^{1 - \frac{\gamma}{2}}\;.
\end{equation}
Plugging this bound in the first of (\ref{eq:bdsfin1}) we get, for $0<\delta < 1$:
\begin{eqnarray}\label{eq:errH00}
| \mathcal{E}_{\mathbb{H}_{0}}(\xi_{\lambda}) | &\leq& CL^{3} \rho^{\frac{7}{3} - \frac{\gamma}{2}} + CL^{\frac{9}{4}} \rho^{\frac{7}{4} - \frac{\gamma}{2}} \| \mathbb{H}_{0}^{\frac{1}{2}} \xi_{1} \|^{\frac{1}{2}}\nonumber\\
&\leq& C_{\delta}L^{3} \rho^{\frac{7}{3} - \frac{2\gamma}{3}} + \delta \langle \xi_{1}, \mathbb{H}_{0} \xi_{1}\rangle\;,
\end{eqnarray}
where in the last step we used Young's inequality $|ab| \leq C_{pq} (|a|^{p} + |b|^{q})$ with $1/p + 1/q = 1$, with $p=4$ and $q=4/3$. Similarly, using  the bound (\ref{eq:interp1}) in the second of (\ref{eq:bdsfin1}):
\begin{eqnarray}\label{eq:errQ1}
|\mathcal{E}_{\widetilde{\mathbb{Q}}_{1}}(\xi_{\lambda})| &\leq& C_{\beta} L^{3} \rho^{3 - 2\gamma - \eta} + C_{\beta} L^{3} \rho^{3 + \frac{3\eta}{2} - \frac{5\gamma}{2}} + C_{\beta} L^{\frac{9}{4}} \rho^{\frac{25}{12} + \frac{3\eta}{2} - 2\gamma} \| \mathbb{H}_{0}^{\frac{1}{2}} \xi_{1}\|^{\frac{1}{2}}\nonumber\\
&\leq& C_{\beta} L^{3} \rho^{3 - 2\gamma - \eta} + C_{\beta} L^{3} \rho^{3 + \frac{3\eta}{2} - \frac{5\gamma}{2}} + C_{\beta,\delta} L^{3} \rho^{\frac{25}{9} + 2\eta - \frac{8}{3}\gamma} + \delta\langle \xi_{1}, \mathbb{H}_{0} \xi_{1}\rangle\;.
\end{eqnarray}
For $\eta \leq \frac{4}{9} + \frac{\gamma}{3}$, which is implied by $\eta < \gamma$ and $\gamma \leq 1/3$:
\begin{eqnarray}\label{eq:errQ11}
|\mathcal{E}_{\widetilde{\mathbb{Q}}_{1}}(\xi_{\lambda})| &\leq& C_{\beta} L^{3} \rho^{3 - 2\gamma - \eta} + C_{\beta,\delta} L^{3} \rho^{\frac{25}{9} + 2\eta - \frac{8}{3}\gamma} + \delta\langle \xi_{1}, \mathbb{H}_{0} \xi_{1}\rangle\nonumber\\
&\leq& C_{\beta} L^{3} \rho^{\frac{79}{27} - \frac{20}{9}\gamma} + \delta\langle \xi_{1}, \mathbb{H}_{0} \xi_{1}\rangle\;,
\end{eqnarray}
where in the last step we optimized over $\eta$, $\eta = \frac{2}{27} + \frac{2}{9}\gamma$, which is strictly less than $\frac{4}{9} + \frac{\gamma}{3}$. The same bound holds for $|\mathcal{E}_{\mathbb{Q}_{1}}(\xi_{\lambda})|$. Finally, for $\gamma \leq 1/3$, plugging the bound (\ref{eq:interp1}) in (\ref{eq:bdsfin2}):
\begin{eqnarray}\label{eq:interp2}
|\langle \xi_{\lambda}, ( \mathbb{T}_{1} + \mathbb{T}_{2} + \widetilde{\mathbb{Q}}^{\text{r}}_{4}) \xi_{\lambda} \rangle| &\leq& CL^{3} \rho^{2 + \gamma} + CL^{\frac{9}{4}} \rho^{\frac{13}{12} + \frac{3\gamma}{2}} \|\mathbb{H}_{0}^{\frac{1}{2}} \xi_{1}\|^{\frac{1}{2}}  \nonumber\\
&\leq&  C_{\delta}L^{3} \rho^{\frac{13}{9} + 2\gamma} + \delta \langle \xi_{1}, \mathbb{H}_{0} \xi_{1}\rangle\;.
\end{eqnarray}
These bounds will play an important role in the proof of the lower bound for the ground state energy, discussed in the next section.

\section{Lower bound on the ground state energy}\label{sec:lwbd}
In this section we shall prove the lower bound for the ground state energy of the dilute Fermi gas. In what follows, $\psi$ will be approximate ground state, in the sense of Definition \ref{def:appgs}. In the following we shall always assume that $5/18\leq \gamma \leq 1/3$, which is the range of values of $\gamma$ for which the estimates (\ref{eq:interp1})-(\ref{eq:interp2}) hold.

The starting point is, recall Proposition \ref{prp:conj} and the bounds of Proposition \ref{prp:bogbd}:
\begin{equation}\label{eq:begin}
\langle \psi, \mathcal{H} \psi \rangle \geq E_{\text{HF}}(\omega) + \langle \xi_{0}, (\mathbb{H}_{0} + \widetilde{\mathbb{Q}}_{1} + \widetilde{\mathbb{Q}}_{4}) \xi_{0} \rangle + \mathcal{E}_{\text{1}}(\psi)\;,
\end{equation}
where the error term $\mathcal{E}_{1}(\psi)$ is bounded as:
\begin{equation}
| \mathcal{E}_{1}(\psi) | \leq C\rho^{\alpha} \langle \xi_{0}, \widetilde{\mathbb{Q}}_{1} \xi_{0}\rangle + C\rho^{1 - \alpha} \langle \xi_{0}, \mathcal{N} \xi_{0} \rangle\;.
\end{equation}
From the estimate (\ref{eq:Nimpro})  for the number operator, together with the a priori estimate (\ref{eq:aprQ1}) for $\widetilde{\mathbb{Q}}_{1}$, we get, optimizing over $\alpha$, $\alpha = 1/9$:
\begin{eqnarray}
| \mathcal{E}_{1}(\psi) | \leq C_{\delta}L^{3} \rho^{2 + \frac{1}{9}} + \delta \langle \xi_{1}, \mathbb{H}_{0} \xi_{1} \rangle\;.
\end{eqnarray}
To extract the correlation energy at order $\rho^{2}$ we shall use an interpolation argument. We write:
\begin{eqnarray}\label{eq:step1}
\langle \xi_{0}, (\mathbb{H}_{0} + \widetilde{\mathbb{Q}}_{1} + \widetilde{\mathbb{Q}}_{4}) \xi_{0} \rangle &=&  \langle \xi_{1}, (\mathbb{H}_{0} + \widetilde{\mathbb{Q}}_{1}) \xi_{1} \rangle - \int_{0}^{1} d\lambda\, \frac{d}{d\lambda} \langle \xi_{\lambda}, (\mathbb{H}_{0} + \widetilde{\mathbb{Q}}_{1})  \xi_{\lambda}\rangle + \langle \xi_{0}, \widetilde{\mathbb{Q}}_{4} \xi_{0} \rangle \nonumber\\
&=&  \langle \xi_{1}, (\mathbb{H}_{0} + \widetilde{\mathbb{Q}}_{1}) \xi_{1} \rangle + \int_{0}^{1} d\lambda\, \langle \xi_{\lambda}, (\mathbb{T}_{1} + \mathbb{T}_{2}) \xi_{\lambda} \rangle + \langle \xi_{0}, \widetilde{\mathbb{Q}}_{4} \xi_{0} \rangle + \mathcal{E}_{2}(\psi)\;,
\end{eqnarray}
where the $\mathbb{T}_{1}, \mathbb{T}_{2}$ operators are defined in (\ref{eq:Tdef}) and, thanks to the bounds (\ref{eq:errH00}), (\ref{eq:errQ11}):
\begin{eqnarray}
|\mathcal{E}_{2}(\psi)| &\leq& \max_{\lambda \in [0;1]} |\mathcal{E}_{\mathbb{H}_{0}}(\xi_{\lambda})| + \max_{\lambda \in [0;1]} |\mathcal{E}_{\widetilde{\mathbb{Q}}_{1}}(\xi_{\lambda})| \nonumber\\
&\leq& C_{\delta}\rho^{\frac{7}{3} - \frac{2\gamma}{3}} + \delta\langle \xi_{1}, \mathbb{H}_{0}\xi_{1}\rangle\;,
\end{eqnarray}
where we used the condition $\gamma \leq 1/3$. Next, in order to make use of the cancellation due to the scattering equation, we rewrite:
\begin{eqnarray}\label{eq:intt}
\langle \xi_{0}, (\mathbb{H}_{0} + \widetilde{\mathbb{Q}}_{1} + \widetilde{\mathbb{Q}}_{4}) \xi_{0} \rangle &=& \langle \xi_{1}, (\mathbb{H}_{0} + \widetilde{\mathbb{Q}}_{1}) \xi_{1} \rangle + \int_{0}^{1} d\lambda\, \langle \xi_{\lambda}, (\mathbb{T}_{1} + \mathbb{T}_{2} + \widetilde{\mathbb{Q}}_{4}^{\text{r}}) \xi_{\lambda} \rangle\nonumber\\
&& - \int_{0}^{1} d\lambda\, \langle \xi_{\lambda}, \widetilde{\mathbb{Q}}_{4}^{\text{r}} \xi_{\lambda} \rangle + \langle \xi_{0}, \widetilde{\mathbb{Q}}_{4} \xi_{0} \rangle + \mathcal{E}_{2}(\psi)\;,\nonumber\\
&\equiv& \langle \xi_{1}, (\mathbb{H}_{0} + \widetilde{\mathbb{Q}}_{1}) \xi_{1} \rangle - \int_{0}^{1} d\lambda\, \langle \xi_{\lambda}, \widetilde{\mathbb{Q}}_{4}^{\text{r}} \xi_{\lambda} \rangle + \langle \xi_{0}, \widetilde{\mathbb{Q}}_{4} \xi_{0} \rangle + \mathcal{E}_{3}(\psi) + \mathcal{E}_{2}(\psi)\;,
\end{eqnarray}
where, using (\ref{eq:interp2}):
\begin{eqnarray}\label{eq:E3}
|\mathcal{E}_{3}(\psi)| &\leq& \max_{\lambda \in [0;1]} |\langle \xi_{\lambda}, (\mathbb{T}_{1} + \mathbb{T}_{2} + \widetilde{\mathbb{Q}}_{4}^{\text{r}}) \xi_{\lambda} \rangle|\nonumber\\
&\leq& C_{\delta}L^{3} \rho^{\frac{13}{9} + 2\gamma} + \delta \langle \xi_{1}, \mathbb{H}_{0} \xi_{1}\rangle\;.
\end{eqnarray}
Let us now consider the second and third term in Eq. (\ref{eq:intt}). We rewrite it as:
\begin{eqnarray}\label{eq:derQ4}
- \int_{0}^{1} d\lambda\, \langle \xi_{\lambda}, \widetilde{\mathbb{Q}}_{4}^{\text{r}} \xi_{\lambda} \rangle + \langle \xi_{0}, \widetilde{\mathbb{Q}}_{4} \xi_{0} \rangle &=& \langle \xi_{1},(\widetilde{\mathbb{Q}}_{4} - \widetilde{\mathbb{Q}}^{\text{r}}_{4}) \xi_{1} \rangle - \int_{0}^{1}d\lambda\, \frac{d}{d\lambda} \langle \xi_{\lambda} \widetilde{\mathbb{Q}}_{4} \xi_{\lambda}\rangle + \int_{0}^{1} d\lambda \int_{\lambda}^{1} d\lambda' \frac{d}{d\lambda'} \langle \xi_{\lambda'}, \widetilde{\mathbb{Q}}_{4}^{\text{r}} \xi_{\lambda'} \rangle \nonumber\\
&\equiv& \mathcal{E}_{4}(\psi)  - \int_{0}^{1}d\lambda\, \frac{d}{d\lambda} \langle \xi_{\lambda} \widetilde{\mathbb{Q}}_{4} \xi_{\lambda}\rangle + \int_{0}^{1} d\lambda \int_{\lambda}^{1} d\lambda' \frac{d}{d\lambda'} \langle \xi_{\lambda'}, \widetilde{\mathbb{Q}}_{4}^{\text{r}} \xi_{\lambda'} \rangle\;,
\end{eqnarray}
where, as proven in Appendix \ref{app:ab}:
\begin{eqnarray}
| \mathcal{E}_{4}(\psi) | &=& |\langle \xi_{1},(\widetilde{\mathbb{Q}}_{4} - \widetilde{\mathbb{Q}}^{\text{r}}_{4}) \xi_{1} \rangle| \nonumber\\
&\leq& \frac{C}{\delta} L^{3} \rho^{2 + \frac{\epsilon}{3}} +\delta \langle \xi_{1}, \widetilde{\mathbb{Q}}_{1} \xi_{1}\rangle\;.
\end{eqnarray}
The correlation energy at order $\rho^{2}$ arises from the last two terms in Eq. (\ref{eq:derQ4}). We will prove that:
\begin{eqnarray}\label{eq:commQ40}
\frac{d}{d\lambda}\langle \xi_{\lambda}, \widetilde{\mathbb{Q}}_{4} \xi_{\lambda} \rangle &=& 2 L^{3}\rho_{\uparrow} \rho_{\downarrow} \int dx\, V(x) \varphi(x) + \text{subleading terms.} \nonumber\\
\frac{d}{d\lambda}\langle \xi_{\lambda}, \widetilde{\mathbb{Q}}^{\text{r}}_{4} \xi_{\lambda} \rangle &=& 2 L^{3}\rho_{\uparrow} \rho_{\downarrow} \int dx\, V(x) \varphi(x) + \text{subleading terms.}
\end{eqnarray}
The explicit terms are precisely what we need to compute the ground state energy at order $\rho^{2}$. The quantitative version of the statement (\ref{eq:commQ40}) is the content of the next proposition.
\begin{proposition}[Extracting the correlation energy.]\label{prp:intQ4} Under the assumptions of Theorem \ref{thm:main} the following holds. Let $\psi$ be an approximate ground state. Take $\frac{5}{18} \leq \gamma \leq \frac{1}{3}$, $\epsilon \geq 0$ and $\frac{1}{6} + \frac{\epsilon}{12} < \gamma$. Then:
\begin{eqnarray}\label{eq:derQQ4}
\frac{d}{d\lambda}\langle \xi_{\lambda}, \widetilde{\mathbb{Q}}^{\text{r}}_{4} \xi_{\lambda} \rangle &=& 2\rho^{\text{r}}_{\uparrow}\rho^{\text{r}}_{\downarrow} \int dxdy\, V(x-y) \varphi(x-y) + \mathcal{E}_{\widetilde{\mathbb{Q}}^{\text{r}}_{4}}(\xi_{\lambda}) \nonumber\\
\frac{d}{d\lambda}\langle \xi_{\lambda}, \widetilde{\mathbb{Q}}_{4} \xi_{\lambda} \rangle &=& 2\rho^{\text{r}}_{\uparrow}\rho^{\text{r}}_{\downarrow} \int dxdy\, V(x-y) \varphi(x-y) + \mathcal{E}_{\widetilde{\mathbb{Q}}_{4}}(\xi_{\lambda})\;,
\end{eqnarray}
where, for any $0<\delta< 1$:
\begin{eqnarray}\label{eq:errQ4}
|\mathcal{E}_{\widetilde{\mathbb{Q}}^{\text{r}}_{4}}(\psi)| &\leq& CL^{3} \rho^{\frac{7}{3}} + C_{\beta,\alpha} L^{3}\rho^{\frac{26}{9} - \frac{4}{3}\gamma - \frac{7}{18}\epsilon} + \delta \langle \xi_{1}, \mathbb{H}_{0} \xi_{1}\rangle\nonumber\\
|\mathcal{E}_{\widetilde{\mathbb{Q}}_{4}}(\psi)| &\leq& CL^{3} \rho^{\frac{7}{3}} + C_{\beta,\alpha} L^{3}\rho^{\frac{26}{9} - \frac{4}{3}\gamma - \frac{7}{18}\epsilon} + \delta \langle \xi_{1}, \mathbb{H}_{0} \xi_{1}\rangle\;.
\end{eqnarray}
\end{proposition}
\begin{remark}[Notations.] Unless otherwise stated, with a slight abuse of notation in the following we will use the notation $u^{\text{r}}(x;y)$ to denote the function $(u^{\text{r}})^{2}(x;y)$. The two functions satisfy the same estimates, recall Proposition \ref{prp:decreg}.
\end{remark}
\begin{proof} We shall only discuss the proof of the statement concerning $\widetilde{\mathbb{Q}}^{\text{r}}_{4}$, the one for $\widetilde{\mathbb{Q}}_{4}$ being completely analogous (it is actually simpler). We write:
\begin{equation}
\frac{d}{d\lambda} \langle \xi_{\lambda}, \widetilde{\mathbb{Q}}_{4}^{\text{r}} \xi_{\lambda}\rangle = \langle \xi_{\lambda}, [ \widetilde{\mathbb{Q}}^{\text{r}}_{4}, B] \xi_{\lambda}\rangle + \text{c.c.}\;,
\end{equation}
where:
\begin{eqnarray}\label{eq:1/2}
[ \widetilde{\mathbb{Q}}^{\text{r}}_{4}, B] = \frac{1}{2} \sum_{\sigma\neq \sigma'} \int dxdydzdz'\, V(x-y)\varphi(z-z')  [ a^{*}_{\sigma}(u^{\text{r}}_{x}) a^{*}_{\sigma'}(u^{\text{r}}_{y}) a^{*}_{\sigma'}(\overline{v}^{\text{r}}_{y}) a^{*}_{\sigma}(\overline{v}^{\text{r}}_{x}),  a_{\uparrow}(u^{\text{r}}_{z}) a_{\uparrow}(\overline{v}^{\text{r}}_{z}) a_{\downarrow}(u^{\text{r}}_{z'}) a_{\downarrow}(\overline{v}^{\text{r}}_{z'})]\;.\nonumber
\end{eqnarray}
We rewrite the commutator as:
\begin{eqnarray}\label{eq:bigcom}
&&\big[ a^{*}_{\sigma}(u^{\text{r}}_{x}) a^{*}_{\sigma'}(u^{\text{r}}_{y}) a^{*}_{\sigma'}(\overline{v}^{\text{r}}_{y}) a^{*}_{\sigma}(\overline{v}^{\text{r}}_{x}),  a_{\uparrow}(u^{\text{r}}_{z}) a_{\uparrow}(\overline{v}^{\text{r}}_{z}) a_{\downarrow}(u^{\text{r}}_{z'}) a_{\downarrow}(\overline{v}^{\text{r}}_{z'}) \big] \\
&&\qquad = -a^{*}_{\sigma}(u^{\text{r}}_{x}) a^{*}_{\sigma'}(u^{\text{r}}_{y}) \big[  a^{*}_{\sigma'}(\overline{v}^{\text{r}}_{y}) a^{*}_{\sigma}(\overline{v}^{\text{r}}_{x}),   a_{\uparrow}(\overline{v}^{\text{r}}_{z})  a_{\downarrow}(\overline{v}^{\text{r}}_{z'})\big] a_{\uparrow}(u^{\text{r}}_{z}) a_{\downarrow}(u^{\text{r}}_{z'})\nonumber\\
&&\qquad \quad - a_{\uparrow}(\overline{v}^{\text{r}}_{z})a_{\downarrow}(\overline{v}^{\text{r}}_{z'}) \big[ a^{*}_{\sigma}(u^{\text{r}}_{x}) a^{*}_{\sigma'}(u^{\text{r}}_{y}) ,  a_{\uparrow}(u^{\text{r}}_{z})  a_{\downarrow}(u^{\text{r}}_{z'}) \big] a^{*}_{\sigma'}(\overline{v}^{\text{r}}_{y}) a^{*}_{\sigma}(\overline{v}^{\text{r}}_{x})\;.\nonumber
\end{eqnarray}
As it will be clear, the only terms contributing at the order $\rho^{2}$ will be those anti-normal ordered, which arise from the last line. The first contribution in the right-hand side of Eq. (\ref{eq:bigcom}) is partially normal ordered; as such, it will give rise to an error term. The commutator produces contractions between the fermionic operators; we have, omitting the spin label for simplicity:
\begin{eqnarray}\label{eq:onemore}
\big[  a^{*}(\overline{v}^{\text{r}}_{y}) a^{*}(\overline{v}^{\text{r}}_{x}),   a(\overline{v}^{\text{r}}_{z})  a(\overline{v}^{\text{r}}_{z'})\big] &=& a^{*}(\overline{v}^{\text{r}}_{y}) \omega^{\text{r}}(x;z) a(\overline{v}^{\text{r}}_{z'}) - a^{*}(\overline{v}^{\text{r}}_{y}) \omega^{\text{r}}(x;z') a(\overline{v}^{\text{r}}_{z}) \nonumber\\
&& + a(\overline{v}^{\text{r}}_{z'}) \omega^{\text{r}}(y;z) a^{*}(\overline{v}^{\text{r}}_{x}) - a(\overline{v}^{\text{r}}_{z}) \omega^{\text{r}}(y;z') a^{*}(\overline{v}^{\text{r}}_{x})\;.
\end{eqnarray}
These four terms give rise to contributions to (\ref{eq:bigcom}) that will be estimated in exactly the same way (the lack of normal ordering in the last two terms in Eq. (\ref{eq:onemore}) will not matter). For instance, consider:
\begin{eqnarray}\label{eq:Q4I0}
&&\text{I} := \sum_{\sigma\neq \sigma'}\int dxdydzdz'\, V(x-y) \varphi(z-z') \tilde \omega^{\text{r}}(z;y)  \langle \xi_{\lambda}, a^{*}_{\sigma}(u^{\text{r}}_{x}) a^{*}_{\sigma'}(u^{\text{r}}_{y}) a^{*}_{\sigma}(\overline{v}^{\text{r}}_{x}) a_{\downarrow}(\overline{v}^{\text{r}}_{z'}) a_{\uparrow}(u^{\text{r}}_{z}) a_{\downarrow}(u^{\text{r}}_{z'}) \xi_{\lambda}\rangle\nonumber\;.
\end{eqnarray}
To begin, we write:
\begin{equation}\label{eq:splitI}
\text{I} = \text{I}_{1} + \text{I}_{2}\;,
\end{equation}
where $\text{I}_{1}$ is obtained from $\text{I}$ replacing $u^{\text{r}}_{x}$, $u^{\text{r}}_{y}$ with $u_{x}$, $u_{y}$, and $\text{I}_{2}$ is an error term. Let us first consider the term $\text{I}_{1}$. To improve the estimate for this term, we shall study separately different contributions in momentum space. Let $0\leq \eta < \gamma \leq 1/3$, $\delta > 1$, to be chosen later. We write $u^{\text{r}}_{z} = u^{<}_{z} + u^{0}_{z} + u^{>}_{z}$, with:
\begin{equation}
\hat u^{\sharp}(k) = \hat u^{\text{r}}(k) \chi_{\sharp}(k)\;,\qquad \sharp =\, <\,, 0\,, >\;,
\end{equation}
where
\begin{equation}
\chi_{<}(k) =  \chi\Big(\frac{|k|}{\rho^{\eta}}\Big)\;,\qquad \chi_{0}(k) = \chi\Big(\frac{|k|}{\rho^{\frac{\eta}{\delta}}}\Big) - \chi\Big(\frac{|k|}{\rho^{\eta}}\Big)\;,\qquad \chi_{>}(k) = 1 - \chi\Big(\frac{|k|}{\rho^{\frac{\eta}{\delta}}}\Big)\;.
\end{equation}
Correspondingly, we write $\text{I}_{1} = \text{I}_{1}^{<} + \text{I}_{1}^{0} + \text{I}_{1}^{>}$. The key observation is that in each term $\text{I}^{\sharp}_{1}$ we can replace $\varphi$ by $\varphi_{\sharp}$, with $\hat \varphi_{\sharp}$ with similar support properties as $u_{\sharp}$. In fact, recalling that:
\begin{equation}
a_{\sigma}(u^{\text{r}}_{z}) = \sum_{k} \overline{u_{\sigma}^{\text{r}}(k)} \overline{f_{k}(z)} \hat a_{k,\sigma}\;,
\end{equation}
with $\hat a_{k,\sigma} = a_{\sigma}(f_{k})$ and $f_{k}(z) = L^{-3/2} e^{-ik\cdot z}$, we get:
\begin{equation}\label{eq:consphi0}
\int dz\, \varphi(z-z') \tilde \omega^{\text{r}}(z;y) a_{\uparrow}(u^{\text{r}}_{z}) = \frac{1}{L^{\frac{3}{2}}}\sum_{k} \overline{\hat u_{\uparrow}^{\text{r}}(k)} \hat a_{k,\uparrow} \int dz\, e^{ik\cdot z} \varphi(z-z') \tilde \omega^{\text{r}}(z;y)\;,
\end{equation}
where:
\begin{equation}\label{eq:consphi}
\int dz\, e^{ik\cdot z} \varphi(z-z') \tilde \omega^{\text{r}}(z;y) = \frac{1}{L^{3}} \sum_{q} \hat \varphi(k+q) \hat \omega^{\text{r}}(q) e^{-iz'\cdot (k+q)} e^{-iy\cdot q}\;.
\end{equation}
Suppose that $k\in \text{supp}\, \hat u^{\text{r}}_{\sharp}$. Due to the fact that $\hat \omega^{\text{r}}(q)$ is supported for $|q| \leq C\rho^{\frac{1}{3}}$ and that $\rho^{\eta} \gg \rho^{\frac{1}{3}}$, we can freely replace $\hat \varphi(k+q)$ in Eq. (\ref{eq:consphi}) with:
\begin{equation}
\hat \varphi_{\sharp}(k+q) = \hat \varphi(k+q) \tilde \chi_{\sharp}(k+q)\;,
\end{equation}
with $\tilde \chi_{\sharp}$ defined as:
\begin{equation}\label{eq:consphi2}
\tilde \chi_{<}(p) = \chi\Big(\frac{|k|}{2\rho^{\eta}}\Big)\;,\qquad \tilde \chi_{0}(p) = \chi\Big(\frac{|k|}{2\rho^{\frac{\eta}{\delta}}}\Big) - \chi\Big(\frac{2|k|}{\rho^{\eta}}\Big)\;,\qquad  \tilde \chi_{>}(p) = \chi(4\rho^{\beta}|k|) - \chi\Big(\frac{2|k|}{\rho^{\frac{\eta}{\delta}}}\Big)\;.
\end{equation}
As discussed in Appendix \ref{sec:scat}:
\begin{equation}
\| \varphi_{<} \|_{1} \leq C\rho^{-2\gamma} |\log \rho|\;,\qquad \| \varphi_{0} \|_{1} \leq C\rho^{-2\eta} |\log \rho|\;,\qquad \| \varphi_{>} \|_{1} \leq C\rho^{-\frac{2\eta}{\delta}} |\log \rho|\;.
\end{equation} 
The last two estimates improve on $\|\varphi\|_{1} \leq C\rho^{-2\gamma}$ since $\eta < \gamma$. To estimate $\text{I}_{1}^{\sharp}$, we shall also use that:
\begin{equation}
\| u^{<}_{z} \|_{2} \leq C\rho^{\frac{3\eta}{2}}\;,\qquad \| u^{0}_{z} \|_{2} \leq C\rho^{\frac{3\eta}{2\delta}}\;,\qquad \| u^{>}_{z} \|_{2} \leq C_{\beta}\;.
\end{equation}
The first two estimates are better than $\| u^{\text{r}}_{z} \|_{2} \leq C_{\beta}$. Moreover, we shall use that:
\begin{eqnarray}\label{eq:various}
&&\int dz\, \| a(u^{<}_{z}) \xi_{\lambda} \|^{2} \leq C\langle \xi_{\lambda}, \mathcal{N} \xi_{\lambda} \rangle\;,\qquad \int dz\, \| a(u^{0}_{z}) \xi_{\lambda} \|^{2} \leq C\rho^{-2\eta}\langle \xi_{\lambda}, \mathbb{H}_{0} \xi_{\lambda} \rangle\;,\nonumber\\
&&\qquad\qquad\qquad\qquad \int dz\, \| a(u^{>}_{z}) \xi_{\lambda} \|^{2} \leq C\rho^{-\frac{2\eta}{\delta}}\langle \xi_{\lambda}, \mathbb{H}_{0} \xi_{\lambda} \rangle\;.
\end{eqnarray}
\medskip

\noindent{\it\underline{Estimate for $\text{I}_{1}^{<}$.}} Here we replace $\varphi$ with $\varphi_{<}$. Also, using the fact that $\hat v^{\text{r}}(k)$ is supported for momenta $|k| \leq C\rho^{\frac{1}{3}}$, repeating an argument similar to the one of (\ref{eq:consphi0})-(\ref{eq:consphi2}) but this time for the $z'$ integration, we can freely replace $u^{\text{r}}_{z'}$ with $\tilde u^{<}_{z'}$; in the following, we shall denote by $\tilde u^{\sharp}_{z'}$ function whose Fourier transform has similar support properties as $\hat \varphi_{\sharp}$, possibly replacing the factors $2$ and $1/2$ in Eq. (\ref{eq:consphi2}) by $4$ and $1/4$. We then have:
\begin{eqnarray}\label{eq:Q4I}
&&|\text{I}_{1}^{<}| \leq \sum_{\sigma \neq \sigma'} \int dxdydzdz'\, V(x-y) | \varphi_{<}(z-z')| |\tilde \omega^{\text{r}}(z;y)| \| \overline{v}^{\text{r}}_{x} \|_{2} \| \overline{v}^{\text{r}}_{z'} \|_{2} \| u^{<}_{z} \|_{2} \| a_{\sigma}(u_{x}) a_{\sigma'}(u_{y}) \xi_{\lambda} \| \| a_{\downarrow}(\tilde u^{<}_{z'}) \xi_{\lambda} \| \nonumber\\
&&\quad \leq C\rho^{1 + \frac{3\eta}{2}} \sum_{\sigma \neq \sigma'} \int dxdydzdz'\, V(x-y) | \varphi_{<}(z-z')| |\tilde \omega^{\text{r}}(z;y)| \| a_{\sigma}(u_{x}) a_{\sigma'}(u_{y}) \xi_{\lambda} \| \| a_{\downarrow}(\tilde u^{<}_{z'}) \xi_{\lambda} \| \nonumber\\
&&\quad \leq C|\log \rho|\rho^{1 + \frac{3\eta}{2} - 2\gamma - \frac{\epsilon}{3}} \| \widetilde{\mathbb{Q}}^{\frac{1}{2}}_{1} \xi_{\lambda} \| \| \mathcal{N}^{\frac{1}{2}} \xi_{\lambda} \|\;,
\end{eqnarray}
the last step following from Cauchy-Schwarz inequality, combined with $\| \varphi_{<} \|_{1} \leq C\rho^{-2\gamma}$, $\| \omega^{\text{r}} \|_{1} \leq C\rho^{-\frac{\epsilon}{3}}$. 
\medskip

\noindent{\it\underline{Estimate for $\text{I}_{1}^{0}$.}} Here we replace $\varphi$ with $\varphi_{0}$ and $u^{\text{r}}_{z'}$ with $\tilde u^{0}_{z'}$. We have:
\begin{eqnarray}\label{eq:Q4I02}
&&|\text{I}_{1}^{0}| \leq \sum_{\sigma \neq \sigma'} \int dxdydzdz'\, V(x-y) | \varphi_{0}(z-z')| |\tilde \omega^{\text{r}}(z;y)| \| \overline{v}^{\text{r}}_{x} \|_{2} \| \overline{v}^{\text{r}}_{z'} \|_{2} \| u^{0}_{z} \|_{2} \| a_{\sigma}(u_{x}) a_{\sigma'}(u_{y}) \xi_{\lambda} \| \| a_{\downarrow}(\tilde u^{0}_{z'}) \xi_{\lambda} \| \nonumber\\
&&\quad \leq C\rho^{1 + \frac{3\eta}{2\delta}} \sum_{\sigma \neq \sigma'} \int dxdydzdz'\, V(x-y) | \varphi_{0}(z-z')| |\tilde \omega^{\text{r}}(z;y)| \| a_{\sigma}(u_{x}) a_{\sigma'}(u_{y}) \xi_{\lambda} \| \| a_{\downarrow}(\tilde u^{0}_{z'}) \xi_{\lambda} \| \nonumber\\
&&\quad \leq C|\log \rho|\rho^{1 + \frac{3\eta}{2\delta} - 3\eta - \frac{\epsilon}{3}} \| \widetilde{\mathbb{Q}}^{\frac{1}{2}} \xi_{\lambda} \| \| \mathbb{H}_{0}^{\frac{1}{2}}\xi_{\lambda} \|\;,
\end{eqnarray}
by Cauchy-Schwarz inequality, this time using that $\|\varphi_{0}\| \leq C|\log \rho|\rho^{-2\eta}$ and the second of (\ref{eq:various}).
\medskip

\noindent{\it\underline{Estimate for $\text{I}^{>}_{1}$.}} Here we replace $\varphi$ with $\varphi_{>}$ and $u^{\text{r}}_{z'}$ with $\tilde u^{>}_{z'}$. We then have:
\begin{eqnarray}\label{eq:Q4Imag}
&&|\text{I}_{1}^{>}| \leq \sum_{\sigma \neq \sigma'} \int dxdydzdz'\, V(x-y) | \varphi_{>}(z-z')| |\tilde \omega^{\text{r}}(z;y)| \| \overline{v}^{\text{r}}_{x} \|_{2} \| \overline{v}^{\text{r}}_{z'} \|_{2} \| u^{>}_{z} \|_{2} \| a_{\sigma}(u_{x}) a_{\sigma'}(u_{y}) \xi_{\lambda} \| \| a_{\downarrow}(\tilde u^{>}_{z'}) \xi_{\lambda} \| \nonumber\\
&&\quad \leq C_{\beta}\rho \sum_{\sigma \neq \sigma'} \int dxdydzdz'\, V(x-y) | \varphi_{>}(z-z')| |\tilde \omega^{\text{r}}(z;y)| \| a_{\sigma}(u_{x}) a_{\sigma'}(u_{y}) \xi_{\lambda} \| \| a_{\downarrow}(\tilde u^{>}_{z'}) \xi_{\lambda} \| \nonumber\\
&&\quad \leq C_{\beta} \rho^{1 - \frac{3\eta}{\delta} - \frac{\epsilon}{3}} \| \widetilde{\mathbb{Q}}^{\frac{1}{2}} \xi_{\lambda} \| \| \mathbb{H}_{0}^{\frac{1}{2}}\xi_{\lambda} \|\;,
\end{eqnarray}
again by CS inequality, using that $\|\varphi_{>}\| \leq C|\log \rho|\rho^{-2\eta/\delta}$ and the last of (\ref{eq:various}).
\medskip

\noindent{\it\underline{Putting it together: estimate for $\text{I}_{1}$.}} From (\ref{eq:Q4I}), (\ref{eq:Q4I02}), (\ref{eq:Q4Imag}), we get:
\begin{eqnarray}
| \text{I}_{1} | &\leq& C|\log\rho|\rho^{1 + \frac{3\eta}{2} - 2\gamma - \frac{\epsilon}{3}} \| \widetilde{\mathbb{Q}}^{\frac{1}{2}}_{1} \xi_{\lambda} \| \| \mathcal{N}^{\frac{1}{2}} \xi_{\lambda} \|\nonumber\\&& + C|\log \rho|\rho^{1 + \frac{3\eta}{2\delta} - 3\eta - \frac{\epsilon}{3}} \| \widetilde{\mathbb{Q}}^{\frac{1}{2}} \xi_{\lambda} \| \| \mathbb{H}_{0}^{\frac{1}{2}}\xi_{\lambda} \| + C_{\beta} \rho^{1 - \frac{3\eta}{\delta} - \frac{\epsilon}{3}} \| \widetilde{\mathbb{Q}}^{\frac{1}{2}} \xi_{\lambda} \| \| \mathbb{H}_{0}^{\frac{1}{2}}\xi_{\lambda} \|\;.
\end{eqnarray}
The optimal value of $\delta$ is $\delta = 3/2$. For this value, using also the propagation of the a priori estimates for $\widetilde{\mathbb{Q}}_{1}$, $\mathbb{H}_{0}$, Eqs. (\ref{eq:derA}), we get, for $0<\alpha < 1$:
\begin{eqnarray}\label{eq:I1est00}
|\text{I}_{1}| &\leq& C|\log\rho|L^{\frac{3}{2}}\rho^{2 + \frac{3\eta}{2} - 2\gamma - \frac{\epsilon}{3}} \| \mathcal{N}^{\frac{1}{2}} \xi_{\lambda} \| + C_{\beta} L^{3} \rho^{3 - 2\eta - \frac{\epsilon}{3}} \\
&\leq& C|\log\rho|L^{3} \rho^{3 + \frac{3\eta}{2} - \frac{5}{2}\gamma - \frac{\epsilon}{3}} + C_{\beta} L^{3} \rho^{3 - 2\eta - \frac{\epsilon}{3}} + C_{\alpha}L^{3} |\log\rho|^{\frac{4}{3}} \rho^{\frac{25}{9} + 2\eta - \frac{8}{3}\gamma - \frac{4\epsilon}{9}} + \alpha \langle \xi_{1}, \mathbb{H}_{0} \xi_{1}\rangle\;.\nonumber
\end{eqnarray}
The last estimate follows after using the bound (\ref{eq:Nimpro}) for the number operator, and using Young's inequality. For $\eta \leq 1/3$, the third term is bigger than the first. Hence:
\begin{eqnarray}\label{eq:estI1}
|\text{I}_{1}| &\leq& C_{\beta} L^{3} \rho^{3 - 2\eta - \frac{\epsilon}{3}} + C_{\alpha}L^{3} |\log \rho|^{\frac{4}{3}} \rho^{\frac{25}{9} + 2\eta - \frac{8}{3}\gamma - \frac{4\epsilon}{9}} + \alpha \langle \xi_{1}, \mathbb{H}_{0} \xi_{1}\rangle \nonumber\\
&\leq& C_{\beta,\alpha} L^{3}\rho^{\frac{26}{9} - \frac{4}{3}\gamma - \frac{7}{18}\epsilon} + \alpha \langle \xi_{1}, \mathbb{H}_{0} \xi_{1}\rangle\;,
\end{eqnarray}
where we optimized over $\eta$, $\eta = \frac{1}{18} + \frac{2\gamma}{3} + \frac{\epsilon}{36}$, which is less $\gamma$ for $\frac{1}{6} + \frac{\epsilon}{12} < \gamma$. All the other contributions arising from the first line of (\ref{eq:bigcom}) are estimated in this way. 
\medskip

\noindent{\it\underline{Estimate for $\text{I}_{2}$.}} Consider now the term $\text{I}_{2}$ in Eq. (\ref{eq:splitI}). As discussed in Appendix \ref{app:6uv}, this term satisfies the same estimate as $\text{I}_{1}$:
\begin{eqnarray}\label{eq:estI2}
|\text{I}_{2}| \leq C_{\beta,\alpha} L^{3}\rho^{\frac{26}{9} - \frac{4}{3}\gamma - \frac{7}{18}\epsilon} + \alpha \langle \xi_{1}, \mathbb{H}_{0} \xi_{1}\rangle\;.
\end{eqnarray}
\noindent{\it\underline{Other contributions to (\ref{eq:bigcom}).}} Consider now the second line of (\ref{eq:bigcom}). We rewrite it as:
\begin{eqnarray}\label{eq:bigcomm2}
&&-a_\uparrow(\overline{v}^{\text{r}}_z)a_\downarrow(\overline{v}^{\text{r}}_{z'})[a_{\sigma}^{*}(u^{\text{r}}_{x})a_{\sigma'}^{*}(u^{\text{r}}_{y}), a_{\uparrow}(u^{\text{r}}_z)a_{\downarrow}(u^{\text{r}}_{z'})]a^{*}_{\sigma'}(\overline{v}_y)a^{*}_\sigma(\overline{v}_x)  \nonumber\\
&&\qquad = -a^{*}_{\sigma'}(\overline{v}_y)a^{*}_\sigma(\overline{v}_x)[a^{*}_\sigma(u^{\text{r}}_x)a^{*}_{\sigma'}(u^{\text{r}}_y), a_\uparrow(u^{\text{r}}_z)a_\downarrow(u^{\text{r}}_{z'})]a_{\uparrow}(\overline{v}^{\text{r}}_z)a_\downarrow(\overline{v}^{\text{r}}_{z'})\nonumber\\
&&\quad\qquad - [a^{*}_\sigma(u_x^{\text{r}})a^{*}_{\sigma'}(u^{\text{r}}_y),a_\uparrow(u^{\text{r}}_z)a_\downarrow(u^{\text{r}}_{z'}) ][a_\uparrow(\overline{v}^{\text{r}}_z)a_\downarrow(\overline{v}^{\text{r}}_{z'}), a^{*}_{\sigma'}(\overline{v}_y^{\text{r}})a_{\sigma}(\overline{v}_{x}^{\text{r}})]\;.
\end{eqnarray}
Let us consider the first term in the right-hand side. We rewrite the commutator as, omitting the spin for simplicity:
\begin{eqnarray}\label{eq:onemore2}
[a^{*}(u^{\text{r}}_x)a^{*}(u^{\text{r}}_y), a(u^{\text{r}}_z)a(u^{\text{r}}_{z'})] &=& a^{*}(u^{\text{r}}_{y}) u^{\text{r}}(x;z) a(u^{\text{r}}_{z'}) - a^{*}(u^{\text{r}}_{y}) u^{\text{r}}(x;z') a(u^{\text{r}}_{z}) \nonumber\\
&& + a(u^{\text{r}}_{z'}) u^{\text{r}}(y;z) a^{*}(u^{\text{r}}_{x}) - a(u^{\text{r}}_{z}) u^{\text{r}}(y;z') a^{*}(u^{\text{r}}_{x})\;.
\end{eqnarray}
The first two terms give rise to error terms of the form:
\begin{equation}
\text{II}_{a} = \int dxdy dzdz'\, V(x-y) \varphi(z-z') u^{\text{r}}(y;z) \langle \xi_{\lambda}, a^{*}(u^{\text{r}}_{x}) a^{*}(\overline{v}_y)a^{*}(\overline{v}_x) a(\overline{v}^{\text{r}}_z)a(\overline{v}^{\text{r}}_{z'}) a(u^{\text{r}}_{z'})\xi_{\lambda} \rangle\;.
\end{equation}
We then get, using that $\| u^{\text{r}}_{y} \|_{1} \leq C$:
\begin{eqnarray}\label{eq:estII}
|\text{II}_{a}| &\leq& C \rho^{2} \int dxdydzdz'\, V(x-y) \varphi(z-z') |u^{\text{r}}(y;z) | (\| a(u^{\text{r}}_{x}) \xi_{\lambda} \|^{2} + \| a(u^{\text{r}}_{z'}) \xi_{\lambda} \|^{2} ) \nonumber\\
&\leq& C\rho^{2 - 2\gamma} \langle \xi_{\lambda}, \mathcal{N} \xi_{\lambda} \rangle\;.
\end{eqnarray}
To bound the error terms produced by the last two terms in (\ref{eq:onemore2}), we normal order them. The normal ordered contribution can be estimated as $\text{II}_{a}$. The new contraction produces an error term of the form:
\begin{equation}
\text{II}_{b} =\int dxdy dzdz'\, V(x-y) \varphi(z-z') u^{\text{r}}(y;z) u^{\text{r}}(x;z') \langle \xi_{\lambda}, a^{*}(\overline{v}_y)a^{*}(\overline{v}_x) a(\overline{v}^{\text{r}}_z)a(\overline{v}^{\text{r}}_{z'}) \xi_{\lambda} \rangle\;.
\end{equation}
We have, using that $|\varphi(z-z')| \leq C$:
\begin{eqnarray}\label{eq:estIIb}
|\text{II}_{b}| &\leq& C\rho \int dxdy dzdz\, V(x-y) |u^{\text{r}}(y;z)| |u^{\text{r}}(x;z')| (\| a(\overline{v}_y) \xi_{\lambda} \|^{2} + \| a(\overline{v}^{\text{r}}_{z'}) \xi_{\lambda} \|^{2}) \nonumber\\
&\leq& C\rho \langle \xi_{\lambda}, \mathcal{N} \xi_{\lambda} \rangle\;.
\end{eqnarray}
This concludes the analysis of the error terms produced by the first term in the right-hand side of (\ref{eq:bigcomm2}). We are left with the second term in (\ref{eq:bigcomm2}), involving two commutators. We compute:
\begin{eqnarray}\label{eq:commu}
	[a^{*}_\sigma(u_x^{\text{r}})a^{*}_{\sigma'}(u^{\text{r}}_y),a_\uparrow(u^{\text{r}}_z)a_\downarrow(u^{\text{r}}_{z'}) ] &=& \delta_{\sigma',\uparrow} u^{\text{r}}_{\sigma'}(z;y)a^{*}_\sigma(u_x^{\text{r}})a_{\downarrow}(u^{\text{r}}_{z'})-\delta_{\sigma',\downarrow} u^{\text{r}}_{\sigma'}(z';y)a^{*}_\sigma(u_x^{\text{r}})a_\uparrow(u^{\text{r}}_z)\\
	&& -\delta_{\sigma, \uparrow} u^{\text{r}}_\sigma(z;x)a^{*}_{\sigma'}(u^{\text{r}}_y)a_\downarrow(u^{\text{r}}_{z'})-\delta_{\sigma,\downarrow}u^{\text{r}}_\sigma(z';x)a^{*}_{\sigma'}(u^{\text{r}}_y)a_\uparrow(u^{\text{r}}_z)\;,\nonumber\\
	&& +\delta_{\sigma,\uparrow}\delta_{\sigma',\downarrow}u^{\text{r}}_\sigma(z;x) u^{\text{r}}_{\sigma'}(z';y) - \delta_{\sigma,\downarrow}\delta_{\sigma', \uparrow}u^{\text{r}}_{\sigma}(z';x) u^{\text{r}}_{\sigma'}(z;y)\nonumber
\end{eqnarray}
and
\begin{eqnarray}\label{eq:commv}
	[a_\uparrow(\overline{v}^{\text{r}}_z)a_{\downarrow}(\overline{v}^{\text{r}}_{z'}), a^{*}_{\sigma'}(\overline{v}^{\text{r}}_y)a^{*}_{\sigma}(\overline{v}_x)]&=& \delta_{\sigma', \uparrow}\omega^{\text{r}}_{\sigma'}(z;y)a^{*}_\sigma(\overline{v}_x)a_\downarrow(\overline{v}^{\text{r}}_{z'}) -\delta_{\sigma',\downarrow}\omega^{\text{r}}_{\sigma'}(z';y)a^{*}_\sigma(\overline{v}^{\text{r}}_x)a_\uparrow(\overline{v}^{\text{r}}_z)\\
	&& +\delta_{\sigma,\uparrow}\omega^{\text{r}}_\sigma(z;x)a^{*}_{\sigma'}(\overline{v}_y^{\text{r}})a_\downarrow(\overline{v}^{\text{r}}_{z'})-\delta_{\sigma,\downarrow}\omega^{\text{r}}_\sigma(z';x)a^{*}_{\sigma'}(\overline{v}^{\text{r}}_y)a_{\downarrow}(\overline{v}^{\text{r}}_{z'})\nonumber\\
	&& +\delta_{\sigma,\uparrow}\delta_{\sigma',\downarrow}\omega^{\text{r}}_\sigma(z;x)\omega_{\sigma'}^{\text{r}}(z';y) -\delta_{\sigma,\downarrow}\delta_{\sigma', \uparrow}\omega^{\text{r}}_{\sigma}(z';x)\omega^{\text{r}}_{\sigma'}(z;y)\;.\nonumber
\end{eqnarray}
The last two terms in (\ref{eq:commu}) times the last two terms in (\ref{eq:commv}) produce the explicit $O(\rho^{2})$ term in the final claim. Summing also the complex conjugate, and using that $\varphi(x) = \varphi(-x)$, we have:
\begin{eqnarray}\label{eq:prima}
\text{I}_{\text{main}} &=& \sum_{\sigma\neq \sigma'}\int dxdydzdz'\, V(x-y) \varphi(z-z')  (\delta_{\sigma,\uparrow}\delta_{\sigma',\downarrow}u^{\text{r}}_\sigma(z;x) u^{\text{r}}_{\sigma'}(z';y) - \delta_{\sigma,\downarrow}\delta_{\sigma', \uparrow}u^{\text{r}}_{\sigma}(z';x) u^{\text{r}}_{\sigma'}(z;y))\nonumber\\&&\qquad \qquad \cdot  (\delta_{\sigma,\uparrow}\delta_{\sigma',\downarrow}\omega^{\text{r}}_\sigma(z;x)\omega_{\sigma'}^{\text{r}}(z';y) -\delta_{\sigma,\downarrow}\delta_{\sigma', \uparrow}\omega^{\text{r}}_{\sigma}(z';x)\omega^{\text{r}}_{\sigma'}(z;y)) \nonumber\\
&=& 2\int dxdydzdz'\, V(x-y) \varphi(z-z') u^{\text{r}}_{\uparrow} (z;x) u^{\text{r}}_{\downarrow}(z';y) \omega^{\text{r}}_{\uparrow}(z;x)\omega_{\downarrow}^{\text{r}}(z';y)\;,
\end{eqnarray}
where in the last step we used that $\varphi(x) = \varphi(-x)$. As proven in Appendix \ref{sec:UV4}, thanks to the regularity of the potential, the function $u^{\text{r}}_{\sigma}(z;x)$ can be replaced by the Dirac delta $\delta(z-x)$, up to higher order terms in the density:
\begin{equation}\label{eq:Imainreg}
\text{I}_{\text{main}} = 2\rho^{\text{r}}_{\uparrow} \rho^{\text{r}}_{\downarrow} \int dxdy\, V(x-y) \varphi(x-y) + \mathcal{E}_{\text{main}}\;,\qquad |\mathcal{E}_{\text{main}}| \leq CL^{3}\rho^{3 - 2\gamma}\;.
\end{equation}
All the other terms arising in the product of (\ref{eq:commu}) and of (\ref{eq:commv}) give rise to subleading contributions. For instance, a typical term is:
\begin{eqnarray}
\text{III}_{a} &=& \int  dxdydzdz'\, V(x-y) \varphi(z-z') u^{\text{r}}(z;y)\omega^{\text{r}}(z;y)\langle \xi_{\lambda}, a^{*}(u_x^{\text{r}})a(u^{\text{r}}_{z'})a^{*}(\overline{v}^{\text{r}}_x)a(\overline{v}^{\text{r}}_{z'}) \xi_{\lambda}\rangle\nonumber\\
&\equiv& \int  dxdydzdz'\, V(x-y) \varphi(z-z') u^{\text{r}}(z;y)\omega^{\text{r}}(z;y)\langle \xi_{\lambda}, a^{*}(u_x^{\text{r}})a^{*}(\overline{v}^{\text{r}}_x)a(\overline{v}^{\text{r}}_{z'}) a(u^{\text{r}}_{z'})\xi_{\lambda}\rangle\;,
\end{eqnarray}
where we used the orthogonality between $u^{\text{r}}$ and $\overline{v}^{\text{r}}$. Proceeding as for $\text{II}_{a}$, using that $\| u^{\text{r}}_{y} \|_{1} \leq C$ and that $\| \omega^{\text{r}}_{y} \|_{\infty} \leq C\rho$, we get:
\begin{equation}\label{eq:estIIIa}
|\text{III}_{a}| \leq C\rho^{2 - 2\gamma} \langle \xi_{\lambda}, \mathcal{N} \xi_{\lambda} \rangle\;.
\end{equation}
Another typical term is:
\begin{equation}
\text{III}_{b}  = \int  dxdydzdz'\, V(x-y)\varphi(z-z')\omega^{\text{r}}(z;x)\omega^{\text{r}}(z';y)u^{\text{r}}(z;y) \langle \xi_{\lambda}, a^{*}(u^{\text{r}}_x)a(u^{\text{r}}_{z'}) \xi_{\lambda}\rangle\;,
\end{equation}
which we bound as:
\begin{equation}
|\text{III}_{b}| \leq C \rho^{2 - 2\gamma} \langle \xi_{\lambda}, \mathcal{N} \xi_{\lambda} \rangle\;.
\end{equation}
The last type of error term arising in the product of (\ref{eq:commu}) and (\ref{eq:commv}) is:
\begin{equation}\label{eq:estIIIc}
\text{III}_{c} = \int  dxdydzdz'\, V(x-y)\varphi(z-z')u^{\text{r}}(z;x)u^{\text{r}}(z';y)\omega^{\text{r}}(z;y) \langle \xi_{\lambda}, a^{*}(\overline{v}_x^{\text{r}})a(\overline{v}^{\text{r}}_{z'}) \xi_{\lambda}\rangle\;,
\end{equation}
which we estimate as:
\begin{equation}
|\text{III}_{c}| \leq C \rho \langle \xi_{\lambda}, \mathcal{N} \xi_{\lambda}\rangle\;.
\end{equation}
\medskip

\noindent{\it\underline{Conclusion.}} Putting together (\ref{eq:estI1}), (\ref{eq:estI2}), (\ref{eq:estII}), (\ref{eq:estIIb}), (\ref{eq:Imainreg}), (\ref{eq:estIIIa})-(\ref{eq:estIIIc}) we have, using that $\rho^{2 - 2\gamma} \leq \rho$ for $\gamma \leq 1/2$:
\begin{eqnarray}
|\mathcal{E}_{\widetilde{\mathbb{Q}}_{4}}(\xi_{\lambda})| &\leq& C \rho \langle \xi_{\lambda}, \mathcal{N} \xi_{\lambda}\rangle + C_{\beta,\alpha} L^{3}\rho^{\frac{26}{9} - \frac{4}{3}\gamma - \frac{7}{18}\epsilon} + \alpha \langle \xi_{1}, \mathbb{H}_{0} \xi_{1}\rangle \nonumber\\
&\leq& CL^{3} \rho^{\frac{7}{3}} + C_{\beta,\alpha} L^{3}\rho^{\frac{26}{9} - \frac{4}{3}\gamma - \frac{7}{18}\epsilon} + \alpha \langle \xi_{1}, \mathbb{H}_{0} \xi_{1}\rangle\;,
\end{eqnarray}
where in the last step we used the bound (\ref{eq:Nimpro}) for the number operator. This concludes the proof.
\end{proof}
\noindent{\bf Conclusion: proof of the lower bound.} We are now ready to prove a lower bound for the ground state energy. We shall collect all the error terms, starting from Eq. (\ref{eq:begin}). We have, for $0<\alpha<1$:
\begin{eqnarray}
&&\langle \psi, \mathcal{H} \psi \rangle \geq E_{\text{HF}}(\omega) - \rho_{\uparrow}\rho_{\downarrow} \int dxdy\, V(x-y) \varphi(x-y) + \langle \xi_{1}, (\mathbb{H}_{0} + \widetilde{\mathbb{Q}}_{1}) \xi_{1} \rangle (1 - \alpha) \nonumber\\
&& - C_{\alpha} L^{3} \rho^{2 + \frac{1}{9}} - C_{\alpha} L^{3} \rho^{\frac{7}{3} - \frac{2\gamma}{3}} - C_{\alpha} L^{3} \rho^{\frac{13}{9} + 2\gamma} - C_{\alpha} L^{3} \rho^{2 + \frac{\epsilon}{3}} - CL^{3} \rho^{\frac{7}{3}} - C_{\beta,\alpha} L^{3}\rho^{\frac{26}{9} - \frac{4}{3}\gamma - \frac{7}{18}\epsilon}\;,
\end{eqnarray}
where we also used that $|\rho_{\sigma} - \rho_{\sigma}^{\text{r}}| \leq \rho^{1 + \frac{\epsilon}{3}}$. The integral in the right-hand side can be written as, up to a boundary term:
\begin{equation}\label{eq:Vinfty}
\int dxdy\, V(x-y) \varphi(x-y) = L^{3}\int_{\mathbb{R}^{3}} dx\, V_{\infty}(x) \varphi_{\infty}(x) + \frak{e}_{L}\;,
\end{equation}
with $\frak{e}_{L} = O(L^{2})$. Recall that
\begin{equation}\label{eq:aagamma}
8\pi a_{\gamma} = \int_{\mathbb{R}^{3}} dx\, V_{\infty}(x) (1 - \varphi_{\infty}(x))\;,\qquad |a - a_{\gamma}| \leq C\rho^{\gamma}\;.
\end{equation}
The optimal choice of parameters is:
\begin{equation}
\epsilon = \frac{1}{3}\;,\qquad \gamma = \frac{1}{3}\;,
\end{equation}
which fulfills the assumptions of Proposition \ref{prp:intQ4}. Taking $\frac{1}{2} \leq \alpha  < 1$, we finally have, for $L$ large enough:
\begin{equation}\label{eq:lowfin}
\frac{\langle \psi, \mathcal{H} \psi \rangle}{L^{3}} \geq \frac{3}{5}(6\pi^{2})^{\frac{2}{3}} (\rho_{\uparrow}^{\frac{5}{3}} + \rho_{\downarrow}^{\frac{5}{3}}) + 8\pi a \rho_{\uparrow} \rho_{\downarrow} - CL^{3}\rho^{2 + \frac{1}{9}} + \langle \xi_{1}, (\mathbb{H}_{0} + \widetilde{\mathbb{Q}}_{1}) \xi_{1} \rangle (1 - \alpha)\;.
\end{equation}
This concludes the proof of the lower bound.
\begin{remark}[Improved condensation estimate.]\label{rem:cond} The inequality (\ref{eq:lowfin}) can be used to prove an improved estimate for $\langle \xi_{1}, \mathbb{H}_{0} \xi_{1}\rangle$, for states $\psi$ that are energetically close enough to the ground state. Let $\psi$ be a fermionic state such that:
\begin{equation}\label{eq:distE}
\frac{\langle \psi, \mathcal{H} \psi \rangle}{L^{3}} - \frac{3}{5}(6\pi^{2})^{\frac{2}{3}} (\rho_{\uparrow}^{\frac{5}{3}} + \rho_{\downarrow}^{\frac{5}{3}}) + 8\pi a \rho_{\uparrow} \rho_{\downarrow} \leq C\rho^{2 + \frac{1}{9}}\;.
\end{equation}
As we will see in the next section, such states exists; in particular, the ground state satisfies the inequality (\ref{eq:distE}). Eqs. (\ref{eq:lowfin}), (\ref{eq:distE}) imply:
\begin{equation}
\langle \xi_{1}, \mathbb{H}_{0}\xi_{1}\rangle \leq CL^{2} \rho^{2 + \frac{1}{9}}\;;
\end{equation}
plugging this bound in (\ref{eq:Nimpro}), we get, for $\gamma = 1/3$:
\begin{equation}
\langle R^{*}\psi, \mathcal{N} R^{*}\psi \rangle \leq CL^{3} \rho^{\frac{11}{9}}\;.
\end{equation}
This inequality can be used to prove, see \cite{BPS}:
\begin{equation}
\tr\, \gamma^{(1)}_{\psi} (1 - \omega) \leq CL^{3} \rho^{\frac{11}{9}}\;.
\end{equation}
This bound improves on the condensation estimate (\ref{eq:condest}). The optimal condensation estimate is expected to be of order $\rho^{\frac{4}{3}}$: this is consistent with the fact that the next order correction to the ground state energy is of order $\rho^{\frac{7}{3}}$, \cite{HY}.
\end{remark}

\section{Upper bound on the ground state energy}\label{sec:upper}
In this section we shall conclude the proof of Theorem \ref{thm:main}, by proving an upper bound on the ground state energy that matches the lower bound we obtained in Section \ref{sec:lwbd}, up to $o(\rho^{2})$. This will be done taking the natural trial state $\psi = R T \Omega$, with $T$ the correlation structure defined in Section \ref{sec:T}, for a suitable value of the parameter $\gamma$, to be optimized. 

To begin, notice that $\psi$ is an $N$-particle state, with $N_{\uparrow}$ particles with spin $\uparrow$ and $N_{\downarrow}$ particles with spin $\downarrow$. In fact, we can rewrite $\psi$ as:
\begin{equation}
\psi = \widetilde{T} R\Omega\;,\qquad \widetilde{T} = e^{\widetilde B - \widetilde B^{*}}\;,\qquad \widetilde B := \int dzdz'\, \varphi(z-z') a_{\uparrow}(u^{\text{r}}_{z}) a^{*}_{\uparrow}(\overline{v} \overline{v}^{\text{r}}_{z}) a_{\downarrow}(u^{\text{r}}_{z'}) a^{*}_{\downarrow}(\overline{v} \overline{v}^{\text{r}}_{z'})\;,
\end{equation}
and $[\widetilde{B}, \mathcal{N}_{\sigma}] = 0$, with $\mathcal{N}_{\sigma} = \int dx\, a^{*}_{x,\sigma} a_{x,\sigma}$. Also, by the defining properties of fermionic Bogoliubov transformations (\ref{eq:defbogi}), we know that $R\Omega$ is an $N$-particle state, with $N_{\uparrow}$ particles with spin $\uparrow$ and $N_{\downarrow}$ particles with spin $\downarrow$. Therefore:
\begin{equation}
\mathcal{N}_{\sigma} \psi = \mathcal{N}_{\sigma} \widetilde{T} R\Omega = \widetilde{T} \mathcal{N}_{\sigma} R\Omega = N_{\sigma} \psi\;.
\end{equation}
Moreover, being $R$ and $T$ unitary operators, $\|\psi\| = \|\Omega\| = 1$.

To compute the energy of $\psi$, we shall rely on the estimates we have already proved for the lower bound. An important role in the upper bound is played by the following bound for the number operator, for $\gamma \leq 1/2$:
\begin{equation}\label{eq:Nup}
\langle \xi_{\lambda}, \mathcal{N} \xi_{\lambda}\rangle \leq CL^{3} \rho^{2-\gamma}\;,\qquad \xi_{\lambda} := T_{1-\lambda} \Omega\;.
\end{equation}
This bound follows from (\ref{eq:Nimpro}), using that now $\xi_{1} = \Omega$. Thanks to (\ref{eq:Nup}), it is not difficult to see that to prove the propagation of the estimates in Proposition \ref{prp:prop2} it is enough to assume $\gamma \leq 1/3$. The only point where we required a lower bound for $\gamma$ is the estimate (\ref{eq:TTQgro}), which now holds for all $\gamma$, as it is clear from the bound (\ref{eq:Nup}).

By Propositions \ref{prp:conj}, \ref{prp:bogbd}, we have:
\begin{eqnarray}
E_{L}(N_{\uparrow}, N_{\downarrow}) &\leq& E_{\text{HF}}(\omega) + \langle T\Omega, \mathbb{H}_{0} T\Omega \rangle + \langle T\Omega, \mathbb{X} T\Omega \rangle + \langle T\Omega, \mathbb{Q} T\Omega \rangle\nonumber\\
&\leq& E_{\text{HF}}(\omega) + \langle T\Omega, (\mathbb{H}_{0} + \mathbb{Q}_{1} + \mathbb{Q}_{4}) T\Omega \rangle + \mathcal{E}_{1}(\psi)\;,
\end{eqnarray}
with:
\begin{equation}\label{eq:E1up}
|\mathcal{E}_{1}(\psi)| \leq C\rho \langle T\Omega, \mathcal{N} T\Omega \rangle\ \leq CL^{3} \rho^{3 - \gamma}\;.
\end{equation}
Here we crucially used that the state $\xi_{0} = T\Omega$ is such that $\xi_{0}^{(n)} = 0$ unless $n = 4k$ for $k\in \mathbb{N}$, and hence that:
\begin{equation}
\langle T\Omega, \mathbb{Q}_{3} T\Omega\rangle = 0\;,
\end{equation}
recall Eq. (\ref{eq:Q3canc}). Consider now the $\langle T\Omega, \mathbb{Q}_{4} T\Omega\rangle$ term. We rewrite it as:
\begin{equation}
\langle T\Omega, \mathbb{Q}_{4} T\Omega\rangle = \langle T\Omega, \widetilde{\mathbb{Q}}_{4} T\Omega\rangle + \langle T\Omega, \widehat{\mathbb{Q}}_{4} T\Omega\rangle\;,
\end{equation}
with $\widehat{\mathbb{Q}}_{4}$ the contribution to $\mathbb{Q}_{4}$ with aligned spins, 
\begin{equation}
\widehat{\mathbb{Q}}_{4} = \frac{1}{2} \sum_{\sigma} \int dxdy\, V(x-y) a^{*}_{\sigma}(u_{x}) a^{*}_{\sigma}(u_{y}) a^{*}_{\sigma}(\overline{v}_{y}) a^{*}_{\sigma}(\overline{v}_{x})  + \text{h.c..}
\end{equation}
We claim that $\langle T\Omega, \widehat{\mathbb{Q}}_{4} T\Omega\rangle = 0$. To prove this, we shall use that $T\Omega$ and $\widehat{\mathbb{Q}}_{4} T\Omega$ belong to different spin sectors, and hence they are orthogonal vectors in the Fock space. Let $\mathcal{S}$ be the spin operator,
\begin{equation}
\mathcal{S} = \sum_{\sigma} \sigma \mathcal{N}_{\sigma}\;,\qquad \mathcal{N}_{\sigma} = \int dx\, a^{*}_{x,\sigma} a_{x,\sigma},
\end{equation}
where we identify $\uparrow\, \equiv +$ and $\downarrow\, \equiv -$. Clearly, $\mathcal{S} \Omega = 0$. Also, since $[ B, \mathcal{S} ] = 0$, we have $[T, \mathcal{S}] = 0$. Therefore, $\mathcal{S} T \Omega = 0$. At the same time,
\begin{eqnarray}
&&\mathcal{S} \int dxdy\, V(x-y) a^{*}_{\sigma}(u_{x}) a^{*}_{\sigma}(u_{y}) a^{*}_{\sigma}(\overline{v}_{y}) a^{*}_{\sigma}(\overline{v}_{x}) T\Omega \\
&&\qquad \qquad \qquad = \int dxdy\, V(x-y) a^{*}_{\sigma}(u_{x}) a^{*}_{\sigma}(u_{y}) a^{*}_{\sigma}(\overline{v}_{y}) a^{*}_{\sigma}(\overline{v}_{x}) (\mathcal{S} + \sigma 4) T\Omega \nonumber\\
&&\qquad \qquad \qquad = \sigma 4 \int dxdy\, V(x-y) a^{*}_{\sigma}(u_{x}) a^{*}_{\sigma}(u_{y}) a^{*}_{\sigma}(\overline{v}_{y}) a^{*}_{\sigma}(\overline{v}_{x}) T\Omega \neq 0\;.
\end{eqnarray}
Therefore, $\langle T\Omega, \widehat{\mathbb{Q}}_{4} T\Omega\rangle = 0$ by orthogonality between different spin sectors. We are then left with:
\begin{equation}\label{eq:Eup0}
E_{L}(N_{\uparrow}, N_{\downarrow}) \leq E_{\text{HF}}(\omega) + \langle T\Omega, (\mathbb{H}_{0} + \mathbb{Q}_{1} + \widetilde{\mathbb{Q}}_{4}) T\Omega \rangle + \mathcal{E}_{1}(\psi)\;,
\end{equation}
with $\mathcal{E}_{1}(\psi)$ bounded as in (\ref{eq:E1up}). Proceeding exactly as in Section \ref{sec:lwbd}, Eq. (\ref{eq:step1}), we get:
\begin{eqnarray}
\langle T\Omega, (\mathbb{H}_{0} + \mathbb{Q}_{1} + \widetilde{\mathbb{Q}}_{4}) T\Omega \rangle = \int_{0}^{1} d\lambda\, \langle \xi_{\lambda}, (\mathbb{T}_{1} + \mathbb{T}_{2}) \xi_{\lambda} \rangle + \langle \xi_{0}, \widetilde{\mathbb{Q}}_{4} \xi_{0} \rangle + \mathcal{E}_{2}(\psi)\;,
\end{eqnarray}
where $\mathcal{E}_{2}(\psi)$ can be bounded as, for $0\leq \gamma \leq 1/3$, thanks to the estimates (\ref{eq:bdsfin1}) and the bound  (\ref{eq:Nup}) for the number operator:
\begin{eqnarray}
|\mathcal{E}_{2}(\psi)| &\leq& \max_{\lambda \in [0;1]} | \mathcal{E}_{\mathbb{H}_{0}}(\xi_{\lambda}) | + \max_{\lambda \in [0;1]} | \mathcal{E}_{\mathbb{Q}_{1}}(\xi_{\lambda}) | \nonumber\\
&\leq& CL^{3} \rho^{\frac{7}{3} - \frac{\gamma}{2}} + C_{\beta}L^{3} \rho^{3 - \frac{11}{5}\gamma}\nonumber\\
&\leq& CL^{3} \rho^{\frac{7}{3} - \frac{\gamma}{2}}\;.
\end{eqnarray}
To get the second inequality we optimized over the parameter $\eta$ appearing in the bound for $ \mathcal{E}_{\mathbb{Q}_{1}}(\xi_{\lambda})$, $\eta = \frac{\gamma}{5}$, and to get the third we used that $\gamma \leq 1/3$. Then we write, proceeding as in Eq. (\ref{eq:intt}):
\begin{eqnarray}
\langle T\Omega, (\mathbb{H}_{0} + \mathbb{Q}_{1} + \widetilde{\mathbb{Q}}_{4}) T\Omega \rangle = - \int_{0}^{1} \langle \xi_{\lambda}, \widetilde{\mathbb{Q}}_{4}^{\text{r}} \xi_{\lambda} \rangle + \langle \xi_{0}, \widetilde{\mathbb{Q}}_{4} \xi_{0} \rangle + \mathcal{E}_{3}(\psi) + \mathcal{E}_{2}(\psi)\;,
\end{eqnarray}
with:
\begin{eqnarray}
|\mathcal{E}_{3}(\psi)| &\leq& \max_{\lambda \in [0;1]} |\langle \xi_{\lambda}, (\mathbb{T}_{1} + \mathbb{T}_{2} + \widetilde{\mathbb{Q}}_{4}^{\text{r}}) \xi_{\lambda} \rangle|\nonumber\\
&\leq& C L^{3} \rho^{2 + \gamma}\;,
\end{eqnarray}
where we used the estimate (\ref{eq:bdsfin2}), and the bound for the number operator (\ref{eq:Nup}). Next, proceeding as in Eq. (\ref{eq:derQ4}) we have:
\begin{eqnarray}
- \int_{0}^{1} d\lambda\, \langle \xi_{\lambda}, \widetilde{\mathbb{Q}}_{4}^{\text{r}} \xi_{\lambda} \rangle + \langle \xi_{0}, \widetilde{\mathbb{Q}}_{4} \xi_{0} \rangle = - \int_{0}^{1}d\lambda\, \frac{d}{d\lambda} \langle \xi_{\lambda} \widetilde{\mathbb{Q}}_{4} \xi_{\lambda}\rangle + \int_{0}^{1} d\lambda \int_{\lambda}^{1} d\lambda' \frac{d}{d\lambda'} \langle \xi_{\lambda'}, \widetilde{\mathbb{Q}}_{4}^{\text{r}} \xi_{\lambda'} \rangle\;.
\end{eqnarray}
We compute the derivatives using Proposition \ref{prp:intQ4}. We obtain:
\begin{eqnarray}\label{eq:derQQ4up}
\frac{d}{d\lambda}\langle \xi_{\lambda}, \widetilde{\mathbb{Q}}^{\text{r}}_{4} \xi_{\lambda} \rangle &=& 2\rho^{\text{r}}_{\uparrow}\rho^{\text{r}}_{\downarrow} \int dxdy\, V(x-y) \varphi(x-y) + \mathcal{E}_{\widetilde{\mathbb{Q}}^{\text{r}}_{4}}(\xi_{\lambda}) \nonumber\\
\frac{d}{d\lambda}\langle \xi_{\lambda}, \widetilde{\mathbb{Q}}_{4} \xi_{\lambda} \rangle &=& 2\rho^{\text{r}}_{\uparrow}\rho^{\text{r}}_{\downarrow} \int dxdy\, V(x-y) \varphi(x-y) + \mathcal{E}_{\widetilde{\mathbb{Q}}_{4}}(\xi_{\lambda})\;;
\end{eqnarray}
the bound for the error terms can be improved with respect to (\ref{eq:derQQ4}), making use of the estimate for the number operator (\ref{eq:Nup}). Inspection of the proof of Proposition \ref{prp:intQ4} shows that the estimate for the error terms $\mathcal{E}_{\widetilde{\mathbb{Q}}^{\text{r}}_{4}}(\xi_{\lambda})$, $\mathcal{E}_{\widetilde{\mathbb{Q}}_{4}}(\xi_{\lambda})$, are determined by the bound for the term $\text{I}_{1}$ in the first line of Eq. (\ref{eq:I1est00}):
\begin{eqnarray}
|\text{I}_{1}| &\leq& CL^{\frac{3}{2}}\rho^{2 + \frac{3\eta}{2} - 2\gamma - \frac{\epsilon}{3}} \| \mathcal{N}^{\frac{1}{2}} \xi_{\lambda} \| + C_{\beta} L^{3} \rho^{3 - 2\eta - \frac{\epsilon}{3}}\nonumber\\
&\leq& C_{\beta} L^{3} \rho^{3 - \frac{10}{7}\gamma - \frac{\epsilon}{3}}\;,
\end{eqnarray}
where in the last step we used the bound (\ref{eq:Nup}) and we optimized over $\eta$, $\eta = \frac{5}{7}\gamma$. Notice that, with respect to the original proof of Proposition \ref{prp:intQ4}, the optimal value of $\eta$ is now independent of $\epsilon$. We find:
\begin{equation}\label{eq:errQ4up0}
|\mathcal{E}_{\widetilde{\mathbb{Q}}^{\text{r}}_{4}}(\psi)| \leq C_{\beta}L^{3} \rho^{3 - \frac{10}{7}\gamma - \frac{\epsilon}{3}}\;,\quad |\mathcal{E}_{\widetilde{\mathbb{Q}}_{4}}(\psi)| \leq C_{\beta}L^{3} \rho^{3 - \frac{10}{7}\gamma - \frac{\epsilon}{3}}\;.
\end{equation}
With respect to the original proof of Proposition \ref{prp:intQ4}, the bound (\ref{eq:errQ4up0}) holds for all $\epsilon\geq 0$, as a consequence of the fact that the optimal value of $\eta$ does not depend on $\epsilon$.

\medskip

\noindent{\bf Conclusion: proof of the upper bound.} Putting together (\ref{eq:Eup0})-(\ref{eq:errQ4up0}), we find, for $0 \leq \gamma \leq 1/3$:
\begin{eqnarray}
E_{L}(N_{\uparrow}, N_{\downarrow}) &\leq& E_{\text{HF}}(\omega) - \rho_{\uparrow}\rho_{\downarrow} \int dxdy\, V(x-y) \varphi(x-y)\nonumber\\
&& + CL^{3} \rho^{2 + \frac{\epsilon}{3}} + C L^{3} \rho^{2 + \gamma} + C L^{3} \rho^{\frac{7}{3} - \frac{\gamma}{2}} + C_{\beta}L^{3} \rho^{3 - \frac{10}{7}\gamma - \frac{\epsilon}{3}}\;,
\end{eqnarray}
where we replaced $\rho_{\sigma}^{\text{r}}$ with $\rho_{\sigma}$, thus giving rise to an error term $O(\rho^{2 + \frac{\epsilon}{3}})$. Using Eqs. (\ref{eq:Vinfty}), (\ref{eq:aagamma}), we get, for $L$ large enough:
\begin{equation}
\frac{E_{L}(N_{\uparrow}, N_{\downarrow})}{L^{3}} \leq \frac{3}{5}(6\pi^{2})^{\frac{2}{3}} (\rho_{\uparrow}^{\frac{5}{3}} + \rho_{\downarrow}^{\frac{5}{3}}) + 8\pi a \rho_{\uparrow} \rho_{\downarrow} + C\rho^{2 + \frac{\epsilon}{3}} + C \rho^{2 + \gamma} + C\rho^{\frac{7}{3} - \frac{\gamma}{2}} + C_{\beta} \rho^{3 - \frac{10}{7}\gamma - \frac{\epsilon}{3}}\;.\nonumber
\end{equation}
The optimal value of $\epsilon$ is $\epsilon = \frac{3}{2} - \frac{15}{7}\gamma$ (recall that we are assuming $\gamma\leq 1/3$, so that $\epsilon \geq 0$). For $\gamma \leq 7/9$, and for this choice of $\epsilon$, $\rho^{2 + \frac{\epsilon}{3}}$ is smaller than $\rho^{\frac{7}{3} - \frac{\gamma}{2}}$. Therefore:
\begin{equation}
\frac{E_{L}(N_{\uparrow}, N_{\downarrow})}{L^{3}} \leq \frac{3}{5}(6\pi^{2})^{\frac{2}{3}} (\rho_{\uparrow}^{\frac{5}{3}} + \rho_{\downarrow}^{\frac{5}{3}}) + 8\pi a \rho_{\uparrow} \rho_{\downarrow} + C \rho^{2 + \gamma} + C  \rho^{\frac{7}{3} - \frac{\gamma}{2}}\;.
\end{equation}
Optimizing over $\gamma$, $\gamma = 2/9$, we finally get:
\begin{equation}
\frac{E_{L}(N_{\uparrow}, N_{\downarrow})}{L^{3}} \leq \frac{3}{5}(6\pi^{2})^{\frac{2}{3}} (\rho_{\uparrow}^{\frac{5}{3}} + \rho_{\downarrow}^{\frac{5}{3}}) + 8\pi a \rho_{\uparrow} \rho_{\downarrow} + C  \rho^{2 + \frac{2}{9}}\;.
\end{equation}
This concludes the proof of the upper bound, and of Theorem \ref{thm:main}. \qed
\medskip

\noindent{\bf Acknowledgements.} Marco Falconi, Emanuela L. Giacomelli and Marcello Porta acknowledge financial support from the Swiss National Science Foundation, for the project ``Mathematical Aspects of Many-Body Quantum Systems''. Marcello Porta acknowledges financial support from the European Research Council
(ERC) under the European Union's Horizon 2020 research and innovation programme (ERC StG MaMBoQ, grant agreement n.802901). We thank Benjamin Schlein for useful discussions.

\appendix
\section{Properties of the scattering equation}\label{sec:scat}
We start by recalling some useful properties of the solution of the scattering equation (\ref{eq:scat2}). We refer the reader to \cite{LSSY, ESY} for more details.
\begin{lemma}\label{lem:scat} Let $V$ be a non-negative, compactly supported and spherically symmetric function, such that $\text{supp}\, V \subset \{ x\in \mathbb{R}^{3} \mid |x| \leq R_{0} \}$, for some $R_{0}>0$. Let $a$ be the scattering length of $V$. Let $R>R_{0}$ and let $f_{R}$ be the ground state of the Neumann problem on the ball $B_{R}(0) = \{ x\in \mathbb{R}^{3} \mid |x| < R \}$:
\begin{equation}\label{eq:scateq}
(-\Delta + \frac{1}{2} V)f_{R} = E_{R} f_{R}\;,
\end{equation}
with boundary condition:
\begin{equation}
f_{R}(x) = 1\;,\qquad \nabla f_{R}(x) = 0\;,\qquad \text{for $x\in \partial B_{R}(0)$.}
\end{equation}
For $R$ sufficiently large, the following holds.
\begin{itemize}
\item[(i)] We have:
\begin{equation}\label{eq:ERdiff}
|E_{R} - 3 a R^{-3}|\leq \frac{C}{R^{4}}\;.
\end{equation}
\item[(ii)] We have, for all $x\in B_{R}(0)$, for any $n\in \mathbb{N}$, provided $V\in C^{k}$ with $k$ large enough:
\begin{equation}\label{eq:bdf}
0\leq f_{R}(x) \leq 1\;,\qquad 1 - f_{R}(x) \leq \frac{C}{|x|+1}\;,\qquad | \nabla^{n} f_{R}(x) | \leq C_{n}\;.
\end{equation}
\item[(iii)] Let:
\begin{equation}
a_{R} = \frac{1}{8\pi} \int dx\, V(x) f_{R}(x)\;.
\end{equation}
Then:
\begin{equation}
|a - a_{R}| \leq \frac{C}{R}\;.
\end{equation}
\end{itemize}
\end{lemma}
\begin{remark}
Concerning the last bound in (\ref{eq:bdf}), one can also prove that the derivatives decay in $|x|$. We will not need such improvement. For $n=1,2$, this bound is proven in {\it e.g.} \cite{ESY}. For higher values of $n$, the bound follows from the bounds for $n = 1,2$ and from the fact that $f_{R}$ solves the scattering equation. See also \cite{HSscat} for an explicit, nonperturbative expression of the scattering length $a$, in terms of the potential $V$.
\end{remark}
In the following we shall denote by $f = 1 - \varphi$ the extension to $\mathbb{R}^{3}$ of the Neumann solution of the scattering equation on the ball $B_{\rho^{-\gamma}}(0)$, with $\gamma > 0$ (that is, we will drop the $\infty$ symbol, that we used in the bulk of the paper, to avoid a clash of notations with the above lemma). Notice that the second bound in Eq. (\ref{eq:bdf}), together with the compact support in $B_{\rho^{-\gamma}}(0)$, immediately implies:
\begin{equation}
\| \varphi \|_{1} \leq C\rho^{-2\gamma}\;.
\end{equation}
In the proof of Proposition \ref{prp:intQ4},  an important role is played by cut-off versions of $\varphi$. We set:
\begin{eqnarray}\label{eq:phiminmagdef}
\varphi_{\sharp}(x) = \frac{1}{L^{3}} \sum_{p\in \frac{2\pi}{L} \mathbb{Z}^{3}} e^{ip\cdot x} \hat \varphi(p) \tilde \chi_{\sharp}(p)\;,
\end{eqnarray}
with $\tilde \chi_{\sharp}(p)$ as in Eqs. (\ref{eq:consphi2}). Notice that the functions $\varphi_{\sharp}$ are no longer compactly supported. We shall assume that $0\leq \eta < \gamma$, which is the interesting choice of parameters for our analysis.
\begin{lemma}\label{lem:bdL1} Let $V$ be as in Lemma \ref{lem:scat}. Then, for $L$ large enough:
\begin{equation}\label{eq:norms}
\| \varphi_{<} \|_{1} \leq C\rho^{-2\gamma} |\log \rho|\;,\qquad \|\varphi_{0}\|_{1} \leq C\rho^{-2\eta} |\log \rho|\;,\qquad \| \varphi_{>} \|_{1} \leq C\rho^{-\frac{2\eta}{\delta}} |\log \rho|\;.
\end{equation}
\end{lemma}
\begin{proof} In the following, we shall set $B \equiv B_{\rho^{-\gamma}}(0)$. We shall only prove the first and the last inequality, the proof of the second one being analogous to the one of the third.
\medskip

\noindent{\underline{Bound for $\varphi_{<}$}.} It is convenient to write:
\begin{equation}\label{eq:splitphimin}
\varphi_{<}(x) = \varphi_{\ll}(x) +  \tilde \varphi_{<}(x)\;,
\end{equation}
where:
\begin{equation}
\varphi_{\ll}(x) = \frac{1}{L^{3}}\sum_{p\in \frac{2\pi}{L} \mathbb{Z}^{3}} e^{ip\cdot x} \hat \varphi(p) \chi\Big( \frac{|p|}{\rho^{\gamma}}\Big)\tilde \chi_{<}(p)\;,\qquad  \tilde \varphi_{<}(x) = \frac{1}{L^{3}}\sum_{p\in \frac{2\pi}{L} \mathbb{Z}^{3}}  e^{ip\cdot x} \hat \varphi(p) \Big( 1 - \chi\Big(\frac{|p|}{\rho^{\gamma}}\Big)\Big) \tilde \chi_{<}(p)\;.
\end{equation}
Consider first $\varphi_{\ll}(x)$. We have, for $L$ large enough uniformly in $x$:
\begin{eqnarray}
| \varphi_{\ll}(x) | &\leq& C\int dp\, | \hat \varphi(p)| \chi(|p| / \rho^{\gamma}) \nonumber\\
&\leq& C\rho^{\gamma}\;.
\end{eqnarray}
In fact, since $|\varphi(x)| \leq C(1 + |x|)^{-1}$ and $\varphi(x)$ is compactly supported in $B$,
\begin{equation}
|\hat \varphi(p)| \leq \int dx\, |\varphi(x)| \leq C\rho^{-2\gamma}\;.
\end{equation}
Next, integrating by parts, for all $n\geq 1$, for $L$ large enough uniformly in $x$:
\begin{equation}
|x_{k}|^{n}_{L} |\varphi_{\ll}(x)| \leq C\int dp\, | \partial_{p_{k}}^{n} \hat \varphi(p)  \chi(|p| / \rho^{\gamma}) \tilde \chi_{<}(p)|\;.
\end{equation}
Every derivative brings a factor $\rho^{-\gamma}$. This is evident from the derivatives of the cutoff functions. Concerning $\hat \varphi(p)$:
\begin{eqnarray}
| \partial_{p_{k}}^{n} \hat \varphi(p) | &\leq& \int dx\, |x_{k}|^{n} | \varphi(x) | \nonumber\\
&\leq& C_{n} \rho^{- n\gamma - 2\gamma}\;,
\end{eqnarray}
where in the last step we bounded every $|x_{k}|$ factor by $\rho^{-\gamma}$, using the compact support of $\varphi$. Hence:
\begin{equation}
|x_{k}|_{L}^{n} |\varphi_{\ll}(x)| \leq C_{n} \rho^{- n\gamma + \gamma}\;,
\end{equation}
which gives:
\begin{equation}\label{eq:ll}
|\varphi_{\ll}(x)| \leq \frac{C_{n}\rho^{\gamma}}{ 1 + (\rho^{\gamma} |x|_{L})^{n}}\;.
\end{equation}
This bound implies that 
\begin{equation}\label{eq:phimmin}
\| \varphi_{\ll} \|_{1} \leq C\rho^{-2\gamma}\;.
\end{equation}
Let us now consider $\tilde \varphi_{<}$. We will prove decay estimates in configuration space using an integration by parts argument in momentum space. To efficiently estimate the derivatives of $\hat \varphi(p)$, it is convenient to consider the scattering equation in Fourier space.  We have, using that $\hat \varphi(p) = \int_{B} dx\, e^{ip\cdot x} \varphi(x)$, and recalling that $\varphi(x)$ solves (\ref{eq:scat2}) for $x\in B$:
\begin{equation}\label{eq:Fscat}
(|p|^{2} - \lambda_{\gamma}) \hat \varphi(p) + \frac{1}{2} (\hat V(p) - (\hat V * \hat \varphi)(p)) = -\lambda_{\gamma}\int_{B} dx\, e^{ip\cdot x}\;.
\end{equation}
To write Eq. (\ref{eq:Fscat}), we used that both $V(x)$ and $\varphi(x)$ have compact support in $B$, and that $\varphi = \nabla \varphi = 0$ on $\partial B$. We are interested in momenta $|p|$ such that $|p|^{2} \geq \rho^{2\gamma}$; this, together with the estimate $|\lambda_{\gamma}|\leq C\rho^{3\gamma}$, implies that $(|p|^{2} - \lambda_{\gamma})>0$. Therefore,
\begin{equation}\label{eq:phiF}
\partial_{p_{k}}^{n} \hat \varphi(p) = \partial_{p_{k}}^{n} \frac{1}{|p|^{2} - \lambda_{\gamma}} \Big( -  \frac{1}{2} \hat V(p) + \frac{1}{2}(\hat V * \hat \varphi)(p)) - \lambda_{\gamma} \int_{B} dx\, e^{ip\cdot x} \Big)\;.
\end{equation}
The derivatives of the first two terms in the brackets are bounded as, by the regularity of $\hat V$:
\begin{equation}
| \partial_{p_{k}}^{n} \hat V(p) | \leq C_{n}\;,\qquad |\partial_{p_{k}}^{n} (\hat V * \hat \varphi)(p)| \leq \int dx\, |x|^{n} |V(x)| |\varphi(x)| \leq C_{n}\;.\end{equation}
In the last inequality we used that $0\leq \varphi(x) \leq 1$, together with the compact support of $V(x)$. Also,
\begin{equation}
|p_{k}|^{n} |\hat V(p)| \leq C_{n}\;,\qquad |p_{k}|^{n} |(\hat V * \hat \varphi)(p)|\leq \int dx\, | \partial_{x_{k}}^{n} V(x) \varphi(x) | \leq C_{n}\;,
\end{equation}
where we used the fast decay of $\hat V(p)$, implies by the regularity of $V(x)$, and the fact that $\varphi(x)$ is regular in the support of $V$. Consider now the last term in the brackets. We compute:
\begin{eqnarray}\label{eq:phiF2}
\int_{B} dx\, e^{ip\cdot x} &=& 2\pi \int_{0}^{\rho^{-\gamma}} d t\, t^{2} \int_{-1}^{1} d\alpha\, e^{i t |p| \alpha}\nonumber\\
&=& 2\pi \int_{0}^{\rho^{-\gamma}} d t\, t^{2} \frac{2}{t |p|} \sin t |p|\nonumber\\
&=& \frac{4\pi}{|p|^{3}} \int_{0}^{|p| \rho^{-\gamma}} dt\, t \sin t\nonumber\\
&=& \frac{4\pi}{|p|^{3}} ( - |p| \rho^{-\gamma} \cos |p|\rho^{-\gamma} + \sin |p| \rho^{-\gamma} )\;.
\end{eqnarray}
Combined with Eq. (\ref{eq:phiF}), this computation implies, for $|p|\geq 1$:
\begin{equation}\label{eq:derr1}
|\partial_{p_{k}}^{n} \varphi(p)| \leq \frac{C\rho^{2\gamma - n\gamma}}{|p|^{4}} + \frac{C_{k+n}}{|p|^{k}}\;,\qquad \text{for all $k\in \mathbb{N}$.}
\end{equation}
Let us now consider the regime $\rho^{\gamma}\leq |p| \leq 1$. Eq. (\ref{eq:phiF2}) gives:
\begin{equation}
\Big| \partial_{p_{k}}^{n} \lambda_{\gamma} \int_{B} dx\, e^{ip\cdot x}  \Big| \leq C_{n} \frac{\rho^{2\gamma - n\gamma}}{|p|^{2}}\;;
\end{equation}
therefore, from Eq. (\ref{eq:phiF}) we get the bound, for $\rho^{\gamma} \leq |p| \leq 1$:
\begin{equation}\label{eq:esttildephi}
\Big|\partial_{p_{k}}^{n} \hat \varphi(p)\Big| \leq \frac{C_{n}}{|p|^{2}}\Big( \frac{1}{|p|^{n}} + \frac{\rho^{2\gamma - n\gamma}}{|p|^{2}} \Big)\;.
\end{equation}
For $|p| \geq \rho^{\gamma}$ and $n\geq 2$ the second term dominates. Let us now use the bound (\ref{eq:esttildephi}) to prove decay estimates for $ \tilde \varphi_{<}$. We have:
\begin{eqnarray}\label{eq:easy}
| \tilde \varphi_{<}(x) | &\leq& \int dp\, | \varphi(p) | ( 1 - \chi(|p| / \rho^{\gamma}) ) \chi(|p| / \rho^{\eta}) \nonumber\\
&\leq& C\rho^{\eta}\;.
\end{eqnarray}
Also, for $n\geq 1$:
\begin{eqnarray}\label{eq:tildephimin}
| x_{k} |_{L}^{n} |  \tilde \varphi_{<}(x) | &\leq& \int dp\, \Big| \partial^{n}_{p_{k}} \varphi(p) ( 1 - \chi(|p| / \rho^{\gamma}) ) \chi(|p| / \rho^{\eta}) \Big|\;.
\end{eqnarray}
Let $n\geq 2$. The bound (\ref{eq:esttildephi}) implies:
\begin{eqnarray}
| x_{k} |_{L}^{n} | \tilde  \varphi_{<}(x) | &\leq&  C_{n}\int_{\rho^{\gamma}}^{2\rho^{\gamma}} dp\, \frac{\rho^{-n\gamma}}{|p|^{2}} \nonumber\\
&& + C_{n} \int dp\, \frac{1}{|p|^{2}} \frac{\rho^{2\gamma - n\gamma}}{|p|^{2}} ( 1 - \chi(|p| / \rho^{\gamma}) ) \chi(|p| / \rho^{\eta})\;.
\end{eqnarray}
The first term bounds the terms where at least one derivative hits the characteristic functions, while the second arises from the estimate (\ref{eq:esttildephi}). Therefore, for $n\geq 2$:
\begin{equation}
| x_{k} |_{L}^{n} |  \tilde \varphi_{<}(x) | \leq C_{n} \rho^{-(n-1) \gamma} = C_{n} \rho^{-(n-1)\gamma - \eta} \rho^{\eta}\;.
\end{equation}
All together, recalling (\ref{eq:easy}), for $n\geq 2$:
\begin{equation}\label{eq:decnnn}
|  \tilde \varphi_{<}(x) | \leq \frac{C_{n} \rho^{\eta}}{ 1 + \Big( \rho^{\frac{(n-1)\gamma + \eta}{n}}|x|_{L} \Big)^{n}}\;.
\end{equation}
To estimate $\| \tilde \varphi_{<}\|_{1}$, we write:
\begin{equation}\label{eq:splittilde}
\|\tilde \varphi_{<}\|_{1} \leq \| \tilde \varphi_{<} \chi(|\cdot|_{L} \rho^{\gamma}) \|_{1} +  \| \tilde \varphi_{<} \chi^{c}(|\cdot|_{L} \rho^{\gamma})\|_{1}\;,
\end{equation}
and we shall study the two terms separately. Consider the first. Here we use (\ref{eq:decnnn}) with $n=3$. We get:
\begin{equation}\label{eq:splittilde1}
\int_{|x|_{L} \leq \rho^{-\gamma}} dx\, |  \tilde \varphi_{<}(x) | \leq C\rho^{\eta} \rho^{-3( 2\gamma/3 + \eta/3 )}  |\log \rho|= C\rho^{-2\gamma} | \log \rho |\;.
\end{equation}
Consider now the second term in (\ref{eq:splittilde}). Here we use (\ref{eq:decnnn}) with $n=4$. We get:
\begin{eqnarray}\label{eq:splittilde2}
\int_{|x|_{L} > \rho^{-\gamma}} dx\, |  \tilde \varphi_{<}(x) | &\leq& C\rho^{\eta} \rho^{-3 ( 3\gamma / 4 + \eta / 4 )} \frac{1}{ 1 + \rho^{-\gamma + 3\gamma / 4 + \eta / 4 } } \nonumber\\
&\leq& C\rho^{-2\gamma}\;.
\end{eqnarray}
Therefore, (\ref{eq:splittilde}), (\ref{eq:splittilde1}), (\ref{eq:splittilde2}) imply:
\begin{equation}
\| \tilde \varphi_{<} \|_{1} \leq C\rho^{-2\gamma} |\log \rho|\;.
\end{equation}
Combined with (\ref{eq:splitphimin}), (\ref{eq:phimmin}) we get:
\begin{equation}
\| \varphi_{<} \|_{1} \leq C\rho^{-2\gamma} |\log \rho|\;.
\end{equation}
This concludes the proof of the first of (\ref{eq:norms}). 

\medskip

\noindent{\it \underline{Bound for $\varphi_{>}$}.} Let us now prove the third estimate in (\ref{eq:norms}). To do this, it is convenient to write, for $n\in \mathbb{N}$ large enough:
\begin{equation}
\varphi_{>}(x) = \chi(|x|_{L} < \rho^{-n}) \varphi_{>}(x) + \chi(|x|_{L} \geq \rho^{-n}) \varphi_{>}(x)\;.
\end{equation}
For the second term, we use the nonoptimal bound $|\varphi_{>}(x)| \leq C_{m} \rho^{-3\beta - 2\gamma}/(1 + (\rho^{\gamma} |x|_{L})^{m})$, which can be proven as (\ref{eq:ll}), to show that:
\begin{equation}
\| \chi(|\cdot|_{L} \geq \rho^{-n}) \varphi_{>} \|_{1} \leq C\;,\qquad \text{for $n$ large enough.}
\end{equation}
Next, for the first term we approximate $\varphi_{>}(x)$ by its infinite volume counterpart $\varphi^{\infty}_{>}(x)$. We have:
\begin{equation}
\| \chi(|\cdot|_{L} < \rho^{-n}) \varphi_{>}\|_{1} \leq \| \chi(|\cdot|_{L} < \rho^{-n}) \varphi^{\infty}_{>}\|_{1} + \| \chi(|\cdot|_{L} < \rho^{-n})( \varphi^{\infty}_{>} - \varphi_{>})\|_{1}\;.
\end{equation}
Using that $| \varphi_{>}(x) - \varphi^{\infty}_{>}(x) | \leq C/L$ for fixed $x$, we have, for $L$ large enough:
\begin{equation}
\| \chi(|\cdot|_{L} < \rho^{-n})( \varphi^{\infty}_{>} - \varphi_{>})\|_{1} \leq C\;.
\end{equation}
Therefore, for $L$ large enough:
\begin{equation}\label{eq:basta}
\| \varphi_{>} \|_{1} \leq \| \varphi^{\infty}_{>} \|_{1} + C\;.
\end{equation}
Let us now focus on $\|\varphi_{>}^{\infty}\|_{1}$. We use that:
\begin{eqnarray}
\varphi^{\infty}_{>}(x) &=& \int dp\, e^{ip\cdot x} \hat \varphi(p) \chi(\rho^{\beta} |p|) (1 - \chi(|p| / \rho^{\eta/\delta})) \nonumber\\
&=& \frac{4\pi}{|x|}\int d t\, t \hat \varphi(t) \chi(\rho^{\beta} t) (1 - \chi(t / \rho^{\eta/\delta})) \sin (t |x|)
\end{eqnarray}
where: $t \equiv |p|$; in the last equality we used that, with a slight abuse of notation, $\varphi(p) \equiv \varphi(|p|)$; we performed the angular integration. Therefore,
\begin{equation}
|x|^{n} \varphi^{\infty}_{>}(x) = 4\pi \int d t\, t \hat \varphi(t) \chi(\rho^{\beta} t) (1 - \chi(t / \rho^{\eta/\delta})) |x|^{n-1}\sin (t |x|)\;.
\end{equation}
Using that $|x| \sin (t|x|) = -\partial_{t} \cos (t|x|)$, $|x| \cos(t|x|) = \partial_{t} \sin (t |x|)$, we get, integrating by parts:
\begin{equation}\label{eq:part}
|x|^{n} |\varphi^{\infty}_{>}(x)| \leq 4\pi \int d t\, \Big|\partial_{t}^{n-1} \Big(t \hat \varphi(t) \chi(\rho^{\beta} t) (1 - \chi(t / \rho^{\eta/\delta})) \Big)\Big|\;,
\end{equation}
where we used that all boundary terms vanish thanks to the characteristic functions. We are interested in estimating the right-hand side of (\ref{eq:part}) for $n=3$ and for $n=4$. We have various cases, depending on which function the derivatives hit.

Consider the terms where at least one derivative hits $\chi(\rho^{\beta} t)$. Then, using that $\partial^{k} \chi(\rho^{\beta} t) = \rho^{k\beta } \chi^{(k)}(\rho^{\beta} t)$, we see that $t$ is forced to be $O(\rho^{-\beta})$. Thanks to (\ref{eq:derr1}), it is not difficult to see that the resulting contribution to (\ref{eq:part}) is bounded uniformly in $\rho$.

Consider the case when all derivatives hit $(1 - \chi(t / \rho^{\eta/\delta}))$. Then, from $\partial_{t}^{n-1} (1 - \chi(t / \rho^{\eta/\delta})) = -\rho^{-(n-1)\frac{\eta}{\delta}} \chi^{(k)}(t / \rho^{\eta/\delta})$, using Eq. (\ref{eq:esttildephi}) we see that the resulting contribution is bounded as $\rho^{-(n-1)\frac{\eta}{\delta}}$. More generally, the same estimate holds true as long as the number of derivatives hitting $t\hat \varphi(t)$ is less or equal than $2$, and all the other derivatives hit $(1 - \chi(t / \rho^{\eta/\delta}))$.

The only case left to consider is when $n=4$, and all the $(n-1) = 3$ derivatives hit $t \hat \varphi(t)$. Thanks to (\ref{eq:esttildephi}), we see that this contribution, after integrating for $t\geq \rho^{-\eta/\delta}$, is bounded as $\rho^{-\frac{2\eta}{\delta}} \rho^{-\gamma}$. In conclusion, for $n=3,4$:
\begin{equation}
|x|^{n} |\varphi^{\infty}_{>}(x)| \leq C + C\rho^{-(n-1)\frac{\eta}{\delta}} + C\rho^{-\frac{2\eta}{\delta}} \rho^{-(n-3)\gamma} \leq C\rho^{-\frac{2\eta}{\delta}} \rho^{-(n-3)\gamma}\;,
\end{equation}
where in the last step we used that $n=3,4$ and that $\eta \leq \gamma$. Using also that $|\varphi^{\infty}_{>}(x)| \leq \int dp\, |\hat \varphi_{>}(p)| \leq C$, we get:
\begin{equation}\label{eq:bdphi>}
|\varphi^{\infty}_{>}(x)| \leq \frac{C}{1 + \big( \rho^{\frac{2\eta/\delta + (n-3)\gamma}{n}} |x| \Big)^{n}}\;,\qquad n = 3, 4\;.
\end{equation}
We are now ready to prove the second of (\ref{eq:norms}). We write:
\begin{equation}\label{eq:phifin0}
\| \varphi^{\infty}_{>} \|_{1} \leq \| \varphi^{\infty}_{>} \chi(|\cdot|_{L} \leq \rho^{-\gamma}) \|_{1} + \| \varphi^{\infty}_{>} \chi(|\cdot|_{L} > \rho^{-\gamma}) \|_{1}\;,
\end{equation}
and we estimate the two terms separately. For the first, we use (\ref{eq:bdphi>}) with $n=3$. We get:
\begin{equation}\label{eq:phifin1}
\| \varphi^{\infty}_{>} \chi(|\cdot|_{L} \leq \rho^{-\gamma}) \|_{1}  \leq C\rho^{-\frac{2\eta}{\delta}} |\log \rho|\;.
\end{equation}
For the second, we use (\ref{eq:bdphi>}) with $n=4$. We have:
\begin{eqnarray}\label{eq:phifin2}
\| \varphi^{\infty}_{>} \chi(|\cdot|_{L} > \rho^{-\gamma}) \|_{1} &\leq& C \rho^{-3 \big( \frac{2\eta/\delta + \gamma}{4} \big)} \frac{1}{1 + \rho^{-\gamma + \frac{2\eta/\delta + \gamma}{4}}}\nonumber\\
&\leq& C \rho^{-\frac{2\eta}{\delta}}\;.
\end{eqnarray}
Therefore, (\ref{eq:phifin1}), (\ref{eq:phifin2}) imply:
\begin{equation}
\| \varphi^{\infty}_{>} \|_{1} \leq C\rho^{-\frac{2\eta}{\delta}} |\log \rho|\;.
\end{equation}
Together with (\ref{eq:basta}), this proves the last of (\ref{eq:norms}). The proof of the second inequality in (\ref{eq:norms}) is completely analogous to the one we just discussed, we omit the details.
\end{proof}
\section{Proof of Lemma \ref{lem:phi}}\label{app:lemphi}
Let us start from the first bound. We proceed in exactly the same way as for the proof of Lemma \ref{lem:bosonbd}, with the only difference that $\overline{v}^{\text{r}}_{x}$ is replaced by $\partial^{n_{2}}\overline{v}^{\text{r}}_{x}$, which satisfies the bound $\|\partial^{n_{2}}\overline{v}^{\text{r}}_{x}\|_{2} \leq C\rho^{\frac{n_{2}}{3} + \frac{1}{2}}$. Therefore, we get:
\begin{eqnarray}
\Big| \int dxdy\, \varphi(x-y) \langle \xi_{\lambda}, a_{\uparrow}(u^{\text{r}}_{x}) a_{\uparrow}(\partial^{n_{2}}\overline{v}^{\text{r}}_{x}) a_{\downarrow}(u^{\text{r}}_{y}) a_{\downarrow}(\overline{v}^{\text{r}}_{y}) \xi_{\lambda} \rangle\Big| &\leq& C\rho^{1-2\gamma + \frac{n_{2}}{3}} \langle \xi_{\lambda}, \mathcal{N}\xi_{\lambda}\rangle + CL^{\frac{3}{2}} \rho^{1 - \frac{\gamma}{2} + \frac{n_{2}}{3}} \| \mathcal{N}^{\frac{1}{2}}\xi_{\lambda} \| \nonumber\\
&\leq& CL^{\frac{3}{2}} \rho^{1 - \frac{\gamma}{2} + \frac{n_{2}}{3}} \| \mathcal{N}^{\frac{1}{2}}\xi_{\lambda} \|\;.
\end{eqnarray}
The second inequality follows from $\| \mathcal{N}^{\frac{1}{2}}\xi_{\lambda} \| \leq CL^{\frac{3}{2}} \rho^{\frac{7}{12}}$ (propagation of the a priori estimate) and from $\gamma \leq 7/18$. Let us now prove the second bound. We shall proceed as for the first bound. The only difference is that instead of the estimate (\ref{eq:csineq}) we use:
\begin{eqnarray}
\rho^{\frac{1}{2} - \frac{\gamma}{2}} \int dx\, \| a_{\sigma}(\partial u^{\text{r}}_{y}) a_{\sigma}(\partial^{n_{3}}\overline{v}^{\text{r}}_{y}) \xi_{\lambda}\| &\leq& CL^{\frac{3}{2}} \rho^{1 - \frac{\gamma}{2} + \frac{n_{2}}{3}} \Big(\int dx\, \| a_{\sigma}(\partial u^{\text{r}}_{y}) \xi_{\lambda}\|^{2}\Big)^{\frac{1}{2}} \nonumber\\
&\leq& CL^{\frac{3}{2}} \rho^{1 - \frac{\gamma}{2} + \frac{n_{2}}{3}} ( \| \mathbb{H}_{0}^{\frac{1}{2}}\xi_{\lambda}\| + \rho^{\frac{1}{3}} \|\mathcal{N}^{\frac{1}{2}}\xi_{\lambda} \| )\;.
\end{eqnarray}
The last inequality follows from:
\begin{eqnarray}
\int dx\, \| a_{\sigma}(\partial u^{\text{r}}_{y}) \xi_{\lambda}\|^{2} &\leq& \sum_{k} |k|^{2} \| \hat a_{k,\sigma} \xi_{\lambda} \|^{2} \nonumber\\
&\leq& \langle \xi_{\lambda}, \mathbb{H}_{0} \xi_{\lambda}\rangle + C\rho^{\frac{2}{3}} \langle \xi_{\lambda}, \mathcal{N}\xi_{\lambda}\rangle\;.
\end{eqnarray}
This concludes the proof of Lemma \ref{lem:phi}. \qed
\section{Regularizations}\label{app:UV}
\subsection{Proof of Lemma \ref{lem:UV2}}\label{sec:UV2}
Let $R_{0}$ be such that $\text{supp}\, V_{\infty} \subset B_{R_{0}}(0)$. We rewrite:
\begin{eqnarray}\label{eq:Iapp}
&&\text{I} = \int dxdydzdz'\, V(x-y) \varphi(z-z') \delta^{\text{r}}_{\uparrow}(z;x) \delta^{\text{r}}_{\downarrow}(z';y) \langle \xi_{\lambda}, a_{\uparrow}(\overline{v}^{\text{r}}_{z}) a_{\uparrow}(u_{x}) a_{\downarrow}(\overline{v}^{\text{r}}_{z'}) a_{\downarrow}(u_{y}) \xi_{\lambda} \rangle \\
&&= \int dxdydzdz'\, V(x-y) \varphi(z-z') \chi(|z -z'|_{L} / (8R_{0}))\delta^{\text{r}}_{\uparrow}(z;x) \delta^{\text{r}}_{\downarrow}(z';y) \langle \xi_{\lambda}, a_{\uparrow}(\overline{v}^{\text{r}}_{z}) a_{\uparrow}(u_{x}) a_{\downarrow}(\overline{v}^{\text{r}}_{z'}) a_{\downarrow}(u_{y}) \xi_{\lambda} \rangle\nonumber\\
&& + \int dxdydzdz'\, V(x-y) \varphi(z-z') \chi^{c}(|z -z'|_{L} / (8R_{0}))\delta^{\text{r}}_{\uparrow}(z;x) \delta^{\text{r}}_{\downarrow}(z';y) \langle \xi_{\lambda}, a_{\uparrow}(\overline{v}^{\text{r}}_{z}) a_{\uparrow}(u_{x}) a_{\downarrow}(\overline{v}^{\text{r}}_{z'}) a_{\downarrow}(u_{y}) \xi_{\lambda} \rangle\nonumber\\
&&\equiv \text{I}_{a} + \text{I}_{b}\;,
\end{eqnarray}
where we set $\chi^{c} = 1 - \chi$. Let us consider $\text{I}_{\text{b}}$. Recall that $\hat \delta^{\text{r}}(p) = \chi(\rho^{\beta} |p|)$; therefore, a simple integration by parts argument shows that, for all $n\in \mathbb{N}$:
\begin{equation}\label{eq:deltardec}
| \delta^{\text{r}}(z;x) | \leq \frac{C_{n}\rho^{-3\beta}}{1 + (\rho^{-\beta} |z-x|_{L})^{n}}\;.
\end{equation}
We then have:
\begin{eqnarray}\label{eq:Ibest}
|\text{I}_{b}| &\leq& C\rho \int dxdydzdz'\, V(x-y) \varphi(z-z') \chi^{c}(|z -z'|_{L} / (8R_{0})) |\delta^{\text{r}}_{\uparrow}(z;x)| |\delta^{\text{r}}_{\downarrow}(z';y)| \| a_{\uparrow}(u_{x}) a_{\downarrow}(u_{y}) \xi_{\lambda} \| \nonumber\\
&\leq& C \rho^{1+ \beta (n-3)} \int dxdy\, V(x-y) \| a_{\uparrow}(u_{x}) a_{\downarrow}(u_{y}) \xi_{\lambda} \|\nonumber\\
&\leq& C_{n} \rho^{\beta (n-3)} ( CL^{3} \rho^{2} + \langle \xi_{\lambda}, \widetilde{\mathbb{Q}}_{1} \xi_{\lambda}\rangle )\;,
\end{eqnarray}
where the second inequality follows from (\ref{eq:deltardec}), and the last from Cauchy-Schwarz inequality. More precisely, to prove the first inequality we use that, for $|x- y|_{L} \leq R_{0}$ (a contraint imposed by the compact support of $V$ in $\Lambda_{L}$):
\begin{eqnarray}
&&\int dzdz'\, \chi^{c}(|z - z'|_{L} / (8R_{0}))  | \delta^{\text{r}}(z;x) | | \delta^{\text{r}}(z';y) | \nonumber\\
&&\qquad \leq \int dzdz'\, (\chi^{c}(|z-x|_{L} / R_{0}) + \chi^{c}(|z'-y|_{L} / R_{0}))  | \delta^{\text{r}}(z;x) | | \delta^{\text{r}}(z';y) | \nonumber\\
&&\qquad \leq C_{n} \rho^{\beta (n-3)}\;.
\end{eqnarray}
Consider now the term $\text{I}_{a}$. We claim that, for any two vectors $\xi, \psi \in \mathcal{F}$:
\begin{eqnarray}\label{eq:C5}
&&\int dzdz'\, \varphi(z-z') \chi(|z -z'|_{L} / (8R_{0})) \delta^{\text{r}}_{\uparrow}(z;x) \delta^{\text{r}}_{\downarrow}(z';y) \langle \xi_{\lambda}, a_{\uparrow}(\overline{v}^{\text{r}}_{z})a_{\downarrow}(\overline{v}^{\text{r}}_{z'}) \psi \rangle \nonumber\\
&&\qquad = \varphi(x-y) \chi(|x -y|_{L} / (8R_{0}))\langle \xi, a_{\uparrow}(\overline{v}^{\text{r}}_{x}) a_{\downarrow}(\overline{v}^{\text{r}}_{y}) \psi\rangle + \mathcal{E}_{x,y}(\xi, \psi)\;,
\end{eqnarray}
with, for all $n\geq 4$:
\begin{equation}\label{eq:Exyest}
| \mathcal{E}_{x,y}(\xi, \psi) | \leq C_{n} \rho^{1+ \beta (n-3)} \|\xi\| \|\psi\|\;,
\end{equation}
uniformly in $x,y$. To prove this, we proceed as follows. Let $m(z-z') = \varphi(z-z') \chi(|z -z'|_{L} / (8R_{0}))$. Being supported away from $|z-z'|_{L} = \rho^{-\gamma}$, this function is $C^{n}$ for any $n\in \mathbb{N}$, provided $(1 + |p|^{k})\hat V(p) \in L^{\infty}$ for $k$ large enough. This follows from the fact that $\varphi_{\infty}$ solves the scattering equation inside the ball of radius $\rho^{-\gamma}$; recall Lemma \ref{lem:scat}. Therefore, a simple integration by parts argument shows that:
\begin{equation}\label{eq:decf}
|\hat m(p)| \leq \frac{C_{n}}{1 + |p|^{n}}\qquad \text{for any $n\in \mathbb{N}$}\;,
\end{equation}
provided $V$ is regular enough. Next, we rewrite the approximate delta functions $\delta^{\text{r}}_{\sigma}$ as:
\begin{equation}\label{eq:deltasplit}
\delta^{\text{r}}_{\sigma}(z;x) = \delta_{\sigma}(z;x) - \delta_{\sigma}^{>}(z;x)
\end{equation}
where $\delta(\cdot)$ is the periodic Dirac delta over $\Lambda_{L}$, and $\hat \delta_{\sigma}^{>}(p) = 1 - \chi(\rho^{\beta} |p|)$. After performing the replacement in the left-hand side of (\ref{eq:C5}), we get:
\begin{equation}\label{eq:Iamain}
(\ref{eq:C5}) = m(x-y) \langle \xi, a_{\uparrow}(\overline{v}^{\text{r}}_{x}) a_{\downarrow}(\overline{v}^{\text{r}}_{y}) \psi\rangle + \mathcal{E}_{x,y}(\xi, \psi)
\end{equation}
where the error term $\mathcal{E}_{x,y}(\xi, \psi)$ collects terms with at least one $\delta^{>}$. Let us estimate it. Consider the term with two $\delta^{>}$:
\begin{eqnarray}\label{eq:ff}
&&\int dzdz'\, m(z-z') \delta_{\uparrow}^{>}(z;x) \delta_{\downarrow}^{>}(z';y) \langle \xi, a_{\uparrow}(\overline{v}^{\text{r}}_{z})  a_{\downarrow}(\overline{v}^{\text{r}}_{z'}) \psi \rangle \\
&& \qquad = \frac{1}{L^{3}} \sum_{q\in \frac{2\pi}{L} \mathbb{Z}^{3}} \hat m(q) \int dzdz'\, e^{iq\cdot z} e^{-iq\cdot z'}  \delta_{\uparrow}^{>}(z;x) \delta_{\downarrow}^{>}(z';y) \langle \xi, a_{\uparrow}(\overline{v}^{\text{r}}_{z}) a_{\downarrow}(\overline{v}^{\text{r}}_{z'}) \psi \rangle\;.\nonumber
\end{eqnarray}
We rewrite the innermost integral as:
\begin{equation}
\int dzdz'\,dr_{1} dr_{2} e^{iq\cdot z} e^{-iq\cdot z'} \delta_{\uparrow}^{>}(z;x) \delta_{\downarrow}^{>}(z';y) v_{\uparrow}^{\text{r}}(r_{1};z) v^{\text{r}}_{\downarrow}(r_{2};z') \langle \xi, a_{r_{1},\uparrow} a_{r_{2},\downarrow} \psi \rangle\;.
\end{equation}
Let us perform the $z$, $z'$ integrations. We have:
\begin{eqnarray}
\int dz\, e^{iq\cdot x} \delta_{\uparrow}^{>}(z;x) v_{\uparrow}^{\text{r}}(r_{1}; z) &=& \int dz\, e^{iq\cdot z} \int dp\, e^{ip\cdot (z-x)} \hat \delta_{\uparrow}^{>}(p) \int dp'\, e^{ip'\cdot (r_{1} + z)} \hat v^{\text{r}}_{\uparrow}(p')\nonumber\\
&=& \int dp\, e^{-ip\cdot (x+r_{1})} e^{-iq\cdot r_{1}} \hat \delta_{\uparrow}^{>}(p) \hat v^{\text{r}}_{\uparrow}(p+q) \nonumber\\
&=:& g_{x,q}^{\uparrow}(r_{1})\;.
\end{eqnarray}
Similarly,
\begin{equation}
\int dz\, e^{-iq\cdot y} \delta_{\downarrow}^{>}(z';y) v_{\downarrow}^{\text{r}}(r_{2}; z') =: g^{\downarrow}_{y,q}(r_{2})\;.
\end{equation}
These functions are in $L^{2}$; in fact,
\begin{equation}
\| g^{\sigma}_{x,q} \|_{2}^{2} \leq \int dp\, |\hat \delta_{\uparrow}^{>}(p) \hat v^{\text{r}}_{\uparrow}(p+q)|^{2} \leq C\rho\;.
\end{equation}
Also, notice that the $p$ integration in the definition of $g^{\sigma}_{x,q}$ is restricted to $|p+q|\leq k_{F}^{\sigma}$ (by the support properties of $\hat v^{\text{r}}(p+q)$), and $|p|\geq C\rho^{-\beta}$ (by the support properties of $\hat \delta^{>}(p)$). This implies that $g^{\sigma}_{x,q} = 0$ unless $|q| \geq c\rho^{-\beta}$; hence,
\begin{eqnarray}\label{eq:ff2}
\Big|\int dzdz'\, m(z-z') \delta_{\uparrow}^{>}(z-x) \delta_{\downarrow}^{>}(z'-y) \langle \xi, a_{\uparrow}(\overline{v}^{\text{r}}_{z}) a_{\downarrow}(\overline{v}^{\text{r}}_{z'})\psi \rangle\Big| &=& \Big| \frac{1}{L^{3}}\sum_{|q| \geq c\rho^{-\beta}} \hat m(q) \langle \xi, a_{\uparrow}(\overline{g_{x,q}}) a_{\downarrow}(\overline{g_{y,q}}) \psi\rangle\Big|\nonumber\\
&\leq& C\rho \int_{|q| \geq c\rho^{-\beta}} dq\, | \hat m(q) | \|\xi\| \|\psi\|\;.
\end{eqnarray}
Using the bound (\ref{eq:decf}), we get:
\begin{equation}
\Big|\int dzdz'\, m(z-z') \delta_{\uparrow}^{>}(z-x) \delta_{\downarrow}^{>}(z'-y) \langle \xi, a_{\uparrow}(\overline{v}^{\text{r}}_{z}) a_{\downarrow}(\overline{v}^{\text{r}}_{z'})\psi \rangle\Big| \leq C_{n} \rho^{1 + (n-3)\beta } \|\xi\| \|\psi\|\;,
\end{equation}
uniformly in $x$ and $y$. To conclude, consider the remaining error terms, of the form:
\begin{equation}
\int dzdz'\, m(z-z') \delta(z-x) \delta_{\downarrow}^{>}(z'-y) \langle \xi, a_{\uparrow}(\overline{v}^{\text{r}}_{z}) a_{\downarrow}(\overline{v}^{\text{r}}_{z'}) \psi\rangle 
= \int dz'\, m(x-z') \delta_{\downarrow}^{>}(z'-y)  \langle \xi, a_{\uparrow}(\overline{v}^{\text{r}}_{x}) a_{\downarrow}(\overline{v}^{\text{r}}_{z'}) \psi\rangle\;.
\end{equation}
Proceeding as before, we estimate this term as:
\begin{eqnarray}
\Big|\frac{1}{L^{3}}\sum_{q\in \frac{2\pi}{L} \mathbb{Z}^{3}} \hat m(q) e^{iq\cdot x} \langle \xi, a_{\uparrow}(\overline{v}^{\text{r}}_{x}) a_{\downarrow}(\overline{g_{y,q}}) \psi \rangle\Big|  &\equiv& \Big|\frac{1}{L^{3}}\sum_{|q| \geq c\rho^{-\beta}} \hat m(q) e^{iq\cdot x} \langle \xi, a_{\uparrow}(\overline{v}^{\text{r}}_{x}) a_{\downarrow}(\overline{g_{y,q}}) \psi \rangle\Big| \nonumber\\
&\leq& C\rho \int_{|q|\geq c\rho^{-\beta}} dq\, | \hat m(q) | \|\xi\| \|\psi\|\nonumber\\
&\leq& C_{n} \rho^{1 + \beta (n - 3)} \|\xi\| \|\psi\|\;.
\end{eqnarray}
This concludes the proof of (\ref{eq:C5}), (\ref{eq:Exyest}). Let us now plug (\ref{eq:C5}) into $\text{I}_{a}$. We have, using that by the support properties of $V$, $V(x-y)\chi(|x -y|_{L} / (8R_{0})) = V(x-y)$:
\begin{eqnarray}\label{eq:Iaerr}
\text{I}_{a} &=& \int dxdy\, V(x-y) \varphi(x-y)\langle \xi_{\lambda}, a_{\uparrow}(\overline{v}^{\text{r}}_{x}) a_{\uparrow}(u_{x}) a_{\downarrow}(\overline{v}^{\text{r}}_{y}) a_{\downarrow}(u_{y}) \xi_{\lambda} \rangle\nonumber\\
&& + \tilde{\text{I}}_{a}\;,
\end{eqnarray}
where $\tilde{\text{I}}_{a}$ is bounded as:
\begin{eqnarray}\label{eq:tildeIaest}
|\tilde{\text{I}}_{a}| &\leq& C_{n} \rho^{1+ \beta (n-3)} \int dxdy\, V(x-y) \| a_{\uparrow}(u_{x}) a_{\downarrow}(u_{y}) \xi_{\lambda}\| \nonumber\\
&\leq& C_{n} \rho^{\beta (n-3)} ( CL^{3} \rho^{2} + \langle \xi_{\lambda}, \widetilde{\mathbb{Q}}_{1} \xi_{\lambda} \rangle )\;.
\end{eqnarray}
Putting together (\ref{eq:Iapp}), (\ref{eq:Ibest}), (\ref{eq:tildeIaest}), we get:
\begin{eqnarray}
\text{I} &=& \int dxdy\, V(x-y) \varphi(x-y)\langle \xi_{\lambda}, a_{\uparrow}(\overline{v}^{\text{r}}_{x}) a_{\uparrow}(u_{x}) a_{\downarrow}(\overline{v}^{\text{r}}_{y}) a_{\downarrow}(u_{y}) \xi_{\lambda} \rangle + \widehat{\mathcal{E}}_{\widetilde{\mathbb{Q}}_{1}}(\xi_{\lambda}) \nonumber\\
| \widehat{\mathcal{E}}_{\widetilde{\mathbb{Q}}_{1}}(\xi_{\lambda}) | &\leq& C_{n} \rho^{\beta (n - 3)} ( CL^{3} \rho^{2} + \langle \xi_{\lambda}, \widetilde{\mathbb{Q}}_{1} \xi_{\lambda}\rangle )\;,
\end{eqnarray}
which concludes the of Lemma \ref{lem:UV2}.\qed

\subsection{Regularization of $\mathbb{T}_{2}$ and of $\widetilde{\mathbb{Q}}_{4}$}\label{app:ab}
\noindent{\bf Regularization of $\mathbb{T}_{2}$.} Let us start by discussing the regularization of $\mathbb{T}_{2}$, recall Eq. (\ref{eq:Tdef}). We write
\begin{equation}\label{eq:splitu}
\hat u(k) = \hat u^{\text{r}}(k) + \hat \alpha(k) + \hat \delta^{>}(k)\;,
\end{equation}
with $\hat \alpha(k)$ supported for $k_{F} \leq |k| \leq 2k_{F}$ and $\hat\delta^{>}(k)$ supported for $|k| \geq \rho^{-\beta}$. Let $\mathbb{T}_{2}^{\text{r}}$ be the operator obtained from $\mathbb{T}_{2}$ replacing $u$ by $u^{\text{r}}$. We write:
\begin{equation}
\mathbb{T}_{2} - \mathbb{T}^{\text{r}}_{2} = \mathbb{T}_{2;a} + \mathbb{T}_{2;b}\;,
\end{equation}
where $\mathbb{T}_{2;b}$ contains at least one operator $a(\delta^{>}_{x})$, while $\mathbb{T}_{2;a}$ contains at least one operator $a(\alpha_{x})$, and no operator $a(\delta^{>}_{x})$.
\medskip

\noindent{\it \underline{Bound for $\mathbb{T}_{2,\beta}$.}} Let $m(x-y) = V(x-y)\varphi(x-y)$. We claim that, omitting the spins for simplicity, for any $\xi \in \mathcal{F}$ and for $L$ large enough:
\begin{equation}\label{eq:XY}
\Big\| \int dy\, a(\delta^{>}_{y}) m(x-y) a(\overline{v}^{\text{r}}_{y}) \xi \Big\| \leq C_{n} \rho^{\beta n} \|\xi\|\;,\qquad \text{for any $n\in \mathbb{N}$ large enough}\;,
\end{equation}
provided $(1 + |p|^{k}) \hat V \in L^{\infty}$, for $k$ large enough. This bound allows to prove that $\langle \xi_{\lambda}, \mathbb{T}_{2;b} \xi_{\lambda}\rangle$ is small. For instance, consider (we omit the spin for simplicity):
\begin{equation}
\text{I} = \int dxdy\, m(x-y) \langle \xi_{\lambda}, a(\delta^{>}_{x}) a(\delta^{>}_{y}) a(\overline{v}_{x}) a(\overline{v}_{y})  \xi_{\lambda}\rangle\;.
\end{equation}
We have, from (\ref{eq:XY}):
\begin{eqnarray}
|\text{I}| &\leq& \int dy\, \Big\|  \Big(\int dx\, m(x-y) a(\delta^{>}_{x})  a(\overline{v}_{x})\Big)  a(\delta^{>}_{y}) a(\overline{v}_{y}) \xi_{\lambda}  \Big\|\nonumber\\
&\leq& C_{n}\rho^{\beta n} \int dy\, \| a(\delta^{>}_{y}) a(\overline{v}_{y}) \xi_{\lambda} \| \nonumber\\
&\leq& C_{n} \rho^{\beta n + \frac{1}{2}} L^{\frac{3}{2}} \| \mathcal{N}^{\frac{1}{2}} \xi_{\lambda} \|\;,
\end{eqnarray}
where in the last step we used $\|a(\overline{v}_{y}) \|\leq C\rho^{\frac{1}{2}}$ and Cauchy-Schwarz inequality. The contribution to $\mathbb{T}_{2;b}$ corresponding to the operators $a(\delta^{>}_{x})$, $a(\alpha_{y})$ can be estimated in exactly the same way, we omit the details. Using the propagation of the a priori estimate for the number operator, $\| \mathcal{N}^{\frac{1}{2}} \xi_{\lambda} \|\leq CL^{\frac{3}{2}} \rho^{\frac{7}{12}}$, we get:
\begin{equation}\label{eq:XY2}
|\langle \xi_{\lambda}, \mathbb{T}_{2;b} \xi_{\lambda}\rangle|\leq C_{n}L^{3} \rho^{\beta n + \frac{13}{12}}\;. 
\end{equation}
Let us prove the bound (\ref{eq:XY}). The statement is trivially true if $m$ is replaced by a constant, by the orthogonality of $\delta^{>}_{y}$ and of $v^{\text{r}}_{y}$. For nonconstant $m$, we proceed as follows.

We consider an operator $w$ with integral kernel $w(x;y) \equiv w(x-y)$, such that $\hat w(k) = 1$ for $|k| \leq \rho^{1/3}$ and $\hat w(k) = 0$ for $|k| > 2\rho^{1/3}$, and it smoothly interpolates between $1$ and $0$ for $\rho^{1/3} \leq |k| \leq 2 \rho^{1/3}$. Since $\hat v^{\text{r}}(k)$ is supported for $|k| \leq \rho^{\frac{1}{3}}$, we have $v^{\text{r}} = v^{\text{r}} w$. Hence:
\begin{equation}
a(\overline{v}^{\text{r}}_{y}) = \int dz\, a(\overline{v}^{\text{r}}_{z}) w(y;z)\;.
\end{equation}
Therefore,
\begin{eqnarray}
\int dy\, a^{*}(\delta^{>}_{y}) m(x-y) a(\overline{v}^{\text{r}}_{y}) &=& \int dydz\, a^{*}(\delta^{>}_{y}) m(x-y) w(y;z) a(\overline{v}^{\text{r}}_{z})\nonumber\\
&\equiv& \int dz\, a^{*}(A_{x,z}) a(\overline{v}^{\text{r}}_{z})
\end{eqnarray}
with:
\begin{equation}
A_{x,z}(r) = \int dy\, \delta^{>}(r;y) m(x-y) w(y;z)\;.
\end{equation}
We will need estimates on the decay properties of this function. For $L$ large enough, integration by parts gives:
\begin{eqnarray}\label{eq:Abd}
| (x-z)^{m_{1}} (x-r)^{m_{2}} A_{x,z}(r) | \leq \int dkdq\, | \partial_{k}^{m_{2}} \partial_{q}^{m_{1}} \hat m(k-q) \hat \delta^{>} (k) \hat w(q)  |\;.\nonumber
\end{eqnarray}
Using that, for any $n, m\in \mathbb{N}$:
\begin{equation}
| \partial^{m}_{k} \hat m(k) | \leq \frac{C_{n+m}}{1 + |k|^{n+m}}\;,\quad | \partial^{m}_{k}\hat w(k) | \leq C_{m}\rho^{-\frac{m}{3}} \chi(k\in \text{supp}\, \hat w)\;,\quad | \partial^{m}_{k} \hat \delta^{>}(k) | \leq C_{m} \rho^{\beta m} \chi(k\in \text{supp}\, \hat \delta^{>})
\end{equation}
we get:
\begin{eqnarray}
| (x-z)^{m_{1}} (x-r)^{m_{2}} A_{x,z}(r) | &\leq& C_{n, m_{1}, m_{2}} \int^{*} dk dq\, \frac{\rho^{-\frac{1}{3} (m_{1} + m_{2})}}{1 + |k-q|^{n}} \nonumber\\
&\leq& C_{n, m_{1}, m_{2}} \rho^{-\frac{1}{3}(m_{1} + m_{2} - 3)}  \rho^{\beta(n - 3)}\;,
\end{eqnarray}
where the asterisk denotes the constraints $q\in \text{supp}\, \hat w$, $k-q \in \text{supp}\, \hat \delta^{>}$. This bound implies:
\begin{equation}\label{eq:Axz}
|A_{x,z}(r)| \leq \frac{C_{n, m_{1}, m_{2}} \rho^{\beta(n - 3)}}{1 + (\rho^{\frac{1}{3}} |x-z|)^{m_{1}}} \frac{1}{1 + (\rho^{\frac{1}{3}} |x-r|)^{m_{2}}}\;.
\end{equation}
Therefore,
\begin{eqnarray}
\Big\|  \int dy\, a^{*}(\delta^{>}_{y}) m(x-y) a(\overline{v}^{\text{r}}_{y}) \varphi \Big\| &\equiv& \Big\|  \int dz\, a^{*}(A_{x,z}) a(\overline{v}^{\text{r}}_{z}) \varphi \Big\| \\
&\leq& C\rho^{\frac{1}{2}} \int dz\, \| A_{x,z} \|_{2} \| \|\varphi\|_{2}\;.\nonumber
\end{eqnarray}
Eq. (\ref{eq:Axz}) implies that, for all $n\in \mathbb{N}$:
\begin{equation}
 \int dz\, \| A_{x,z} \|_{2} \leq C_{n, m_{1}, m_{2}}\int dz\, \frac{\rho^{\beta(n - 3) - \frac{1}{2}}}{1 + (\rho^{\frac{1}{3}} |x-z|)^{m_{1}}} \leq C_{n, m_{1}, m_{2}} \rho^{\beta(n - 3) - 1}\;.
\end{equation}
Taking $n$ large enough, the claim (\ref{eq:XY}) follows.

\medskip

\noindent{\it \underline{Bound for $\mathbb{T}_{2;a}$.}} Consider:
\begin{equation}
\text{I} = \int dxdy\, V(x-y) \varphi(x-y) \langle \xi_{\lambda}, a_{\uparrow}(\overline{v}^{\text{r}}_{x}) a_{\uparrow}(\alpha_{x}) a_{\downarrow}(\overline{v}^{\text{r}}_{y}) a_{\downarrow}(u^{\text{r}}_{y})\xi_{\lambda} \rangle\;.
\end{equation}
The term corresponding to $a(\alpha_{x})$, $a(\alpha_{y})$ is estimated in exactly the same way. We have:
\begin{eqnarray}\label{eq:T2alpha}
|\text{I}| &\leq& C\int dxdy\, V(x-y) \| \overline{v}_{x}^{r} \|_{2} \| \alpha_{x} \|_{2} \| \overline{v}^{\text{r}}_{y} \|_{2} \| a_{\downarrow}(u^{\text{r}}) \xi_{\lambda} \|\nonumber\\
&\leq& CL^{\frac{3}{2}} \rho^{\frac{3}{2}} \| \mathcal{N}^{\frac{1}{2}} \xi_{\lambda} \| \leq CL^{\frac{3}{2}}\rho^{\frac{3}{2}} \| \mathcal{N}^{\frac{1}{2}} \xi_{\lambda} \|\;.
\end{eqnarray}
All the other contributions to $\mathbb{T}_{2;a}$ are bounded in the same way. 
\medskip

\noindent{\it \underline{Conclusion.}} Putting together (\ref{eq:XY2}), (\ref{eq:T2alpha}) we have, taking $n$ large enough in (\ref{eq:XY2}):
\begin{equation}
|\langle \xi_{\lambda}, (\mathbb{T}_{2} - \mathbb{T}_{2}^{\text{r}}) \xi_{\lambda} \rangle| \leq CL^{\frac{3}{2}}\rho^{\frac{3}{2}} \| \mathcal{N}^{\frac{1}{2}} \xi_{\lambda} \|\;.
\end{equation}

\noindent{\bf Regularization of $\widetilde{\mathbb{Q}}_{4}$.} We start by writing 
\begin{equation}\label{eq:veta}
\hat v(k) = \hat v^{\text{r}}(k) + \hat\eta(k)\;
\end{equation}
with $\hat \eta (k)$ supported for $k_{F} - \rho^{\alpha} \leq |k| \leq k_{F}$ and $\alpha = \frac{1}{3} + \frac{\epsilon}{3}$, recall the definition of the correlation structure given in Section \ref{sec:cor}. Let $\widetilde{\mathbb{Q}}^{\text{r}}_{4;1}$ be the operator obtained from $\widetilde{\mathbb{Q}}_{4}$ after replacing all $v$ by $v^{\text{r}}$: 
\begin{equation}
\widetilde{\mathbb{Q}}^{\text{r}}_{4;1} = \int dxdy\, V(x-y) a^{*}_{\uparrow}(u_{x}) a^{*}_{\downarrow}(u_{y}) a^{*}_{\downarrow}(\overline{v}^{\text{r}}_{y}) a^{*}_{\uparrow}(\overline{v}^{\text{r}}_{x})
\end{equation}
Also, recall that:
\begin{equation}
\widetilde{\mathbb{Q}}^{\text{r}}_{4} = \int dxdy\, V(x-y) a^{*}_{\uparrow}(u^{\text{r}}_{x}) a^{*}_{\downarrow}(u^{\text{r}}_{y}) a^{*}_{\downarrow}(\overline{v}^{\text{r}}_{y}) a^{*}_{\uparrow}(\overline{v}^{\text{r}}_{x})\;.
\end{equation}
We set:
\begin{equation}
\widetilde{\mathbb{Q}}_{4} - \widetilde{\mathbb{Q}}^{\text{r}}_{4;1} = \widetilde{\mathbb{Q}}_{4;a}\;.
\end{equation}
\noindent{\it \underline{Bound for $\widetilde{\mathbb{Q}}_{4;a}$.}} Recall (\ref{eq:veta}). Consider the term:
\begin{equation}
\text{I} = \sum_{\sigma\neq \sigma'} \int dxdy\, V(x-y) \langle \xi_{\lambda}, a_{\sigma}(u_{x}) a_{\sigma}(\overline{\eta}_{x}) a_{\sigma'}(u_{y}) a_{\sigma'}(\overline{v}^{\text{r}}_{y}) \xi_{\lambda}\rangle\;.
\end{equation}
Then:
\begin{eqnarray}
|\text{I}| &\leq& \sum_{\sigma\neq \sigma'} \int dxdy\, V(x-y) \Big(\frac{\delta}{2} \| a_{\sigma}(u_{x}) a_{\sigma'}(u_{y}) \xi_{\lambda} \|^{2} + \frac{1}{\delta} \| \overline{\eta}_{x}\|^{2}_{2} \| \overline{v}_{y} \|_{2}^{2}\Big)\nonumber\\
&\leq& \delta \langle \xi_{\lambda}, \widetilde{\mathbb{Q}}_{1} \xi_{\lambda} \rangle + \frac{C}{\delta}L^{3} \rho^{2 + \frac{\epsilon}{3}}\;,
\end{eqnarray}
where we used that $\| \overline{\eta}_{x} \|_{2}^{2} \leq C\rho^{\frac{2}{3} + \alpha}$ and $\alpha = \frac{1}{3} + \frac{\epsilon}{3}$. All the other contributions to $\widetilde{\mathbb{Q}}_{4;a}$ can be estimated in the same way. Hence:
\begin{equation}\label{eq:Q4a}
|\langle \xi_{\lambda}, \widetilde{\mathbb{Q}}_{4;a} \xi_{\lambda}\rangle| \leq C\delta \langle \xi_{\lambda}, \widetilde{\mathbb{Q}}_{1} \xi_{\lambda} \rangle + \frac{C}{\delta}L^{3} \rho^{2 + \frac{\epsilon}{3}}\;.
\end{equation}
\noindent{\it \underline{Conclusion.}} We write:
\begin{equation}
\langle \xi_{\lambda}, (\widetilde{\mathbb{Q}}_{4} - \widetilde{\mathbb{Q}}_{4}^{\text{r}}) \xi_{\lambda} \rangle = \langle \xi_{\lambda}, (\widetilde{\mathbb{Q}}_{4} - \widetilde{\mathbb{Q}}_{4;1}^{\text{r}}) \xi_{\lambda} \rangle + \langle \xi_{\lambda}, (\widetilde{\mathbb{Q}}^{\text{r}}_{4;1} - \widetilde{\mathbb{Q}}_{4}^{\text{r}}) \xi_{\lambda} \rangle\;.
\end{equation}
The first term is bounded as in (\ref{eq:Q4a}), while the second can be bounded exactly as $\langle \xi_{\lambda}, (\mathbb{T}_{2} - \mathbb{T}_{2}^{\text{r}})\xi_{\lambda}\rangle$. We get:
\begin{eqnarray}
| \langle \xi_{\lambda}, (\widetilde{\mathbb{Q}}_{4} - \widetilde{\mathbb{Q}}_{4}^{\text{r}}) \xi_{\lambda}\rangle | &\leq& | \langle \xi_{\lambda}, (\widetilde{\mathbb{Q}}_{4} - \widetilde{\mathbb{Q}}^{\text{r}}_{4;1}) \xi_{\lambda}\rangle | + | \langle \xi_{\lambda}, (\widetilde{\mathbb{Q}}^{\text{r}}_{4;1} - \widetilde{\mathbb{Q}}_{4}^{\text{r}}) \xi_{\lambda}\rangle |\nonumber\\
&\leq& C\delta \langle \xi_{\lambda}, \widetilde{\mathbb{Q}}_{1} \xi_{\lambda} \rangle + \frac{C}{\delta}L^{3} \rho^{2 + \frac{\epsilon}{3}} + CL^{\frac{3}{2}}\rho^{\frac{3}{2}} \| \mathcal{N}^{\frac{1}{2}} \xi_{\lambda} \|\;.
\end{eqnarray}
\subsection{Proof of Eq. (\ref{eq:estI2})}\label{app:6uv}
We write:
\begin{equation}
\text{I}_{2} = \text{I}_{2;a} + \text{I}_{2;b}\;,
\end{equation}
where $\text{I}_{2;b}$ is obtained from $\text{I}$ replacing at least one between $a^{*}(u_{x}), a^{*}(u_{y})$ with either $a^{*}(\delta^{>}_{x})$ or $a^{*}(\delta^{>}_{y})$, recall Eq. (\ref{eq:splitu}). The contribution of this term can be proven to be smaller than any power of $\rho^{\beta}$, proceeding as for $\mathbb{T}_{2,\beta}$, and we shall omit the details. Consider now $\text{I}_{2;a}$. One term contributing to $\text{I}_{2;a}$ is:
\begin{equation}\label{eq:Q4I03}
\sum_{\sigma\neq \sigma'}\int dxdydzdz'\, V(x-y) \varphi(z-z') \tilde \omega^{\text{r}}(z;y)  \langle \xi_{\lambda}, a^{*}_{\sigma}(\alpha_{x}) a^{*}_{\sigma'}(u_{y}) a^{*}_{\sigma}(\overline{v}^{\text{r}}_{x}) a_{\downarrow}(\overline{v}^{\text{r}}_{z'}) a_{\uparrow}(u^{\text{r}}_{z}) a_{\downarrow}(u^{\text{r}}_{z'}) \xi_{\lambda}\rangle\;.
\end{equation}
To estimate (\ref{eq:Q4I03}), we proceed exactly as for $\text{I}_{1}$, recall Eqs. (\ref{eq:Q4I})-(\ref{eq:estI1}) The only difference is that now the estimate:
\begin{eqnarray}\label{eq:inst}
\int dxdy\, V(x-y) \| a_{\sigma}(u_{x}) a_{\sigma'}(u_{y}) \xi_{\lambda} \|^{2} &\leq& \langle \xi_{\lambda}, \widetilde{\mathbb{Q}}_{1} \xi_{\lambda}\rangle\nonumber\\
&\leq& CL^{3} \rho^{2}\;,
\end{eqnarray}
is replaced by:
\begin{eqnarray}
\int dxdy\, V(x-y) \| a_{\sigma}(\alpha_{x}) a_{\sigma'}(u_{y}) \xi_{\lambda} \|^{2} &\leq& C\rho \langle \xi_{\lambda}, \mathcal{N} \xi_{\lambda}\rangle \nonumber\\
&\leq& C L^{3}\rho^{\frac{13}{6}}\;,
\end{eqnarray}
which is better than (\ref{eq:inst}). In conclusion, we can estimate $\text{I}_{2;a}$ (in a nonoptimal way) using the same bound we obtained for $\text{I}_{1}$, Eq. (\ref{eq:I1est00}).

\subsection{Proof of Eq. (\ref{eq:Imainreg})}\label{sec:UV4}
Here we show how to go from (\ref{eq:prima}) to (\ref{eq:Imainreg}). We rewrite $\hat u_{\sigma}^{\text{r}}(k) = \chi(\rho_{\sigma}^{\beta} |k|) - \hat\nu_{\sigma}(k)$, with $\hat \nu_{\sigma}(k)$ smooth and such that $\hat \nu_{\sigma}(k) = 1$ for $|k|\leq k_{F}^{\sigma}$, $\hat \nu_{\sigma}(k) = 0$ for $|k| > 2k_{F}^{\sigma}$. Therefore, for all $n\in \mathbb{N}$, its inverse Fourier transform $\nu_{\sigma}(x-y)$ decays as:
\begin{equation}
| \nu_{\sigma}(x-y) | \leq \frac{C_{n} \rho}{1 + \rho^{\frac{n}{3}} |x-y|^{n}}\;.
\end{equation}
Furthermore, let $\delta_{\sigma}^{\text{r}}(x-y)$ be the inverse Fourier transform of $\chi(\rho_{\sigma}^{\beta} |k|)$. We write $\delta_{\sigma}^{\text{r}}(x-y) = \delta(x-y) - \delta_{\sigma}^{>}(x-y)$, with $\delta(x-y)$ the periodic Dirac delta and $\hat \delta^{>}_{\sigma}(x-y)$ supported for $|k| \geq 2\rho^{-\beta}$. All together:
\begin{equation}
u^{\text{r}}_{\sigma}(x;y) = \delta(x-y) - \delta_{\sigma}^{>}(x-y) - \nu_{\sigma}(x-y)\;.
\end{equation}
We then get:
\begin{equation}
\int dxdydzdz'\, V(x-y) \varphi(z-z') u^{\text{r}}_{\uparrow} (z;x) u^{\text{r}}_{\downarrow}(z';y) \omega^{\text{r}}_{\uparrow}(z;x)\omega_{\downarrow}^{\text{r}}(z';y) = \rho^{\text{r}}_{\uparrow} \rho^{\text{r}}_{\downarrow} \int dxdy\, V(x-y) \varphi(x-y) + \mathcal{E}_{a} + \mathcal{E}_{b}
\end{equation}
where $\rho^{\text{r}}_{\sigma} = \omega_{\sigma}^{\text{r}}(x;x)$. The term $\mathcal{E}_{a}$ collects terms with no $\delta^{>}$ function and at least one $\nu$ function, while the term $\mathcal{E}_{b}$ collects terms with at least one $\delta^{>}$ function. Consider $\mathcal{E}_{a}$, and let us start from the terms with only one $\nu$ function:
\begin{eqnarray}\label{eq:1nu}
&&\int dxdydzdz'\, V(x-y) \varphi(z-z') \nu_{\uparrow}(z;x) \delta(z'-y) \omega^{\text{r}}_{\uparrow}(z;x) \omega^{\text{r}}_{\downarrow}(z';y) \nonumber\\
&&\qquad = \rho_{\downarrow}^{\text{r}}\int dxdydz\, V(x-y) \varphi(z-y) \nu_{\uparrow}(z;x) \omega^{\text{r}}_{\uparrow}(z;x) \;.
\end{eqnarray}
We have, using $\|\nu_{x}\|_{\infty} \leq C\rho$, $\|\omega^{\text{r}}_{x}\|_{\infty} \leq C\rho$:
\begin{equation}\label{eq: first est prop 5.25}
\left| \rho_{\downarrow}^{\text{r}}\int dxdydz\, V(x-y) \varphi(z-y) \nu_{\uparrow}(z;x) \omega^{\text{r}}_{\uparrow}(z;x) \right| \leq C L^{3}\rho^{3} \|V\|_{1} \|\varphi\|_{1} \leq CL^{3} \rho^{3-2\gamma}\;.
\end{equation}
Consider now the term with two $\nu$ functions. We get, using that $\| \nu_{\sigma} \|_{1} \leq C$:
\begin{equation}
\Big|\int dxdydzdz'\, V(x-y) \varphi(z-z') \nu_{\uparrow}(z;x) \nu_{\downarrow}(z';y) \omega^{\text{r}}_{\uparrow}(z;x) \omega^{\text{r}}_{\downarrow}(z';y)\Big| \leq CL^{3} \rho^{3 - 2\gamma}\;.
\end{equation}
Hence:
\begin{equation}\label{eq:bdEa}
|\mathcal{E}_{a}| \leq CL^{3} \rho^{3 - 2\gamma}\;.
\end{equation}
Let us now consider $\mathcal{E}_{b}$. Let us omit the spin for simplicity. To simplify the notation, in the following we shall set $\int dk\,(\cdots) = L^{-3} \sum_{k}(\cdots)$. We start from the term:
\begin{equation}
\text{I} = \int dxdydzdz'\, V(x-y) \varphi(z-z') \delta^{>}(z;x) \delta^{>}(z';y) \omega^{\text{r}}(z;x)\omega^{\text{r}}(z';y)\;.
\end{equation}
We rewrite it in momentum space as:
\begin{equation}
\text{I} = L^{3}\int dk_{1} dk_{2} dk_{3}\, \hat V(k_{1} + k_{3}) \hat \varphi(-k_{1} - k_{3}) \hat \delta^{>}(k_{1}) \hat \delta^{>}(k_{2}) \hat \omega^{\text{r}}(k_{3}) \hat \omega^{\text{r}}(-k_{1} -k_{2} - k_{3})\;,
\end{equation}
which we estimate as, using that $|k_{1} + k_{3}| \geq C\rho^{-\beta}$ by the support properties of $\omega^{\text{r}}(k_{3})$ and of $\hat \delta^{>}(k_{1})$:
\begin{equation}
|\text{I}| \leq L^{3}C_{n} \rho^{\beta n} \int dk_{1} dk_{2} dk_{3}\, \chi(|k_{1}+ k_{3}| \geq C\rho^{-\beta})| \hat \varphi(k_{1} + k_{3}) | \hat \omega^{\text{r}}(k_{3}) \hat \omega^{\text{r}}(-k_{1} -k_{2} - k_{3})\;.
\end{equation}
To prove this estimate we used that $|\hat V(k)| \leq C_{n} (1 + |k|^{n})^{-1}$. Also, by the decay properties of $\hat \varphi$, Eq. (\ref{eq:derr1}):
\begin{equation}
\int dk\, \chi(|k| \geq C\rho^{-\beta}) | \hat \varphi(k) | \leq C\;.
\end{equation}
Finally, using also that $\int dk\, \hat \omega^{\text{r}}(k) \leq C\rho$, we have:
\begin{equation}
|\text{I}|\leq L^{3} C_{n} \rho^{\beta n + 2}\;. 
\end{equation}
Consider now the term:
\begin{eqnarray}
\text{II} &=& \int dxdydzdz'\, V(x-y) \varphi(z-z') \delta^{>}(z;x) \nu(z';y) \omega^{\text{r}}(z;x)\omega^{\text{r}}(z';y) \nonumber\\
&=& L^{3} \int dk_{1} dk_{2} dk_{3}\, \hat V(k_{1} + k_{3}) \hat \varphi(-k_{1} - k_{3}) \hat \delta^{>}(k_{1}) \hat \nu(k_{2}) \hat \omega^{\text{r}}(k_{3}) \hat \omega^{\text{r}}(-k_{1} -k_{2} - k_{3})\;.
\end{eqnarray}
Using that $|\hat \nu(k_{2})| \leq 1$, this term can be estimated exactly as $\text{I}$. In conclusion, for any $n\in \mathbb{N}$, taking $V$ regular enough:
\begin{equation}\label{eq:bdEb}
|\mathcal{E}_{b}| \leq C_{n}L^{3} \rho^{2 + n\beta}\;. 
\end{equation}
Putting together (\ref{eq:bdEa}), (\ref{eq:bdEb}) we have:
\begin{equation}
|\mathcal{E}_{a}| + |\mathcal{E}_{b}|\leq CL^{3} \rho^{3 - 2\gamma}\;.
\end{equation}
This concludes the proof of Eq. (\ref{eq:Imainreg}).\qed

\end{document}